\documentclass[10pt]{article}
\usepackage[T1]{fontenc}
\usepackage[utf8]{inputenc}
\usepackage{amsmath}											
\usepackage{amsfonts}
\usepackage{amssymb}
\usepackage{amsthm}
\usepackage{bbm}
\usepackage[normalem]{ulem}
\usepackage{caption}
\usepackage{subcaption}
\usepackage{comment}
\usepackage{esint}
\usepackage[margin=1.0in]{geometry}
\usepackage{graphicx}
\usepackage{mathtools}
\usepackage[final]{showlabels}
\usepackage{xcolor}
\graphicspath{ {images/} }
\usepackage[section]{placeins}
\DeclareMathOperator{\diag}{diag}

\DeclareMathOperator{\sgn}{sgn}
\DeclareMathOperator{\supp}{supp}

\newcommand{\norm}[1]{\left\lVert#1\right\rVert}
\newtheorem{proposition}{Proposition}
\newtheorem{lemma}[proposition]{Lemma}
\newtheorem{theorem}[proposition]{Theorem}
\newtheorem{remark}[proposition]{Remark}
\newtheorem{assumption}[proposition]{Assumption}
\newtheorem{definition}[proposition]{Definition}

\numberwithin{proposition}{section}

\usepackage{tikz}
\usetikzlibrary{backgrounds}

\newcommand{\gb}[1]{{{\color{blue}{GB: #1}}}}

\newcommand{\tm}[1]{{{\color{black}{#1}}}}
\newcommand{\sq}[1]{{{\color{brown}{SQ: #1}}}}
\newcommand{\tbr}[1]{{{\color{black}{#1}}}}

\newcommand{\eps}{\varepsilon}
\newcommand{\Rm}{\mathbb{R}}
\newcommand{\Cm}{\mathbb{C}}
\newcommand{\Zm}{\mathbb{Z}}

\newcommand{\dint}{\displaystyle\int}

\newcommand{\aver}[1]{\langle #1 \rangle}
\newcommand{\sign}{\text{sign}}

\newcommand{\fm}{\mathfrak m}

\newcommand{\fs}{\mathfrak{S}}

\newcommand{\ow}{{\rm Op}^w}
\newcommand{\bdi}{{\rm {\bf BDI}}}

\newcommand{\res}{r}

\newcommand{\tM}{\tilde{M}}

\newcommand{\la}{\mathfrak{a}}
\newcommand{\lb}{\mathfrak{b}}

\newcommand{\cR}{\Lambda}
\newcommand{\cS}{\mathcal{S}}
\newcommand{\os}{\mathfrak{s}}

\newcommand{\bH}{{\rm {\bf H}}}

\newcommand{\br}{{\bf x}}

\newcommand{\hathperp}{\hat{h}_{\delta}} 
\newcommand{\hathperpnot}{\hat{h}_{\delta,0}}

\newcommand{\hperp}{h_\delta} 
\newcommand{\Hperp}{H^{12}}

\newcommand{\NN}{\mathfrak{N}^1}

\newcommand{\rot}{R}
\newcommand{\rotq}{Q}

\newcommand{\dpx}{\mathcal{K}}
\newcommand{\fri}{\mathfrak{I}}
\newcommand{\lv}{v}
\newcommand{\vp}{A}

\newcommand{\es}{\mathcal{T}}
\newcommand{\phase}{\Phi}
\newcommand{\rate}{\mu} 
\newcommand{\cs}{\check{S}}
\newcommand{\sym}{\sigma}

\newcommand{\mbp}{{\rm M}_{\eta,\rate}}
\newcommand{\nbp}{{\rm N}_{\eta,\rate}}

\newcommand{\ncv}{N_{{\rm CV}}}

\newcommand{\atrunc}{\overset{\circ}{a}_{12}}
\newcommand{\adtrunc}{\overset{\circ}{a}_{\delta,12}}
\newcommand{\ares}{\check{a}_{12}}
\newcommand{\adres}{\check{a}_{\delta,12}}
\newcommand{\atruncall}{\overset{\circ}{a}}
\newcommand{\aresall}{\check{a}}
\newcommand{\Htrunc}{\overset{\circ}{H}}
\newcommand{\Hres}{\check{H}}

\newcommand{\tbgparam}{\rho}

\newcommand{\tpn}{\eps} 

\newcommand{\latticeIn}{\Lambda_{\rm in}^{*,\eps}}
\newcommand{\latticeOut}{\Lambda_{\rm out}^{*,\eps}}

\newcommand{\iin}{I_{{\rm in}}}
\newcommand{\iout}{I_{{\rm out}}}

\newcommand{\fh}{\mathfrak{h}}
\newcommand{\ff}{\mathfrak{A}}

\newcommand{\lattice}{\Lambda}

\newcommand{\hrem}{\hathperp^{{\rm rem}}}
\newcommand{\odl}{g_\delta}

\title{Macroscopic approximation of tight-binding models near spectral degeneracies and validity for wave packet propagation}
\author{Guillaume Bal \thanks{Departments of Statistics and Mathematics and Committee on Computational and Applied Mathematics, University of Chicago, Chicago, IL 60637; guillaumebal@uchicago.edu} \and Paul Cazeaux\thanks{Department of Mathematics, Virginia Tech, Blacksburg, VA 24060; cazeaux@vt.edu} \and Daniel Massatt \thanks{Department of Mathematical Sciences, New Jersey Institute of Technology, Newark, NJ 07103; daniel.massatt@njit.edu} \and Solomon Quinn \thanks{Center for Computational Mathematics, Flatiron Institute, New York, NY 10010; squinn@flatironinstitute.org}}

\begin{document}

\maketitle

\begin{abstract}
This paper concerns the derivation and validity of macroscopic descriptions of wave packets supported in the vicinity of degenerate points $(K,E)$ in the dispersion relation of tight-binding models accounting for macroscopic variations. We show that such wave packets are well approximated over long times by macroscopic models with varying orders of accuracy. Our main applications are in the analysis of single- and multilayer graphene tight-binding Hamiltonians modeling macroscopic variations such as those generated by shear or twist. Numerical simulations illustrate the theoretical findings.
\end{abstract}


\noindent {\bf Keywords:} Tight-binding, Haldane model, Dirac equation, Twisted bilayer graphene, semiclassical analysis

\section{Introduction}

Tight-binding (discrete) models are ubiquitous in the analysis of transport properties of electronic and photonic structures \cite{kaxiras2019, photonicTB}. They serve as accurate approximations of ab initio models with locally almost periodic structures \cite{abinitioTB2014, MaxLocalWannier2012, shiang2016, shiang2015, BN2015}. Still, their numerical simulation remains challenging when we aim to capture long-time large-distance transport phenomena. It is often preferable and computationally more efficient to solve macroscopic continuous models. A typical example is single- or multi-layer graphene, which may be modeled by a continuous microscopic Schr\"odinger equation in a periodic crystal \cite{kaxiras2019}, by a tight-binding approximation \cite{MacDonald_TB2013}, as well as by a macroscopic Dirac equation \cite{BH}. A more sophisticated example is the use of continuum models to describe moir\'{e} 2D materials such as through the Bistritizer-MacDonald model \cite{bistritzer2011moire}, which is used as a base model to construct many-body models and analyze exotic many-body phenomena such as superconductivity, correlated insulators, and the fractional quantum hall effect \cite{Bernevig2022, linlin2025, vortexability2023}. Examples of continuous approximations of tight-binding models also abound in photonics applications \cite{cao2023double,huang2011dirac,lin2022mathematical,wu2015scheme}.

The objective of this paper is to present a general framework that allows us to derive macroscopic systems of equations from general tight-binding models with dispersion relations that are locally degenerate. When the dispersion relation involves a well-separated branch, it is well known that macroscopic first-order transport or second-order Schr\"odinger equations accurately represent dynamics \cite{allaire2005homogenization}. However, in the presence of degeneracies, multiple-band macroscopic models are necessary. For instance, the approximation of microscopic Schr\"odinger equations with honeycomb periodicity by Dirac equations is treated in
\cite{fefferman2014wave}. 
Related results for twisted bilayer graphene include the derivation of a continuum model from a linear Schr\"{o}dinger model \cite{cances2023} and the derivation of the same continuum model from tight-binding \cite{watson2023bistritzer} from wave packet analysis, and a derivation of arbitrary order continuum models from tight-binding for the density of states \cite{xue2025} and wave packet evolution \cite{quinn2025}.  \tm{In \cite{vafek2023}, continuum models for bilayer graphene subject to smooth deformations are derived.} The main objective of this paper is to obtain such derivations for general tight-binding problems. As in \cite{fefferman2014wave}, the accuracy of the model is demonstrated by considering the long-time evolution of wave packets spectrally localized in the vicinity of degeneracy points.

The derivation of the macroscopic models is conceptually simplified by modeling the tight-binding and continuous models as pseudo-differential operators (PDO). Under this framework, the objective is to compare the symbols of the microscopic and macroscopic operators, essentially by using Taylor expansion. This is combined with the unitarity of solution operators to Schr\"odinger equations to obtain main estimates. 

This heuristic description, however, needs to be adapted to the different applications we have in mind. They include for concreteness, models of single or multiple layer graphene models, models of twisted bilayer graphene, as well as models of strained single or multiple layers of graphene, leading, e.g., to the the presence of effective magnetic potentials, \tm{which are experimentally observed \cite{pseudomagnetic2021}}. While the PDO framework applies to all cases, the assumptions on the operator symbols vary significantly. General frameworks are presented in section \ref{sec:degenerate} while the motivational examples they are adapted to are given in section \ref{sec:appli}. The theoretical predictions are confirmed by numerical experiments in the specific example of a two-layer rhombohedral graphene model for the evolution of both bulk states and (topological) edge states.


While higher-order approximations of TB models lead to higher-order approximations of time-dependent wave packet propagation, we show that such results do not always hold spectrally. In particular, we show that the topological invariants naturally associated to the macroscopic models depend on the degree of accuracy of the model and even on the choice of unitary gauge transformations that do not modify the physics of the TB Hamiltonian. We demonstrate these results on approximations of the Haldane model of graphene.

\medskip

The rest of the paper is structured as follows. Our main assumption frameworks and theoretical approximation results are presented in section \ref{sec:degenerate}. Applications of the theory to single and multi-layer (twisted or not) graphene in the presence or not of strain are then detailed in section \ref{sec:appli} while numerical simulations confirm some of the theoretical findings in section \ref{sec:num} on the specific setting of two-layer gated rhombohedral graphene. The results on topological classifications of higher-order macroscopic models are given in section \ref{sec:topo}. 
Many auxiliary results and proofs are postponed to the appendices.

\section{Degenerate Tight Binding Hamiltonians and Main results}\label{sec:degenerate}
%
%
\paragraph{Tight binding models and pseudo-differential representation.}


In this paper, we refer to a tight-binding (TB) Hamiltonian as a self-adjoint operator acting on complex- and vector-valued smooth functions $\Rm^d\ni x\mapsto f(x)\in\Cm^n$, where $n$ is the number of local degrees of freedom (e.g., orbitals) of the model, and preserving a ``crystal structure''. A $d$-dimensional crystalline material is described by atom (physical) locations given by
\begin{align}\label{eq:crystal_structure}
    \tilde{\Lambda} := \cup_{j=1}^n (\Lambda + s_j),
\end{align}
where $s_j \in \mathbb{R}^d$ denotes the  positions of the atoms in a unit cell and $\Lambda$ is a Bravais lattice. Let $v_j$ form a basis of $\Rm^d$. The Bravais lattice $\Lambda$ is the union of an origin in $\Rm^d$ and all shifts of the origin by $\sum_{j=1}^d v_j n_j$ with $n_j\in\Zm$ for $1\leq j\leq d$.  We refer to \cite[Chapter 4]{ashcroft1976solid} for more details on lattices of the form $\tilde{\Lambda}$, known as ``crystal structures'' or ``Bravais lattices with a basis''.

The TB Hamiltonian $H$ acts as a finite sum of finitely many compositions of shifts $\tau_{\pm \tilde{v}_j}$, where $\tau_v f(x)=f(x+v)$, and of multiplications by local smooth and bounded $n\times n$ matrices $a_m(x)$ that preserve $\tilde{\Lambda}$ in the sense that if $f_j=f(x_j)\in \Cm^n$ with $(x_j)_j =\tilde\Lambda$ parametrizing the points of the lattice, then $(Hf)(x_k)$ for $x_k\in\tilde\Lambda$ is fully characterized by $(f_j)_j$. In other words, the Hamiltonian may be seen as acting on distributions supported on lattice sites, which is the more prevalent definition of a TB Hamiltonian. While it is most natural to consider the Hamiltonian in this discrete sense, the macroscopic limit is no longer discrete and we find it more convenient to consider general operators acting on functions (or distributions) defined on $\Rm^d$. Each such above composition may be written in the form of $\alpha(x+\frac v2)\tau_{v}$ for $\alpha(x)$ a smooth and bounded matrix-valued function. We assume the resulting Hamiltonian is self-adjoint.

As a general framework including the above TB models, we take as our starting point a Hamiltonian $H$ which can be represented by its Weyl symbol $a(x,\xi)\in \Cm^{n\times n}$ as (see also Appendix \ref{subsec:notation} for notation)
\begin{equation}\label{eq:Weyl}
    Hf(x) =  \ow a f (x) := \dint_{\Rm^{2d}} e^{i(x-y)\cdot\xi} a(\frac{x+y}2,\xi) f(y) \dfrac{dy d\xi}{(2\pi)^d}.
\end{equation}
We verify that the operator of multiplication by $a_m(x)$ has for Weyl symbol $a=a_m(x)$ while the shift operator $\tau_v$ has Weyl symbol $a=e^{iv\cdot\xi}$. These two operators are pseudo-differential operators (PDOs) and the product of such operators remains a pseudo-differential operators so that the symbol of a TB Hamiltonian is indeed well-defined. More precisely, we observe the composition rules: 
\[
  \alpha(x+\frac v2) \tau_v = \ow [\alpha(x)e^{iv\cdot\xi}],\quad 
\big(\alpha(x+\frac v2) \tau_v\big)^* = \tau_{-v} \alpha(x+\frac v2) = \alpha(x-\frac v2)\tau_{-v} = \ow [\alpha(x)e^{-iv\cdot\xi}],
\]
so that the symbol of the TB Hamiltonian is readily available from expressions of the matrices $a_m(x)$ and shifts $\tau_{\pm \tilde{v}_j}$.

As a concrete example illustrating the notation, a one-dimensional crystal with two orbitals $A$ and $B$ is modeled by $\Lambda= v\Zm$ with lattice spacing $v>0$ while $s_A=0$ and $0<s_B=s<v$ the positions of orbitals $A$ and $B$ in the unit cell $[0,v)$. A nearest-neighbor SSH-type Hamiltonian then takes the form
\[
 H = (u\tau_{s}+w\tau_{s-v})\sigma_+ + (u\tau_{-s}+w\tau_{v-s})\sigma_-,
\]
for $u,w>0$ lattice site hopping energies and $\sigma_\pm=\frac12(\sigma_1\pm i\sigma_2)$ where $\sigma_{1,2}$ are the standard Pauli matrices.  The operator acts on 2-vectors representing wave amplitudes at orbitals $A$ and $B$, respectively. Note that orbitals $A$ may be shifted (only) by $s$ and $s-v$ while orbitals $B$ may be shifted by $-s$ and $v-s$ to preserve $\tilde\Lambda=v\Zm \cup (v\Zm+s)$. The symbol of this operator is then
\[H =\ow a,\quad a(x,\xi) = (u e^{is\xi}+w e^{i(s-v)\xi})\sigma_+ + (u e^{-is\xi}+w e^{i(v-s)\xi})\sigma_-,\qquad (x,\xi)\in\Rm\times\Rm. \]

The main advantage of the above formulation is that both the discrete TB model and its continuous approximation are written as PDOs with explicit Weyl symbols. We consider several applications of the formalism in section \ref{sec:appli}.

\paragraph{Degenerate points of TB Hamiltonian.} Let us consider a TB Hamiltonian $H$ with constant coefficients. This implies that its symbol $a=a(\xi)=\sum_j \alpha_j e^{i \tilde{v}_j\cdot\xi}$ for a finite sum with $\alpha_j$ matrices and $\tilde{v}_j$ elements compatible with the crystal structure $\tilde{\Lambda}$. Note that $a(\xi)$, which we assume Hermitian, is periodic on the Brillouin zone $B$, the fundamental unit cell of the dual lattice $\Lambda^*$ to $\Lambda$. The dual lattice is of the form $\sum_j w_j n_j$ for $n_j\in\Zm$ and $w_j \cdot v_k=2\pi \delta_{kj}$; see \cite[Chapter 5]{ashcroft1976solid} for more details. For each $\xi\in B$, $a(\xi)$ admits $n$ real eigenvalues, where $n$ is the dimension of the vectors on which the TB Hamiltonian acts. These eigenvalues $\lambda_k(\xi)$ for $1\leq k\leq n$ form $d$-dimensional bands of absolutely continuous spectrum of $H$. 

When a branch of spectrum is well-separated from the other ones, for instance $\lambda_k(\xi)$ for $\xi$ in the vicinity of $K$ such that $\lambda_k(K)=E$, then wave packets with energy close to $E$ and wavenumber close to {$K$} are well approximated by either a transport equation when the group velocity $\nabla\lambda_k(K)\not=0$ or a parabolic (Schr\"odinger) equation when $\nabla\lambda_k(K)=0$ (with a definite Hessian); see \cite{allaire2005homogenization} or results on high-frequency homogenization.

However, when several bands meet, then the above approximations collapse and a multi-band limiting model needs to be considered. The simplest example is that of TB models for graphene as will be considered in section \ref{sec:Haldane}. We assume here the existence of a point $(K,E)$ with a degeneracy of (full) order $n$, which implies that $\lambda_k(K)=E$ for each $1\leq k\leq n$. In practice, this is not an important restriction since any branch $\lambda_j(\xi)$ such that $\lambda_j(K)\not=E$ will be spectrally uncoupled to the singular branches when the TB Hamiltonian is perturbed by slowly varying coefficients. We do not consider this technical difficulty and assume $\lambda_k(K)=E$ for all branches $k$.

Our objective is to show that limiting macroscopic models may be obtained systematically provided that the constant coefficient TB $H$ is perturbed by small and slowly varying coefficients. This generalizes in such a setting the results obtained in \cite{allaire2005homogenization} for general non-singular Bloch transforms and for Dirac singularities in a general Schr\"odinger model with periodic coefficients in \cite{fefferman2014wave}.

\paragraph{Taylor expansion near degenerate points.}
We consider a general tight-binding model where a point $(E,K)$ is identified in the dispersion relation of a certain model as a degenerate point of order $n\geq2$. 

A parameter $0<\delta\ll1$ implements small, spatially slowly varying, modulations of the degenerate constant coefficient Hamiltonian in the vicinity of $(E,K)$. Our objective is then to apply appropriate Taylor expansions in the symbol of the TB Hamiltonian to derive macroscopic PDE models that accurately predict the dynamics of wave packets localized in the vicinity of $(E,K)$ for such perturbed Hamiltonians.

We consider several scenarios with concrete physical applications and with varying degrees of mathematical difficulty. This leads to several assumptions on the Weyl symbol of the TB model and on the wave packets we consider. Our first set of hypotheses (Assumptions \ref{assumption:a} and \ref{assumption:commutator}) suffices to handle leading approximations to TB Hamiltonians for which the $K$ point where the singularity occurs is independent of the macroscopic variable. Assumption \ref{assumption:a} is later strengthened by Assumption \ref{assumption:a_higher_order}, which is required to obtain higher orders of accuracy.
Both Assumptions \ref{assumption:a} and \ref{assumption:a_higher_order}
use standard symbol classes from, e.g. \cite{Bony, DS, hormander1971fourier, Hormander, Zworski} which we define in more detail in Appendix \ref{subsec:notation}.
A separate set of hypotheses (Assumptions \ref{assumption:a2} and \ref{assumption:commutator2}) is used to capture deformed TB Hamiltonians whose degenerate points depend on the macroscopic spatial variable. In this setting, we are forced to use nonstandard symbol classes including symbols whose derivatives cannot all be bounded by a fixed polynomial at infinity; see Definition \ref{def:es} and Appendix \ref{subsec:notation} for more details.




\begin{assumption}\label{assumption:a}
Fix $\delta_0 > 0$.
For $0 < \delta \le \delta_0$, let
\begin{align}\label{eq:a_delta}
 a_\delta(x,\xi) = a(\delta x,\xi;\delta),
\end{align}
where $a$ is an $n\times n$ Hermitian-valued symbol satisfying
\begin{align}\label{eq:reg_a}
    |\partial^\alpha_X \partial^\beta_\xi a (X, \xi; \delta)| \le C_{\alpha,\beta} \delta^{-\nu_0 - \nu_1|\alpha|- \nu_2|\beta|}, \qquad (X,\xi; \delta) \in \mathbb{R}^{2d} \times (0, \delta_0]
\end{align}
for any multi-indices $\alpha, \beta \in \mathbb{N}_0^d$ and some $\nu_0, \nu_1, \nu_2 \ge 0$. 
Assume the existence of Hermitian-valued symbols $b_{0} (X,\zeta; \delta) = b_{0} \in S^1$ and $b_{1} (X,\zeta; \delta) = b_{1}\in S(\aver{\zeta}^{2})$ such that
\begin{align}\label{eq:decomposition}
    a(X,\xi;\delta) = E I_n + \delta b_{0}(X, \frac{\xi-K}\delta; \delta) + \delta^{1+\rate} b_{1}(X, \frac{\xi-K}\delta;\delta), \qquad X \in \mathbb{R}^d, \quad |\xi - K| \le C_0 \delta^{1-\eta}
\end{align}
for some $C_0, \rate > 0$ and $\eta > \nu_1$. 
Furthermore, assume there exists $c>0$ such that
\begin{align}\label{eq:ellipticty}
    |\det b_{0} (X,\zeta;\delta)|^{1/n} \ge c \aver{\zeta} - c^{-1}, \qquad (X, \zeta; \delta) \in \mathbb{R}^{2d} \times (0,\delta_0].
\end{align}
\end{assumption}
We refer to Appendix \ref{subsec:notation} for definitions of the symbol classes $S^1$ and $S (\aver{\zeta}^2)$.
    We define the tight-binding Hamiltonian by 
    \begin{equation}\label{eq:H_delta}
        H_\delta :=\ow a_\delta.
    \end{equation}
    For $\phi_0 \in \cS (\mathbb{R}^d; \mathbb{C}^n)$, let $\varphi_\delta$ be the solution to the microscopic TB evolution problem
    \begin{equation}\label{eq:microWP}
      (D_t+H_\delta)\varphi_\delta =0,\qquad \varphi_\delta(0,x)= e^{iK\cdot x} \delta^{\frac d2} \phi_0(\delta x).
    \end{equation}
    Here and below, we use the notation $D_z:=\frac 1i\partial_z$ for any variable $z$.

\begin{remark}
    When $a_\delta$ satisfies Assumption \ref{assumption:a}, the operator $H_\delta$ is self-adjoint on $L^2 (\mathbb{R}^d; \mathbb{C}^n)$ as it is bounded and symmetric for every $\delta$. Thus the function $\varphi_\delta$ in \eqref{eq:microWP} is indeed well-defined.
    
    If $a(X,\xi;\delta)$ is smooth in $\delta$, the symbol $b_{0}$ can be obtained by Taylor expanding $a$ about $(\xi;\delta) = (K;0)$, in which case $\rate = 1$; see Proposition \ref{prop:taylor}. The above more general setting includes models of twisted multilayer graphene (or more generaly operators with longer-range interactions) where such smoothness assumptions do not hold. The Schwartz-class assumption on $\phi_0$ is made for simplicity and can be considerably relaxed. 
\end{remark}

The following assumption will allow for stronger approximations of the TB solution.
\begin{assumption}\label{assumption:commutator}
    Fix $m \in \mathbb{N}_0$. When $m\geq1$, suppose that for any $N \in \{0,1, \dots, m-1\}$, the operator
    $$[\nabla_x, H_\delta] : H^{N}(\mathbb{R}^d; \mathbb{C}^n) \to H^N(\mathbb{R}^d; \mathbb{C}^{d \times n})$$
    is bounded, with
    \begin{align}\label{eq:assumption_commutator_bd_reg}
        \norm{[\nabla_x, H_\delta]}_{H^{N} \to H^N} \le C \delta, \qquad 0 < \delta \le \delta_0.
    \end{align}
\end{assumption}
\tbr{Above and throughout the paper, $H^N$ denotes the standard Sobolev space of order $N$.}
\begin{proposition}\label{prop:commutator_sufficient}
    If $H_\delta = \ow a_\delta$ satisfies \eqref{eq:a_delta}-\eqref{eq:reg_a} with $\nu_0 = \nu_1 = \nu_2 = 0$, then $H_\delta$ satisfies Assumption \ref{assumption:commutator} for all $m \in \mathbb{N}_0$.
\end{proposition}
\begin{proof}
    If $\nu_0 = \nu_1 = \nu_2 = 0$, then \eqref{eq:reg_a} implies that $a \in S(1)$ uniformly in $\delta$. This means the symbol of $[\nabla_x, H_\delta]$ is of $O(\delta)$ in $S(1)$; see, e.g. \cite{DS, Zworski} for more details. The bounds \eqref{eq:assumption_commutator_bd_reg} then follow from Lemma \ref{lemma:S_bd}.
\end{proof}

\paragraph{Macrosopic model, localized wave packets, and approximation error.}
Let 
\begin{equation}\label{eq:macroH_0}
 \bH := \ow b_{0}
\end{equation}
be the macroscopic Hamiltonian with $b_{0}$ as in Assumption \ref{assumption:a}.
The assumption that $b_0 \in S^1$ is Hermitian-valued together with the ellipticity condition \eqref{eq:ellipticty} imply that the operator $\bH$ is self-adjoint on $L^2(\mathbb{R}^d; \mathbb{C}^n)$ with domain of definition $H^1(\mathbb{R}^d; \mathbb{C}^n)$. Since all relevant bounds are uniform in $\delta$, this means that 
\begin{align}\label{eq:equivalent_norms}
        c \norm{u}_{H^N} \le \norm{(i+\bH)^N u}_{L^2} \le C \norm{u}_{H^N}, \qquad u \in H^N, \quad 0 < \delta \le \delta_0
\end{align}
for any $N \ge 0$.
We define $\phi$ as the solution to
\begin{equation}\label{eq:macro_0}
  (D_T+\bH) \phi(T,X;\delta)=0,\qquad \phi(0,X; \delta)=\phi_0(X).
\end{equation}
Since $b_{0}$ depends on $\delta$, so does the solution $\phi(T,X;\delta)$.
We next construct the ansatz
\begin{align}\label{eq:psi_def}
  \psi_\delta(t,x) = e^{i(K\cdot x-Et)} \delta^{\frac d2} \phi(\delta t,\delta x;\delta).
\end{align}
\begin{proposition}\label{thm:errortime_0}
Suppose Assumption \ref{assumption:a} 
holds for some $\delta_0, \rate > 0$.
Then for any $\tau > 0$, there exists a constant $C_\tau > 0$ such that
\begin{equation} \label{eq:errortime_0}
    \sup_{0 \le t \le \frac{\tau}{\delta}}
  \|\psi_\delta(t,\cdot)-\varphi_\delta(t,\cdot)\|_{L^2(\Rm^d)} \leq C_\tau \delta^{\rate}
\end{equation}
uniformly in $0 < \delta \le \delta_0$.
\end{proposition}

The above result (a direct consequence of Theorem \ref{thm:main_0} below) provides an error in the $L^2(\Rm^d)$ sense. This corresponds to an estimate averaged over shifted copies of the crystal structure $\tilde\Lambda$ in \eqref{eq:crystal_structure}. To obtain estimates in $\ell^2 (\tilde\Lambda)$,
it is sufficient to ensure that the solutions are sufficiently smooth to be evaluated on grid points. To make this precise, we state a classical result whose proof can be found in, e.g. \cite[Lemma A.4]{quinn2025}.
\begin{proposition}\label{prop:Sobolev}
    For any $m > d/2$, there exists a positive constant $C_m$ such that $$\norm{u}_{\ell^2 (\tilde \Lambda)} \le C_m \norm{u}_{H^m (\mathbb{R}^d)}, \qquad u \in H^m (\mathbb{R}^d).$$
\end{proposition}
We will henceforth establish the validity of our macroscopic approximations in $H^m (\mathbb{R}^d)$, with Proposition \ref{prop:Sobolev} immediately implying the validity of these approximations in the discrete setting (provided $m > d/2$).
We are now ready to state
the first main theoretical result of the paper.

\begin{theorem}\label{thm:main_0}
    Suppose Assumptions \ref{assumption:a} and \ref{assumption:commutator} hold for some $\delta_0, \rate > 0$ and $m \in \mathbb{N}_0$, respectively. Then 
    \begin{align}\label{eq:estimateHm}
        \|(\psi_\delta-\varphi_\delta)(t, \cdot)\|_{H^m(\Rm^d)} \leq C\delta^{1+\rate} t \left(1 + (\delta t)^{m}\right)
    \end{align}
    uniformly in $t \ge 0$ and $0 < \delta \le \delta_0$.
\end{theorem}

Note that this provides a bound of the form $C(T)\delta^{\rate}$ for any $\delta t\leq T$, and thus directly implies Proposition~\ref{thm:errortime_0} (Assumption \ref{assumption:commutator} is vacuously true if $m=0$). 
The result remains meaningful for $1\ll \delta t \ll \delta^{-\frac{\mu}{m+1}}$.

\begin{proof}[Proof]
    Observe that $u_\delta := \psi_\delta - \varphi_\delta$ satisfies $(D_t + H_\delta) u_\delta = (D_t + H_\delta) \psi_\delta$ and $u_\delta (0, \cdot) = 0$. The definitions of $H_\delta$ and $\psi_\delta$ imply the following manipulations that are central to our convergence results:
    \begin{align*}
        H_\delta \psi_\delta(t,x) &= \frac{1}{(2\pi)^d} \dint_{\Rm^{2d}} e^{i(x-y)\cdot \xi} a(\delta\frac{x+y}2,\xi;\delta) e^{i(K\cdot y-Et)} \delta^{\frac d2} \phi(\delta t,\delta y;\delta) dy d\xi\\
        &=
        \dfrac{e^{i(K\cdot x-Et)}}{(2\pi)^d}  \dint_{\Rm^{2d}} e^{i(\delta x-\delta y)\cdot \zeta}  a(\frac{\delta x+\delta y}2,K+\delta \zeta;\delta)  \delta^{\frac d2} \phi(\delta t,\delta y;\delta) d \delta y d\zeta
        \\
        D_t \psi_\delta(t,x) &= e^{i(K\cdot x-Et)} \delta^{\frac d2} (\delta D_T\phi-E\phi)(\delta t,\delta x;\delta).
    \end{align*}
    It follows from the definition \eqref{eq:macro_0} of $\phi$ that
    \begin{align}\label{eq:residual}
        (D_t + H_\delta) \psi_\delta(t,x) &= \frac{e^{i(K\cdot x-Et)}}{(2\pi)^d} \int_{\mathbb{R}^{2d}}e^{i(\delta x-\delta y)\cdot \zeta} \\
        &\hspace{0.5cm}\times \left( a(\frac{\delta x+\delta y}2,K+\delta \zeta;\delta) - E I_n - \delta b_{0}(\frac{\delta x+\delta y}2, \zeta; \delta)\right)\delta^{\frac d2} \phi(\delta t,\delta y;\delta) {\rm d} \delta y {\rm d}\zeta.\nonumber
    \end{align}
    The above calculations illustrate the main strategies used in the paper to obtain convergence results. Provided the above residual symbol is appropriately small (e.g., of order $\delta^2$ in $S(\aver{|\zeta|^2})$ as in Proposition \ref{prop:taylor} below), then $(D_t + H_\delta) u_\delta$ is comparably small in the $L^2$ sense as an application of Lemma \ref{lemma:S_bd} and an elliptic regularity property on $\phi$. Unitarity of the solution operator to $(D_t + H_\delta)$ then allows us to obtain \eqref{eq:estimateHm} when $m=0$. 

    More specifically and more generally, under Assumption \ref{assumption:a}, Lemma \ref{lemma:res} below implies that for any $N\geq0$, then we have the estimate on the residual
    \(\norm{(D_t + H_\delta) u_\delta (t, \cdot)}_{H^N}\le C_N \delta^{1+\rate}\)
    uniformly for $t\geq0$ and $0 < \delta \le \delta_0$.
    Under the commutator Assumption \ref{assumption:commutator}, Lemma \ref{lemma:u_delta_reg} below (with $M=0$ there) allows us to propagate errors on residuals to errors on the difference $\psi_\delta - \varphi_\delta$ and conclude the proof of the result.
\end{proof}

\paragraph{Position-dependent degenerate points.}
%

We now extend the analysis to the case of defining lattice parameters and degenerate points $K$ that may depend on the macroscopic variable $X$. This finds applications in strained and twisted materials, in which the orientation of the underlying crystal may vary at the macroscopic scale. In such settings, the corresponding symbol $a_\delta$ will no longer be in $S(1)$ in general as was required by Assumption \ref{assumption:a}. Indeed, $X$-derivatives of a symbol of the form $e^{iv(X) \cdot \xi}$ grow polynomially in $\xi$. We thus introduce the following more general class of (non-standard) symbols.
\begin{definition}\label{def:es}
    Let $\es$ denote the class of $n \times n$ matrix-valued symbols $a = a(X,\xi)$ such that each entry of $a$ is a finite sum of terms of the form
    \begin{align*}
        e^{i \phase (X,\xi)}\sym(X,\xi) 
    \end{align*}
    for some $\phase \in S^1$ and $\sym \in S(1)$. We require that $-\Im \Phi \le C$ so that $a$ is bounded. If $a = a(X,\xi; \delta)$ also depends on a parameter $\delta$, we say that $a \in \es$ uniformly in $\delta$ (still denoted by $a \in \es$) if all bounds on the symbols $\phase$ and $\sym$ in their respective symbol classes are uniform in $\delta$. If $a \in \es$, we say that $a$ is a symbol of ``exponential type''.
\end{definition}
The symbol classes $S^1$ and $S(1)$ are defined in Appendix \ref{subsec:notation}.
See Appendix \ref{subsec:exp} for boundedness properties of slowly-varying symbols in $\es$.
We then make the following assumption, which generalizes the degenerate point from Assumption \ref{assumption:a}. 
The symbol class $\cs^2$ below captures polynomial growth in $\xi$ of derivatives with respect to $X$, and is also defined in Appendix \ref{subsec:notation}.
Throughout, we use $C^\infty_b (\mathbb{R}^d)$ to denote the space of functions $u \in C^\infty (\mathbb{R}^d)$ such that $\norm{\partial^\alpha u}_{L^\infty (\mathbb{R}^d)}< \infty$ for all multi-indices 
$\alpha \in \mathbb{N}_0^d$.

\begin{assumption}\label{assumption:a2}
    Fix $\delta_0 > 0$. We consider the $n\times n$ Hermitian-valued symbol, 
    \[
     a_\delta(x,\xi) = a(\delta x,\xi;\delta),
    \]
    where
    $a (\cdot, \cdot \; ; \delta) \in \es$ uniformly in $0 < \delta \le \delta_0$. 
    Let $K \in \mathbb{R}^d$ and $B \in C^\infty_b (\mathbb{R}^d)$, and define $A(X) := K \cdot X + B(X)$ and $\dpx (X) := \nabla A(X) = K + \nabla B(X)$.
    Let $E \in \mathbb{R}$ and assume that there exist Hermitian-valued symbols $b_{0} = b_{0} (X,\zeta; \delta)$ and $b_{1} = b_{1} (X,\zeta; \delta)$ with 
    $b_{0}\in S^1$ and $b_{1} \in \cs^2$
    uniformly in $0 < \delta \le \delta_0$, such that
\begin{align}\label{eq:decomposition2}
    a(X,\xi;\delta) = E I_n + \delta b_{0}(X, \frac{\xi-\dpx (X)}\delta; \delta) + \delta^{1+\rate} b_{1}(X, \frac{\xi-\dpx (X)}\delta;\delta) \qquad \forall \quad  
    |\xi - \dpx(X)| \le C_0 \delta^{1-\eta},
\end{align}
for some $0 < \rate \le 1$ and 
$C_0, \eta > 0$. 
Moreover, assume that for some $c>0$,
    \begin{align}\label{eq:ellipticty2}
        |\det b_{0} (X,\zeta)|^{1/n} \ge c \aver{\zeta} - c^{-1}, \qquad (X, \zeta; \delta) \in \mathbb{R}^{2d} \times (0,\delta_0].
    \end{align}
\end{assumption}
As before, $H_\delta$ and $\bH$ are defined in \eqref{eq:H_delta} and \eqref{eq:macroH_0}, respectively, while $\phi$ is the solution of \eqref{eq:macro_0}.
Note that $\phi$ is well-defined, as the hypotheses on $b_0$ in Assumption \ref{assumption:a2} guarantee that $\bH$ is self-adjoint on $L^2 (\mathbb{R}^d; \mathbb{C}^n)$ with domain of definition $H^1(\mathbb{R}^d; \mathbb{C}^n)$.
Our wave-packet ansatz will now be given by
\begin{align}\label{eq:ansatz_2}
      \psi_\delta(t,x) = e^{i(\frac{1}{\delta} \vp (\delta x) -Et)} \delta^{\frac d2} \phi(\delta t,\delta x;\delta).
\end{align}
\begin{remark}\label{remark:dpx}
When $B \equiv 0$, we recover the original ansatz \eqref{eq:psi_def}. The assumption that $\dpx$ be a conservative vector field is used in the ansatz \eqref{eq:ansatz_2}, which is given in terms of the corresponding vector potential $A$. 
    As will be evident in the proof of Theorem \ref{thm:main2} (see in particular \eqref{eq:H00} and below), the ansatz \eqref{eq:ansatz_2} leads naturally to an expansion of the symbol $a(X, \xi; \delta)$ about the point $\xi = \nabla A (X)$. Although the theory might also apply to non-conservative vector fields $\dpx$, we do not pursue this here.
\end{remark}
We make one additional assumption, which generalizes the commutator bound \eqref{eq:assumption_commutator_bd_reg} from Assumption \ref{assumption:commutator}. 
\begin{assumption}\label{assumption:commutator2}
    Fix $m \in \mathbb{N}_0$ and
    suppose $\{ H_\delta : 0 < \delta \le \delta_0\}$ is a family of self-adjoint operators on $L^2 (\mathbb{R}^d; \mathbb{C}^n)$.
    When $m\geq1$, assume that for any $N \in \{0,1, \dots, m-1\}$, the operator
    $$[\nabla_x, H_\delta] : H^{N+1}(\mathbb{R}^d; \mathbb{C}^n) \to H^N(\mathbb{R}^d; \mathbb{C}^{d \times n})$$
    is bounded, with
    \begin{align}\label{eq:assumption_commutator_bd_reg2}
        \norm{[\nabla_x, H_\delta]}_{H^{N+1} \to H^N} \le C \delta, \qquad 0 < \delta \le \delta_0.
    \end{align}
\end{assumption}
Note that Assumption \ref{assumption:a2} does not guarantee the self-adjointness of $H_\delta$, which is why this condition is included in Assumption \ref{assumption:commutator2}. For $\phi_0 \in \cS (\mathbb{R}^d; \mathbb{C}^n)$, we can then define $\varphi_\delta$ as the solution to 
\begin{align}\label{eq:microWP2}
    (D_t+H_\delta)\varphi_\delta =0,\qquad \varphi_\delta(0,x)= e^{iK\cdot x} \delta^{\frac d2} \phi_0(\delta x).
\end{align}
We are now ready to state the extension of Theorem \ref{thm:main_0} to non-uniform degenerate points.
\begin{theorem}\label{thm:main2}
    Suppose Assumptions \ref{assumption:a2} and \ref{assumption:commutator2} hold for some $\delta_0 > 0$, $0 < \rate \le 1$ and $m \in \mathbb{N}_0$, and define $\psi_\delta$ and $\varphi_\delta$ by \eqref{eq:macro_0}-\eqref{eq:ansatz_2}-\eqref{eq:microWP2}.
    For any 
    $\eps > 0$, there exist constants $\gamma, C > 0$ such that
    \begin{align*}
        \|(\psi_\delta-\varphi_\delta)(t, \cdot)\|_{H^m(\Rm^d)} \leq C\delta^{\rate - \eps} \left(e^{\gamma \delta t} - 1\right)
    \end{align*}
    uniformly in $t \ge 0$ and $0 < \delta \le \delta_0$.
\end{theorem}
The proof of this result is postposed to Appendix \ref{subsec:proof_main_2}.

\paragraph{Higher-order approximations.}
We now strengthen Assumption \ref{assumption:a} to capture macroscopic models with a higher order of accuracy. We restrict ourselves to the setting where the lattice parameters and $K$ remain constant macroscopically. The symbol classes $S^p$ and $S (\aver{\zeta}^{p+1})$ used below are defined in Appendix \ref{subsec:notation}.
\begin{assumption}\label{assumption:a_higher_order}
    Suppose Assumption \ref{assumption:a} holds for $\rate, \eta>0$.
    Fix $p \in \mathbb{N}$ and assume
    there exist Hermitian-valued symbols $b_{0p} (X,\zeta; \delta) = b_{0p} \in S^p$ and $b_{1p} (X,\zeta; \delta) = b_{1p}\in S(\aver{\zeta}^{p+1})$ such that
    \begin{align}\label{eq:decomposition_higher_order}
        a(X,\xi;\delta) = E I_n + \delta b_{0p}(X, \frac{\xi-K}\delta; \delta) + \delta^{p+\rate} b_{1p}(X, \frac{\xi-K}\delta;\delta), \qquad X \in \mathbb{R}^d, \quad |\xi - K| \le C_0 \delta^{1-\eta}.
    \end{align}
    Moreover, assume there exists $c>0$ such that we have the semiclassical ellipticity estimate:
    \begin{align}\label{eq:ellipticty_higher_order}
        |\det \delta b_{0p} (X,\zeta;\delta)|^{1/n} \ge c \aver{\delta\zeta}^p - c^{-1}, \qquad (X, \zeta; \delta) \in \mathbb{R}^{2d} \times (0,\delta_0].
    \end{align}  
\end{assumption}
Note that by Assumption \ref{assumption:a}, \eqref{eq:decomposition} still holds above with $b_0$ elliptic as described in \eqref{eq:ellipticty}. The assumptions on $b_{0p}$ imply that the operator $\bH_p := \ow b_{0p}$ is self-adjoint on $L^2 (\mathbb{R}^d; \mathbb{C}^n)$, with
    \begin{align}\label{eq:equivalent_norms_p}
        c \delta^N \norm{u}_{H^N} \le \norm{(i+\bH_p)^{N/p} u}_{L^2} \le C \norm{u}_{H^N}, \qquad u \in H^N, \quad 0 < \delta \le \delta_0,
    \end{align}
    for any $N \ge 0$; see, e.g. \cite{Bony,DS,Zworski}.
We can therefore define $\phi$ as the solution to
\begin{equation}\label{eq:macro}
  (D_T+\bH_p) \phi(T,X;\delta)=0,\qquad \phi(0,X)=\phi_0(X),
\end{equation}
and construct $\psi_\delta$ in \eqref{eq:psi_def}  as before.
We then have
\begin{theorem}\label{thm:main}
    Suppose Assumptions \ref{assumption:commutator} and \ref{assumption:a_higher_order} hold for some $m \in \mathbb{N}_0$,
    $\delta_0, \eta, \rate > 0$ and $p \in \mathbb{N}$, and
    take $\varphi_\delta$ as in \eqref{eq:microWP}. 
    Then there exist constants $C, \mbp > 0$ such that
    \begin{align*}
        \|(\psi_\delta-\varphi_\delta)(t, \cdot)\|_{H^m(\Rm^d)} \leq C\delta^{p+\rate} t \left(1 + (\delta t)^{m+\mbp}\right)
    \end{align*}
    uniformly in $t \ge 0$ and $0 < \delta \le \delta_0$.
\end{theorem} 
This provides a bound of the form $C(T)\delta^{p-1+\rate}$ for any $\delta t\leq T$, and thus strengthens the order of convergence from Theorems \ref{thm:main_0} and \ref{thm:main2}. The proof of this result is postponed to Appendix \ref{subsec:pf_higher_order}. The constant $\mbp$ is given explicitly by $$\mbp := \lceil (m + \max\{\nbp,p\} + 1)/\rate\rceil, \quad \nbp := \lceil (\rate + p + \nu_0 + \nu (\ncv + m))/(\eta - \nu_1) \rceil,$$ with $\nu := \max\{\nu_1, \nu_2\}$ and $\nu_0, \nu_1, \nu_2$ defined in Assumption \ref{assumption:a} (so that $\nu_1 < \eta$), and $\ncv$ as in Lemma \ref{lemma:S_bd}.

\begin{remark}
    The higher order of convergence in Theorem \ref{thm:main} allows for uniform approximations. Indeed, Proposition \ref{prop:Sobolev} and Theorem \ref{thm:main} with $m>d/2$ imply that for any $T > 0$,
    \begin{align*}
    \sup_{x\in\Rm^d}
        |(\psi_\delta - \varphi_\delta) (t, x)| \le \norm{(\psi_\delta - \varphi_\delta) (t, \cdot)}_{\ell^2 (\mathbb{Z}^d)} \le C\|(\psi_\delta-\varphi_\delta)(t, \cdot)\|_{H^m(\Rm^d)} \le C \delta^{p-1 + \rate}
    \end{align*}
    uniformly in $0 \le t \le T/\delta$ and $0 < \delta \le \delta_0$ since any point $x$ can be embedded into a lattice. To control relative errors for initial data $\varphi_\delta (0,x)$ in \eqref{eq:microWP} normalized by $\delta^{d/2}$, the estimate of interest is
    \begin{align*}
        \frac{\sup_{x\in\Rm^d}
        |(\psi_\delta - \varphi_\delta) (t, x)| }{\norm{\varphi_\delta (0, \cdot)}_{L^\infty (\mathbb{R}^d)}} = C \delta^{-\frac d2}\sup_{x\in\Rm^d}
        |(\psi_\delta - \varphi_\delta) (t, x)|\le C \delta^{p-1 + \rate - \frac{d}{2}},
    \end{align*}
    which is meaningful when $p > 1 - \rate + \frac{d}{2}$. Note that the factor $\delta^{\frac d2}$ in \eqref{eq:microWP} is chosen to make $\norm{\varphi_\delta (0, \cdot)}_{L^2 (\mathbb{R}^d)}$ independent of $\delta$, thus the relative and absolute errors in $H^m (\mathbb{R}^d)$ are of the same order.

    More generally, with the notation $\tilde\psi(t,X)=\psi(t,x)$, we find that $\|\tilde\psi\|_{H^m(\Rm^d)}\leq \delta^{-m}\|\psi\|_{H^m(\Rm^d)}$ so that in macroscopic variables, $\|\tilde\psi_\delta-\tilde\varphi_\delta\|_{H^m(\Rm^d)} \leq C \delta^{p+\mu-1-m}$.
\end{remark}

We conclude this section by establishing that if the symbol $a$ from Assumption \ref{assumption:a} is in $S(1)$ uniformly in $\delta$ and is sufficiently smooth as a function of $\delta$, then the symbol $b_{0p}$ can be obtained by a Taylor expansion.
\begin{proposition}\label{prop:taylor}
    Fix $p \in \mathbb{N}$ 
    and suppose $a\in C^{p+1} (\mathbb{R}^{2d} \times [0,\delta_0])$ satisfies \eqref{eq:reg_a} with $\nu_0 = \nu_1 = \nu_2 = 0$. Then
    \begin{align*}
        a(X,\xi;\delta) = a(X,K;0) + \delta b_{0p}(X, \frac{\xi-K}\delta; \delta) + \delta^{p+1} b_{1p}(X, \frac{\xi-K}\delta;\delta), \qquad (X,\xi) \in \mathbb{R}^{2d}, \quad 0 < \delta \le 1,
    \end{align*}
    where  $\bar K := (K; 0) \in \mathbb{R}^{d+1}$, $\bar{\zeta} := (\zeta; 1)\in \mathbb{R}^{d+1}$, and $\partial_\ell := \partial_{\xi_\ell}$ if $1 \le \ell \le d$ and $\partial_{d+1} := \partial_\delta$,
    \begin{align*}
        \delta b_{0p} (X,\zeta;\delta) = \sum_{j=1}^p \frac{\delta^j}{j!} \sum_{i_1, \cdots, i_j = 1}^{d+1} \bar{\zeta}_{i_1} \dots \bar{\zeta}_{i_j} \partial_{i_1, \dots, i_j} a(X,\bar{K}), \qquad b_{0p} \in S^p, \qquad b_{1p} \in S (\aver{\zeta}^{p+1}).
    \end{align*}
\end{proposition}
\begin{proof}
    This is a direct consequence of Taylor's theorem.
\end{proof}

\section{Applications}
\label{sec:appli}

This section applies the above theory to several tight-binding Hamiltonians that find applications primarily in models of single- and multi-layer graphene. 

We begin by analyzing the Haldane model \eqref{eq:Hdelta}-\eqref{eq:a0}, which satisfies Assumptions \ref{assumption:commutator} and \ref{assumption:a_higher_order} (for any $m$ and $p$, respectively) with effective symbol $b_{0p}$ given by \eqref{eq:expansion_p}. When $p=1,2$, we provide explicit formulas \eqref{eq:b01}-\eqref{eq:b02}-\eqref{eq:macroHaldane} for $b_{0p}$ and its corresponding macroscopic Hamiltonian. 
In the presence of strain, Assumptions \ref{assumption:a} and \ref{assumption:commutator} no longer apply so we instead verify Assumptions \ref{assumption:a2} and \ref{assumption:commutator2}. Effective symbols capturing strain and twist effects are given explicitly by \eqref{eq:b_theta}-\eqref{eq:small_twist}-\eqref{eq:twist_corr}. For larger perturbations of the Haldane model (in which deformations of the lattice vectors are no longer small), we define the lattice vectors by \eqref{eq:a_X_def}-\eqref{eq:relations}-\eqref{eq:v1_X_def} and verify that Assumptions \ref{assumption:a2} and \ref{assumption:commutator2} still hold with effective symbol $b_0$ given by \eqref{eq:b0_g}-\eqref{eq:g}.

In Section \ref{sec:multilayer} we show that commensurate multilayer graphene models fit the framework of Assumptions \ref{assumption:commutator} and \ref{assumption:a_higher_order} so that Theorem \ref{thm:main} applies; see \eqref{eq:macromulti} for the macroscopic $n$-layer Hamiltonian. Finally, we apply our theory to a tight-binding model of twisted bilayer graphene (TBG) in Section \ref{subsec:tbg}, presenting a macroscopic model and establishing its validity with Proposition \ref{prop:TBG}. Our macroscopic model for TBG is unitarily equivalent to the Bistritzer-MacDonald model \cite{bistritzer2011moire, watson2023bistritzer}; see Remark \ref{remark:BM}.

\subsection{Haldane model} \label{sec:Haldane}

We start with the classical Haldane model \cite{haldane1988model}, a hexagonal bipartite model with a periodic array of sites $A$ and sites $B$.  See Fig. \ref{fig:geomgraphene} for the geometry of the TB model.
\begin{figure}  \center
\begin{tikzpicture}[scale=1.]
 \pgfmathsetmacro{\r}{1.4} 
 \pgfmathsetmacro{\s}{sqrt(3)*\r}
 \draw[thick,color=black] plot[ samples at={0,60,...,360},variable=\x] (\x:\r) ; 
 \draw[thick,color=black,shift={(1.5*\r,\s/2)}] plot[ samples at={0,60,...,360},variable=\x] (\x:\r) ; 
 \draw[thick,color=black,shift={(1.5*\r,-\s/2)}] plot[ samples at={0,60,...,360},variable=\x] (\x:\r) ; 
 \draw[thick,color=black,shift={(3*\r,0*\s)}] plot[ samples at={0,60,...,360},variable=\x] (\x:\r) ; 
  \draw[thick,magenta,->] (0,0) --(1.5*\r,\s/2);
  \draw[thick,magenta,->] (0,0) --(1.5*\r,-\s/2);
  \node[color=magenta, right] at (1.5*\r,\s/2) {$v_1$};
  \node[color=magenta, right] at (1.5*\r,-\s/2) {$v_2$};
  \draw[thin,color=blue,mark=*] plot[ samples at={60,180,...,420},variable=\x] (\x:\r) ;
  \draw[thin,color=blue,shift={(1.5*\r,\s/2)},mark=*] plot[ samples at={60,180,...,420},variable=\x] (\x:\r) ;
  \draw[thin,color=blue,shift={(1.5*\r,-\s/2)},mark=*] plot[ samples at={60,180,...,420},variable=\x] (\x:\r) ;
  \draw[thin,color=blue,shift={(3*\r,0*\s)},mark=*] plot[ samples at={60,180,...,420},variable=\x] (\x:\r) ;
  \node[color=blue, right] at (2.1*\r,+0.00*\r) {B}; \filldraw[fill=blue!80!white, draw=black]  (2*\r,-0.0*\r) circle (1mm);
  \draw[thin,color=red,mark=*] plot[ samples at={0,120,...,360},variable=\x] (\x:\r) ;
  \draw[thin,color=red,shift={(1.5*\r,\s/2)},mark=*] plot[ samples at={0,120,...,360},variable=\x] (\x:\r) ;
  \draw[thin,color=red,shift={(1.5*\r,-\s/2)},mark=*] plot[ samples at={0,120,...,360},variable=\x] (\x:\r) ;
  \draw[thin,color=red,shift={(3*\r,0*\s)},mark=*] plot[ samples at={0,120,...,360},variable=\x] (\x:\r) ;
  \node[color=red, right] at (0.97*\r,0.2*\r) {A};\filldraw[fill=red!80!white, draw=black]  (\r,-0.0*\r) circle (1mm);
\end{tikzpicture}
\hspace{1.25cm} 
\begin{tikzpicture}[scale=1]
 \pgfmathsetmacro{\r}{1.4} 
 \draw[thick,color=black,mark=*,shift={(\r,\r)}] plot[ samples at={30,90,...,390},variable=\x] (\x:\r) ; 
 \node[color=black] at (1.0*\r,-0.2*\r) {$K$}; \filldraw[fill=black!80!white, draw=black]  (1.*\r,-0.0*\r) circle (1mm);
 \node[color=black] at (1.0*\r,2.2*\r) {$K'$}; \filldraw[fill=black!80!white, draw=black]  (1.*\r,2*\r) circle (1mm);
 \node[color=black] at (1.15*\r,1.05*\r) {$\Gamma$}; \filldraw[fill=black!80!white, draw=black]  (1*\r,1*\r) circle (.65mm);
 \filldraw[fill=white!80!white, draw=white]  (1.*\r,-0.6*\r) circle (1mm);
 \node[color=black,right] at (1.9*\r,1.5*\r) {$K$};\node[color=black,right] at (1.9*\r,.5*\r) {$K'$};
\end{tikzpicture}
\caption{Geometry of graphene bipartite lattice. Left: spatial lattice. Right: dual lattice. }
\label{fig:geomgraphene}
\end{figure}
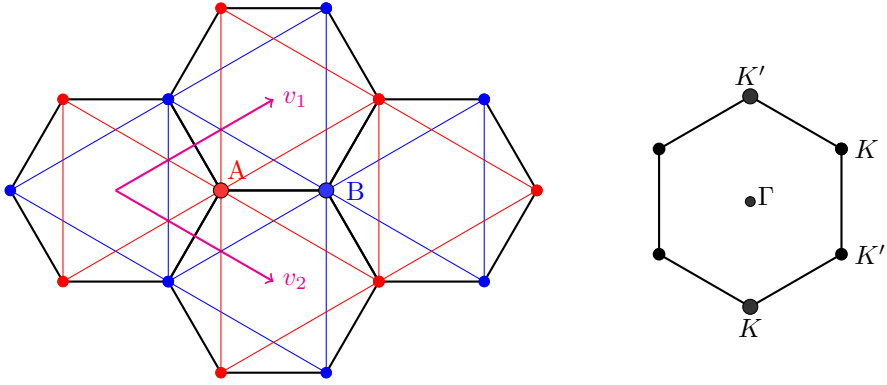

\paragraph{Microscopic description.}
Set $v > 0$ be a lattice spacing between elements of the Bravais lattice (considered to be of order $O(1)$) and define
\begin{align}\label{eq:vs}
    v_1 := \frac{v}{2} (\sqrt{3},1), \qquad v_2 := \frac{v}{2} (\sqrt{3},-1).
\end{align}
We have $\la_1=\frac13(v_1+v_2)=v\frac{\sqrt3}3 e_x$ with $\la_2=R_{2\pi/3}\la_1$ and $\la_3=R_{-2\pi/3}\la_1$ the three vectors linking a site $A$ to its nearest neighbors $B$. Here, $R_\theta$ is the anticlockwise rotation of $\theta$. We then define $\lb_1=v_1$ and $\lb_2=v_2-v_1$ with $\lb_3=-(\lb_1+\lb_2)$ the three vectors linking a site $A$ ($B$) to its nearest neighbors $A$ ($B$). Thus, the six nearest same-type neighbors are described by shifts by $\pm \lb_j$.  Note that both the sites $A$ and the sites $B$ (modulo a fixed translation) occupy the edges of a lattice $\Lambda=\Zm v_1+\Zm v_2$, whence the name of a bipartite lattice modeling the honeycomb structure.  We introduce
\begin{align}\label{eq:H1}
  H_1 = t_1 \sum_{j=1}^3 \big( \tau_{\la_j} \sigma_+ +  \tau_{-\la_j} \sigma_-\big)
\end{align}
where $\sigma_\pm=\frac12(\sigma_1\pm i\sigma_2)$ and the shift operators are defined as $\tau_\la f(x) = f(x+\la)$. 

A term $t_2$ next models interactions between next-nearest neighbors. This interaction also includes a local magnetic component (Aharonov-Bohm phase) $e^{\pm i\phi}$ leading to the Hamiltonian
\[
 H_2 = \frac{t_2}{2} \sum_{j=1}^3 (e^{i\phi} \tau_{\lb_j} + e^{-i\phi} \tau_{-\lb_j} \big) \frac{I-\sigma_3}2 +  \frac{t_2}{2} \sum_{j=1}^3 (e^{i\phi} \tau_{-\lb_j} + e^{-i\phi} \tau_{\lb_j} \big) \frac{I+\sigma_3}2.
\]
We finally have a local mass term $M$ that we assume is asymmetric at the sites $A$ and $B$. The unperturbed Haldane model is given by
\begin{equation}\label{eq:Haldanemodel}
  H = H_1 + H_2 + M \sigma_3.
\end{equation}

In the Fourier domain, $H={\mathcal F}^{-1}(\hat H_1+\hat H_2 + M\sigma_3){\mathcal F}$, where
\begin{align*}
  \hat H_1 
  = t_1 \sum_{j=1}^3\big( \cos \la_j\cdot\xi \sigma_1 - \sin \la_j\cdot\xi \sigma_2\big),\quad
  \hat H_2 = 
  t_2 \cos\phi\sum_{j=1}^3   \cos \xi\cdot \lb_j I + t_2  \sin\phi\sum_{j=1}^3  \sin \xi\cdot \lb_j \sigma_3.
\end{align*}
With this convention, we find that 
$
  \hat H = \sum_{j=0}^3 h_j \sigma_j
$
with $\sigma_0=I_2$ the identity matrix and
\begin{align*}
  h_0 & = t_2  \cos\phi\sum_{j=1}^3  \cos \xi\cdot \lb_j, 
  \ h_1 
  = t_1 \sum_{j=1}^3  \cos \la_j\cdot\xi ,\ h_2 = -  t_1 \sum_{j=1}^3  \sin \la_j\cdot\xi,\ h_3 =M+t_2 \sin\phi\sum_{j=1}^3   \sin \xi\cdot \lb_j.
\end{align*}
These two bands have a common point when all terms $h_{1,2,3}=0$. This implies 
$
 \sum_{j=1}^3 e^{i\la_i\cdot\xi} =0,
$
which admits two inequivalent Dirac points 
{$K=-\frac{4\pi}{3 v}(0,1)$} and $K'=-K$ as well as all translates of $K$ and $K'$ by multiples of the dual lattice. A macroscopic description of this band crossing is then possible for $E=0$ and $\tilde K$ any of the above Dirac points. The term $h_3$ then needs to be small, implying a choice of scaling $\delta M(\delta x)$ and $\delta t_2(\delta x)$ for the two mass terms.

\paragraph{Macroscopic description and Dirac model.}
We now verify Assumptions \ref{assumption:commutator} and \ref{assumption:a_higher_order} in order to apply Theorem \ref{thm:main} in the vicinity of the band crossing $(0,K)$ for $K$ the above Dirac point.
For ease of exposition, we set $\phi := \pi/2$. We now allow for spatial dependence of the coefficients of $H$ in \eqref{eq:Haldanemodel} and define
\begin{align} \label{eq:Hdelta}
   H_\delta =\ow a_\delta\quad \mbox{ with } \quad a_\delta (x,\xi) = a(\delta x,\xi;\delta) = a^0 (\delta x,\xi) + \delta a^1 (\delta x, \xi),
\end{align}
\begin{align}\label{eq:a0}
\begin{split}
    a^0 (X, \xi) &= t_1 (X) \begin{pmatrix}
        0 & e^{i \xi \cdot \la_1} + e^{i\xi \cdot \la_2} + e^{i\xi \cdot \la_3}\\
        e^{-i \xi \cdot \la_1} + e^{-i\xi \cdot \la_2} + e^{-i\xi \cdot \la_3} & 0
    \end{pmatrix},
    \\
    a^1 (X,\xi) &= \Big(M(X) + t_2 (X) \sum_{j=1}^3 \sin (\lb_j \cdot \xi)\Big) \sigma_3,
\end{split}
\end{align}
for some 
functions $M, t_1, t_2 \in C^\infty_b (\mathbb{R}^2)$, with $|t_1 (X)| \ge c > 0$ bounded away from zero. Defining $a^j_\delta (x, \xi) := a^j (\delta x, \xi)$ for $j=1,2$, the operator $H_\delta = \ow a^0_\delta + \ow a^1_\delta$ is given explicitly by
\begin{align*}
    \ow a^0_\delta &= \sum_{j=1}^3 (t_1 (\delta (x+\frac{1}{2}\la_j)) \tau_{\la_j} \sigma_+ + t_1 (\delta (x-\frac{1}{2}\la_j)) \tau_{-\la_j} \sigma_-),\\
    \ow a^1_\delta &= \Big(M(\delta x) + \frac{i}{2}\sum_{j=1}^3 (t_2 (\delta (x-\frac{1}{2} \lb_j))\tau_{-\lb_j} - t_2 (\delta (x+\frac{1}{2} \lb_j))\tau_{\lb_j} \big) \Big) \sigma_3.
\end{align*}
It is clear that $a (\cdot, \cdot \; ; \delta) \in S(1)$ uniformly in $0 < \delta \le 1$. By Proposition \ref{prop:commutator_sufficient}, this means that $H_\delta$ satisfies Assumption \ref{assumption:commutator} for all $m \in \mathbb{N}_0$.
To obtain the symbols $b_{0p}$ and $b_{1p}$ for $p \in \mathbb{N}$, we use the linearity of $a(X, \xi; \delta)$ in $\delta$ to verify that $a$ satisfies the conditions of Proposition \ref{prop:taylor}, with $a (X,K;0) = 0$ for all $X$. It follows that $a$ has the decomposition \eqref{eq:decomposition_higher_order} for some $b_{1p} (X, \zeta; \delta) \in S (\aver{\zeta}^{p+1})$ uniformly in $\delta$, where $\rate = 1$, $E = 0$, and
\begin{align}\label{eq:expansion_p}
    \delta b_{0p} (X,\zeta;\delta) = \sum_{j=1}^p \frac{\delta^j}{j!} \sum_{i_1, \cdots, i_j = 1}^{3} \bar{\zeta}_{i_1} \dots \bar{\zeta}_{i_j} \partial_{i_1, \dots, i_j} a(X,\bar{K}),\qquad b_{0p}\in S^p,
\end{align}
with $\partial_{\ell} a(X,\xi;\delta) := \partial_{\xi_\ell} a(X,\xi,\delta)$ if $1 \le \ell \le 2$ and $\partial_{3} a(X,\xi;\delta) := \partial_\delta a(X,\xi;\delta)$. 

Using that
$\nabla_\xi \sum_{j=1}^3 e^{i \xi \cdot \la_j}\vert_{\xi=K} = \frac{v\sqrt{3}}{2} (i,1)$ and $\sin (\lb_j \cdot K) = -\sqrt{3}/2$ for each $j$, we find that
\begin{align}\label{eq:b01}
    \delta b_0 (X,\zeta;\delta) := \delta b_{01} (X,\zeta;\delta) = \frac{\sqrt{3} v}{2} t_1 (X) \delta (\zeta_2 \sigma_1 - \zeta_1 \sigma_2) +
    \delta \Big(M(X) - \frac{3\sqrt{3}}{2} t_2 (X) \Big) \sigma_3.
\end{align}
Since $t_1$ is bounded away from zero, the ellipticity condition \eqref{eq:ellipticty} holds. We have thus verified Assumption \ref{assumption:a} with $\rate = 1$ 
for the Haldane model \eqref{eq:Hdelta}-\eqref{eq:a0}.
Note that $b_{0}$ is independent of $\delta$.

To calculate the second order terms, we use that $$\nabla^2_\xi \sum_{j=1}^3 e^{i \xi \cdot \la_j}\vert_{\xi=K} = \frac{v^2}{4} \begin{pmatrix}
    -1 & -i\\
    -i & 1
\end{pmatrix}, \qquad \nabla_\xi \sum_{j=1}^3 \sin (\lb_j \cdot \xi)\vert_{\xi=K} = 0$$
to conclude that
\begin{align}\label{eq:b02}
    \delta b_{02} (X, \zeta;\delta) = \delta b_{0} (X, \zeta;\delta) + \frac{v^2}{8} t_1 (X) \delta^2 \begin{pmatrix}
        0 & (\zeta_2 - i\zeta_1)^2\\
        (\zeta_2 + i\zeta_1)^2 & 0
    \end{pmatrix},
\end{align}
where $\delta b_{0}$ is given by \eqref{eq:b01}. 
Observe that $\delta b_{02}$ is a polynomial in $\delta \zeta$, thus our assumption that $|t_1| \ge c$ implies that the ellipticity condition \eqref{eq:ellipticty_higher_order} holds. We conclude that the symbols $a, b_0, b_{02}$ also satisfy Assumption \ref{assumption:a_higher_order}. The macroscopic Hamiltonians up to second-order are 
\begin{align}\label{eq:macroHaldane}
    \begin{split}
    \bH & = \ow b_0 = \frac{\sqrt 3 v}4 \{ t_1(X) , D_2\sigma_1-D_1\sigma_2 \} + \Big(M(X) - \frac{3\sqrt{3}}{2} t_2 (X) \Big) \sigma_3,
    \\
    \bH_2 &= \ow b_{02} = \bH + \delta \frac{v^2}{32} \big( \{ \{t_1(X), D_2-iD_1\}, D_2-iD_1\} \sigma_+ + \{\{ t_1 (X), D_2+iD_1\}, D_2+iD_1\}\sigma_- \big),
    \end{split}
\end{align}
where $\{A,B\}=AB+BA$ and we recall that $\sigma_\pm=\frac12(\sigma_1\pm i\sigma_2)$. We recognize in $\bH$ a standard Dirac operator and in $\bH_2$ a modified Dirac operator with second-order terms. In the rest of the paper, we only present the symbols $b_{0j}(X,\zeta;\delta)$ rather than the operators $\bH_j=\ow b_{0j}$.

\begin{lemma}\label{lem:ellipb0p}
    The symbols $b_{0p}$ in \eqref{eq:expansion_p} are elliptic in $S^p$ for each $p\geq2$ in the sense that \eqref{eq:ellipticty_higher_order} holds.
\end{lemma}
The proof of this technical result is postponed to Appendix \ref{sec:Ab0p}. This result shows that Assumption \ref{assumption:a_higher_order} holds for any $p$, and thus Theorem \ref{thm:main} applies to the Haldane model to arbitrary order in powers of $\delta$ with the limiting differential model ${\bf H}_p=\ow b_{0p}$.

\subsection{Strain and twist effects in graphene} \label{subsec:shear}

We now consider generalizations of the above derivations when the lattice parameters of the model are also allowed to have macroscopic variations. In these settings, Assumption \ref{assumption:a} typically no longer holds and we need to verify the more challenging constraints in Assumption \ref{assumption:a2}. 

\paragraph{Strain effect.} 
Suppose now lattice directions depending on $X$ in the above Haldane model. We still define the Hamiltonian by \eqref{eq:Hdelta}-\eqref{eq:a0} with $|t_1 (X)| \ge c > 0$, only now the vectors $\la_j$ and $\lb_j$ are given by
\begin{align}\label{eq:shear}
    \la_j(X)=\la_{j,0}+\delta \la_{j,1}(X), \qquad
    \lb_j(X)=\lb_{j,0}+\delta \lb_{j,1}(X), \qquad
    \la_{j,1}, \lb_{j,1} \in C^\infty_b (\mathbb{R}^2),
\end{align}
with $\la_{1,0}=\frac13(v_1+v_2)$, $\la_{j,0} = \rot_{2\pi/3} \la_{j-1,0}$ for $j=2,3$, $\lb_{1,0}=v_1$, $\lb_{2,0}=v_2-v_1$, and $\lb_{3,0}=-(\lb_1+\lb_2)$.
The presence of the $X$-dependent exponential terms $e^{\pm i \xi \cdot \la_j (X)}$ in the symbol $a$ means that the latter is not in $S(1)$ and thus Assumption \ref{assumption:a} does not apply, even in the seemingly innocuous setting of small perturbations of order $\delta$. We instead verify Assumption \ref{assumption:a2}.
Since the $X$-dependence of the $\la_j$ is of $O(\delta)$, we take $K = -\frac{4\pi}{3v} (0,1)$ as before and $B \equiv 0$ so that $\dpx = K$ is independent of $X$.

It is clear that $a \in \es$, as the functions $\xi \cdot \la_j (X)$ are real-valued and belong to $S^1$.
One can then follow the derivation of \eqref{eq:b01} to show that
\begin{align*}
    b_{0} (X,\zeta;\delta) = \frac{\sqrt{3} v}{2} t_1 (X) (\zeta_2 \sigma_1 - \zeta_1 \sigma_2) +
    \left(M(X) - \frac{3\sqrt{3}}{2} t_2 (X) \right) \sigma_3
    + t_1 (X) 
    \begin{pmatrix}
    0 & i K \cdot \alpha (X)\\
    -i K \cdot \overline{\alpha (X)} & 0
    \end{pmatrix}
\end{align*}
satisfies the conditions in Assumption \ref{assumption:a2} with $\rate = 1$ and any $C_0, \eta > 0$,
where 
\begin{equation}\label{eq:alphastrain}
    \alpha (X) := \la_{1,1} (X) + e^{-i2\pi/3} \la_{2,1} (X) + e^{i2\pi/3} \la_{3,1} (X).
\end{equation}
We retrieve the well-known result (see, e.g. \cite{barsukova2024direct, guglielmon2021landau, guinea2010energy}) that strain generates an arbitrary synthetic magnetic potential modeled by $\alpha(X)$ in \eqref{eq:alphastrain}.
Note that $b_{0}$ depends only on the leading-order terms of $\delta a^1 (\delta x, \xi)$ from \eqref{eq:Hdelta}, and thus is independent of the $\lb_{j,1}$. 
That the residual symbol $b_{1}$ in \eqref{eq:decomposition2} belongs to $\cs^2$ is a direct consequence of Taylor's theorem. Indeed, with $\bar{\zeta} := (\zeta; \delta) \in \mathbb{R}^{3}$, the explicit formula for the remainder gives
\begin{align}\label{eq:rem}
    b_{1} (X, \zeta; \delta) = \bar{\zeta} \cdot \int_{0}^1  (1-s) \nabla^2 a(X, K+\delta s \zeta; \delta s) ds \bar{\zeta}
\end{align}
element-wise, where $\nabla^2 a \in \cs^0$ is the Hessian of the function $(\xi; \delta) \mapsto a (X, \xi; \delta)$. 

We next refer to 
Lemma \ref{lem:varK} below, which shows that $H_\delta = \ow a_\delta$ defined by \eqref{eq:Hdelta}-\eqref{eq:a0} and \eqref{eq:shear} satisfies Assumption \ref{assumption:commutator2} for all $m \in \mathbb{N}_0$. We conclude that Theorem \ref{thm:main2} applies to the tight-binding and effective Hamiltonians $H_\delta$ and $\bH = \ow b_0$ for any $m \in \mathbb{N}_0$ with $\rate = 1$.

\begin{remark}\label{remark:no_lattice}
    The model defined by \eqref{eq:shear} does not necessarily have a natural lattice interpretation without additional constraints on the $O(\delta)$ corrections $\la_{j,1}$ and $\lb_{j,1}$. Indeed, the operators $P_\delta$ defined in \eqref{eq:P_delta_def} for different choices of $\sigma$ and $\la$ do not necessarily commute, 
    meaning that powers of the TB Hamiltonian $H_\delta^N$ could contain near duplicates of the same translation operator that are $O(\delta)$ apart. However, recalling \eqref{eq:Pf} and using that $y(x) = x + \la (\delta x) + O(\delta)$ there, 
    one can choose the $\la_{j,1}$ and $\lb_{j,1}$ such that  
    the associated translation operators map $\tilde \Lambda$ to itself, with $\tilde \Lambda$ a slowly-varying deformation of the lattice shown in Figure \ref{fig:geomgraphene}.
\end{remark}


\paragraph{Twist effect.}

Consider now the above Haldane model with vectors $\la_j$ and $\lb_j$ respectively replaced by $\rot_\theta \la_j$ and $\rot_\theta \lb_j$, with $\rot_\theta$ the (counter-clockwise) rotation matrix by some angle $\theta$. This corresponds to an overall twist of the underlying lattice.
The symbol $a$ from \eqref{eq:Hdelta} gets replaced by 
\begin{align}\label{eq:a_theta}
    a^\theta(\delta x,\xi;\delta) = a^{0,\theta} (\delta x,\xi) + \delta a^{1,\theta} (\delta x, \xi),
\end{align}
where
\begin{align}\label{eq:a0_theta}
\begin{split}
    a^{0,\theta} (X, \xi) &= t_1 (X) \begin{pmatrix}
        0 & e^{i \xi \cdot R_\theta \la_1} + e^{i\xi \cdot R_\theta \la_2} + e^{i\xi \cdot R_\theta \la_3}\\
        e^{-i \xi \cdot R_\theta\la_1} + e^{-i\xi \cdot R_\theta\la_2} + e^{-i\xi \cdot R_\theta\la_3} & 0 
    \end{pmatrix}
    = a^0 (X, R_{-\theta} \xi),
    \\
    a^{1,\theta} (X,\xi) &= \Big(M(X) + t_2 (X) \sum_{j=1}^3 \sin ((R_\theta \lb_j) \cdot \xi)\Big) \sigma_3 = a^1 (X, R_{-\theta} \xi),
    \end{split}
\end{align}
with the $\la_j$ and $\lb_j$ given by \eqref{eq:shear} to account for strain.
We will now construct an effective symbol $b_{0} = b_{0}^\theta$ that satisfies Assumption \ref{assumption:a2} with $B \equiv 0$ as above.
The relevant Dirac point is now $K^\theta := R_{\theta} K$, where $K = -\frac{4\pi}{3v} (0,1)$ was the Dirac point for the untwisted lattice. With $\overline{K^\theta} := (K^\theta, 0)$, 
it follows that
\begin{align*}
    \nabla_\xi a^{0,\theta} (X,\overline{K^\theta})= R_\theta \nabla_\xi a^0 (X, \overline{K}) = \frac{v \sqrt{3}}{2} t_1 (X) \begin{pmatrix}
        (0,0) & R_\theta (i,1)\\
        R_\theta (-i,1) & (0,0)
    \end{pmatrix}, 
\end{align*}
so that the leading-order symbol of the effective operator is
\begin{align}\label{eq:b_theta}
    b_{0}^\theta (X,\zeta;\delta) = \frac{\sqrt{3} v}{2} t_1 (X) \zeta \cdot R_{\theta + \pi/2} \sigma +
    \left(M(X) - \frac{3\sqrt{3}}{2} t_2 (X) \right) \sigma_3
    + t_1 (X) 
    \begin{pmatrix}
    0 & i K \cdot \alpha (X)\\
    -i K \cdot \overline{\alpha (X)} & 0
    \end{pmatrix}
    ,
\end{align}
where $\alpha(X)$ is still defined in \eqref{eq:alphastrain} and $\sigma := (\sigma_1, \sigma_2)$ is a vector of Pauli matrices, and the rotation matrix $R_{\theta + \pi/2}$ is understood to act on that vector.
As before, the symbol $b_{0}^\theta$ satisfies Assumption \ref{assumption:a2} with $\rate = 1$ and any $C_0, \eta > 0$, where the ellipticity condition \eqref{eq:ellipticty2} again follows from the assumption that $t_1$ is bounded away from zero. Note that the $\theta$-dependence of $b^\theta_{0}$ restricted to the term $\zeta \cdot R_{\theta + \pi/2} \sigma$, since the Dirac point and vectors $\la_{j,1}$ all get rotated by the same angle. 

By Lemma \ref{lem:varK} 
below, we obtain that $a_\delta (x,\xi) = a^\theta(\delta x,\xi;\delta)$
satisfies Assumption \ref{assumption:commutator2} for all $m \in \mathbb{N}_0$, meaning that Theorem \ref{thm:main2} applies (with $\rate = 1$ and any $m \in \mathbb{N}_0$). 

\paragraph{Small twist angle.} If we assume that the twist angle $\theta$ above satisfies $\theta = \beta \delta + O(\delta^2)$ for some $\beta \ne 0$, then we could also keep the untwisted $K = -\frac{4\pi}{3v} (0,1)$ as the degenerate point for the twisted lattice. Following the derivation of \eqref{eq:b01}, the effective symbol $b_{0}$ for the twisted Haldane model becomes
\begin{align}\label{eq:small_twist}
\begin{split}
    b_{0} (X,\zeta;\delta) = \frac{\sqrt{3} v}{2} t_1 (X) (\zeta_2 \sigma_1 - \zeta_1 \sigma_2) &+
    \Big(M(X) - \frac{3\sqrt{3}}{2} t_2 (X) \Big) \sigma_3 \\
    &+ t_1 (X) 
    \begin{pmatrix}
    0 & i K \cdot \alpha (X)\\
    -i K \cdot \overline{\alpha (X)} & 0
    \end{pmatrix}+ \partial_\delta a^{0,\theta} (X,K) \vert_{\delta = 0},
    \end{split}
\end{align}
with $a^{0,\theta}$ given by \eqref{eq:a0_theta}. The third term on the above right-hand side accounts for the fact that $a^{0,\theta} (X,K)$ is nonzero for $\delta > 0$ (due to the twist). Note that $a^{1,\theta}$ does not need a twist correction since it comes with a factor of $\delta$ in \eqref{eq:a_theta}. With $f_j (\delta) := e^{-i K \cdot \rot_\theta \la_j}$, we find that $f_j'(0) = -i \beta K \cdot \rot_{\pi/2} \la_j e^{-i K \cdot \la_j}$. Using the relations
\begin{align*}
    K \cdot \la_1 &= 0, & K \cdot \la_2 &= -\frac{2\pi}{3}, & K \cdot \la_3 &= \frac{2\pi}{3},\\
    K \cdot \rot_{\pi/2} \la_1 &= -\frac{4\pi}{3\sqrt{3}}, & K \cdot \rot_{\pi/2} \la_2 &= \frac{2\pi}{3 \sqrt{3}}, & K \cdot \rot_{\pi/2} \la_3 &= \frac{2\pi}{3 \sqrt{3}},
\end{align*}
we find that
$
    \sum_{j=1}^3 f_j'(0) = \frac{i2\pi}{\sqrt{3}} \beta,
$
and thus
\begin{align}\label{eq:twist_corr}
    \partial_\delta a^{0,\theta} (X,K) \vert_{\delta = 0} = \frac{2\pi}{\sqrt{3}} \beta t_1 (X) \sigma_2.
\end{align}

\paragraph{Expansion about other Dirac points.}
One can similarly derive an effective operator in the vicinity of the Dirac point $K' := -K$. Using that $\nabla_\xi \sum_{j=1}^3 e^{i \xi \cdot \la_j}\vert_{\xi=K'} = \frac{v\sqrt{3}}{2} (i,-1)$ and $\sin (\lb_j \cdot K') = 3\sqrt{3}/2$ for each $j$, the symbol in \eqref{eq:b_theta} becomes
\begin{align*}
    (b_{0}^{\theta})' (X,\zeta;\delta) = \frac{\sqrt{3} v}{2} t_1 (X) \zeta \cdot R_{\theta} \check{\sigma} +
    \left(M(X) + \frac{3\sqrt{3}}{2} t_2 (X) \right) \sigma_3
    + t_1 (X) 
    \begin{pmatrix}
    0 & i K' \cdot \alpha (X)\\
    -i K' \cdot \overline{\alpha (X)} & 0
    \end{pmatrix},
\end{align*}
where $\check\sigma := (\sigma_2, \sigma_1)$. For all other Dirac points $\tilde{K} \in K + \Lambda^*$ and $\tilde{K}' \in K' + \Lambda^*$ (with $\Lambda^*$ the reciprocal lattice), the leading order symbols are, respectively
\begin{align*}
    \tilde{b}_{0}^\theta (X,\zeta;\delta) &= \frac{\sqrt{3} v}{2} t_1 (X) \zeta \cdot R_{\theta + \pi/2} \sigma +
    \left(M(X) - \frac{3\sqrt{3}}{2} t_2 (X) \right) \sigma_3
    + t_1 (X) 
    \begin{pmatrix}
    0 & i \tilde{K} \cdot \alpha (X)\\
    -i \tilde{K} \cdot \overline{\alpha (X)} & 0
    \end{pmatrix}
    \\
    (\tilde{b}_{0}^{\theta})' (X,\zeta;\delta) &= \frac{\sqrt{3} v}{2} t_1 (X) \zeta \cdot R_{\theta} \check{\sigma} +
    \left(M(X) + \frac{3\sqrt{3}}{2} t_2 (X) \right) \sigma_3
    + t_1 (X) 
    \begin{pmatrix}
    0 & i \tilde{K}' \cdot \alpha (X)\\
    -i \tilde{K}' \cdot \overline{\alpha (X)} & 0
    \end{pmatrix}.
\end{align*}

\paragraph{Position-dependent degenerate points.}
We extend our analysis of the 
Haldane model \eqref{eq:Hdelta}-\eqref{eq:a0} to more general position-dependent lattice vectors
\begin{align}\label{eq:a_X_def}
\la_j(X)=\la_{j,0}(X)+\delta \la_{j,1}(X), \qquad
    \lb_j(X)=\lb_{j,0}(X)+\delta \lb_{j,1}(X), \qquad
    \la_{j,k}, \lb_{j,k} \in C^\infty_b (\mathbb{R}^2).
\end{align}
\tm{We are motivated by applications such as bent graphene nanoribbons, effectively making the local $K$ point position dependent \cite{Geim2010, Pavel2017, Koskinen2012}. This can be either a purely in-plane deformation when the graphene is fixed by a substrate, or a suspended graphene material with out-of-plane dependence built into the hopping functions}.
The above generalization of \eqref{eq:shear}, where $\la_{j,0}$ and $\lb_{j,0}$ are now allowed to depend on $X$, is technically significantly more challenging. We showed in the preceding section that variations in the first-order terms $\la_{j,1}(X)$ and $\lb_{j,1}(X)$ already led to leading-order effects in the macroscopic models. It should therefore not be surprising that variations in the leading terms $\la_{j,0}(X)$ and $\lb_{j,0}(X)$ need to satisfy suitable constraints. The main constraint is that the effect on the $K$ point be encoded by an irrotational field. \tbr{Some obstacles to deriving a continuum model in settings such as \eqref{eq:a_X_def} are already highlighted in \cite{roberts2022curved}.}

Assume $\lv_2$, $\la_{j,0}$ and $\lb_{j,0}$ are all determined by a vector $\lv_1 (X) = \lv_1 \in C^\infty_b (\mathbb{R}^2)$ as  
\begin{align}\label{eq:relations}
    \begin{split}
    \lv_2 &:= \rot_{-\pi/3} \lv_1, \quad \la_{1,0}:= \frac{1}{3} (\lv_1 + \lv_2), \quad
    \la_{2,0} := \rot_{2\pi/3} \la_{1,0}, \quad \la_{3,0} := \rot_{-2\pi/3} \la_{1,0}, \\
    \lb_{1,0} &:= \lv_1, \quad \lb_{2,0} := \lv_2 - \lv_1, \qquad \lb_{3,0} := -(\lb_{1,0} + \lb_{2,0}),
    \end{split}
\end{align}
where the $X$-dependence of the above vectors is implied. These relations are consistent with the original model \eqref{eq:Haldanemodel}, only now the vector $\lv_1$ depends on $X$. Assume that $|\lv_1 (X)| \ge c > 0$ is bounded away from zero.
Our previous derivations then imply that the point
\begin{align}\label{eq:dpx_def}
    \dpx (X) := \frac{4\pi}{3 |\lv_1 (X)|^2} \rot_{-2\pi/3} \lv_1 (X)
\end{align}
and symbol
\begin{align*}
    b_{0} (X,\zeta;\delta) &= \frac{\sqrt{3}|\lv_1 (X)|}{2} t_1 (X) \zeta \cdot \rot_{\theta (X) + \pi/2} \sigma +
    \left(M(X) - \frac{3\sqrt{3}}{2} t_2 (X) \right) \sigma_3
    \\&\hspace{4cm} 
    + t_1 (X) 
    \begin{pmatrix}
    0 & i \dpx (X) \cdot \alpha (X)\\
    -i \dpx (X) \cdot \overline{\alpha (X)} & 0
    \end{pmatrix}
\end{align*}
satisfy the decomposition \eqref{eq:decomposition2} with $E=0$ and $\rate = 1$ for some $b_{1} \in \cs^2$, 
where $\theta (X) := \arg \lv_1 (X) - \pi/6$ is the local twist angle (relative to the original model \eqref{eq:vs} for which $\arg \lv_1 = \pi/6$) at the point $X$, and $\alpha (X)$ is given by \eqref{eq:alphastrain} as before.
Note that $a \in \es$ and $b_{0} \in S^1$ as required, and the decomposition \eqref{eq:decomposition2} holds for any choice of $C_0, \eta > 0$.
Using that
\begin{align*}
    |\lv_1 (X)| \rot_{\arg \lv_1 (X)} = \begin{pmatrix}
        \lv_1^1 (X) & -\lv_1^2 (X)\\
        \lv_1^2 (X) & \lv_1^1 (X)
    \end{pmatrix}, \qquad 
    (\lv^1_1 (X), \lv^2_1 (X)) := \lv_1 (X),
\end{align*}
the symbol $b_{0}$ simplifies to
\begin{align}\label{eq:b0_g}
\begin{split}
    b_{0} (X,\zeta; \delta) = \sum_{j,k=1}^2\zeta_j g_{jk} (X) \sigma_k &+
    \left(M(X) - \frac{3\sqrt{3}}{2} t_2 (X) \right) \sigma_3\\
    &+ t_1 (X) 
    \begin{pmatrix}
    0 & i \dpx (X) \cdot \alpha (X)\\
    -i \dpx (X) \cdot \overline{\alpha (X)} & 0
    \end{pmatrix},
\end{split}
\end{align}
where the matrix $g_{jk}$ is defined by
\begin{align}\label{eq:g}
    g (X) := \frac{\sqrt{3}}{2}t_1 (X) \rot_{\pi/3} \begin{pmatrix}
        \lv_1^1 (X) & -\lv_1^2 (X)\\
        \lv_1^2 (X) & \lv_1^1 (X)
    \end{pmatrix}.
\end{align}
Since $|t_1 (X) \lv_1 (X)| \ge c > 0$ is bounded away from zero (by assumption), standard properties of the Pauli matrices imply that the ellipticity condition \eqref{eq:ellipticty2} is satisfied.

Recall Assumption \ref{assumption:a2}, which requires the degenerate point $\dpx (X)$ to be a gradient field. As such, 
let $B \in C^\infty_b (\mathbb{R}^2)$ such that $\nabla B (X) \ne -K$ for all $X$, and take
\begin{align}\label{eq:v1_X_def}
    \vp (X) := K \cdot X + B(X), \qquad \lv_1 (X) := \frac{4\pi/3}{|\nabla \vp (X)|^2} \rot_{2\pi/3} \nabla \vp (X).
\end{align}
Under this assumption on $v_1(X)$, then \eqref{eq:dpx_def} implies that $\dpx (X) = \nabla \vp (X)$ as desired. 

In order for Theorem \ref{thm:main2} to apply, it remains to verify Assumption \ref{assumption:commutator2} which we do with the following
\begin{lemma}\label{lem:varK}
    Let $a_\delta$ be as in \eqref{eq:Hdelta} with lattice vectors satisfying \eqref{eq:a_X_def}-\eqref{eq:relations}. Then Assumption \ref{assumption:commutator2} holds for all $m \in \mathbb{N}_0$. 
\end{lemma}
The proof is postponed to Appendix \ref{sec:varK}. As was the case when the $\la_{j,0}$ and $\lb_{j,0}$ were independent of $X$, additional assumptions on the $\la_{j,1}$ and $\lb_{j,1}$ would be required for the model \eqref{eq:a_X_def}-\eqref{eq:relations}-\eqref{eq:v1_X_def} to respect a deformed-lattice structure; see Remark \ref{remark:no_lattice}.

\subsection{Commensurate multilayer graphene models.} \label{sec:multilayer}

The calculations of section \ref{sec:Haldane} extend to untwisted multilayer graphene models with no difficulty. Consider for instance the following $n$ layer model:
\begin{equation}\label{eq:multilayer}
  H_n = \Omega_n\otimes I_2 + I_n\otimes H_1 + \Gamma_n,
\end{equation}
where $\Omega_n$ implements a different gate potential on each layer, for instance
\begin{equation}\label{eq:Omegan} \Omega_{n} = {\rm Diag}((-n+1)\omega,(-n+3)\omega,\ldots,(n-1)\omega)
\end{equation}
with $2\omega$ the difference of potential between each successive layers, where $H_1$ is a single-layer tight-binding model of graphene, and where $\Gamma_n$ implements a local coupling between layers. 

When $\Omega_n$ and $\Gamma_n$ are of order $\delta$, then $2n$ branches of spectrum coexist in the vicinity of $(K,E)$ for $E=0$ and $K$ a Dirac point of the single-layer model as described in the preceding section. We may then consider a macroscopic description for the rescaled Hamiltonian:
\begin{equation}\label{eq:Hndelta}
  H_{n\delta} = \delta\Omega_n(\delta x)\otimes I_2 + I_n\otimes H_{1\delta} + \delta\Gamma_n(\delta x).
\end{equation}
Here, $\Omega_n(X)$ implements a slowly varying gate potential on each layer while $\Gamma_n(X)$ implements a slowly varying interlayer coupling. Finally, $H_{1\delta}=H_\delta=\ow a_\delta$ is the single-layer Hamiltonian described in \eqref{eq:a0}.

Following the derivation of the preceding section, we obtain that the macroscopic Hamiltonian corresponding to $H_{n\delta}$ is
\begin{equation}\label{eq:macromulti}
    \bH_{np} = \Omega_n(X)\otimes I_2 + I_n \otimes \bH_{p} + \Gamma_n(X),
\end{equation}
where $\bH_{p}$ is the macroscopic model in \eqref{eq:macroHaldane} for single layer graphene.

\paragraph{Examples of bilayer stacking.}
The AB stacking is defined as follows. We assume the two layers untwisted and with a relative shift of $\la_1$.  More precisely, 
we assume the $B$ sites in layer $1$ and $A$ sites in layer $2$ have the positions $\cR := \{n_1 v_1 + n_2 v_2 : (n_1, n_2) \in \mathbb{Z}^2\}$. The $A$ sites in layer $1$ are located at $\cR - \la_1$, and the $B$ sites in layer $2$ at $\cR + \la_1$. We assume the two layers are coupled only through the interactions of the $(B,1)$ and $(A,2)$ sites. With $H_1$ defined by \eqref{eq:H1} and $\sigma_\pm := \frac{1}{2} (\sigma_1 \pm i \sigma_2)$, the coupling term reads
\begin{align*}
    \Gamma_2 =\Gamma_{AB}= \gamma \big(\sigma_+ \otimes \sigma_- + \sigma_- \otimes \sigma_+\big).
\end{align*}
If instead the two layers have a relative shift of $-\la_1$ so that they are coupled via $(B,2)$ and $(A,1)$ interactions, the Hamiltonian is
\begin{align*}
    \Gamma_2 =\Gamma_{BA} = \gamma (\sigma_+ \otimes \sigma_+ + \sigma_- \otimes \sigma_-).
\end{align*}
The work in \cite{bal2023mathematical} analyzes the macroscopic model arising from a domain transition between AB and BA stackings. The corresponding tight-binding model is of the form
\begin{align*}
    H_\delta = \delta \omega (\delta x) \sigma_3 \otimes I_2 + I_2 \otimes H_1 + \delta \chi (\delta x)\Gamma_{AB}+ \delta (1- \chi (\delta x)) \gamma \Gamma_{BA},
\end{align*}
where $\omega$ and $\chi$ are switch functions with $\omega \in \fs (\omega_-, \omega_+)$ and $\chi \in \fs (0,1)$ for some $\omega_\pm \in \mathbb{R}$. Here, $\fs(a,b)$ is the set of smooth functions on $\Rm$ equal to $a$ near $-\infty$ and equal to $b$ near $+\infty$. The corresponding macroscopic Hamiltonian, analyzed in detail in \cite{bal2023mathematical}, takes the following form to leading order:
\[ \bH_{2,1} = \omega(X)\sigma_3\otimes I_2 +  I_2 \otimes \bH + \chi (X)\Gamma_{AB} + (1- \chi (X)) \Gamma_{BA}.\]

The above stacking model easily generalizes to multilayered settings with larger values of the number of layers $n$ (referred to as rhombohedral graphene \cite{han2024large,min2008electronic}). It is also straightforward to show that the hypotheses leading to Theorem \ref{thm:main} are satisfied by both the TB and continuous models. Twist and shear effects may also be considered as we did for the Haldane model, although we do not pursue this here.





\subsection{Twisted bilayer graphene}\label{subsec:tbg}
In this section, we demonstrate how to apply our theory to moir\'e materials by focusing on the setting of twisted bilayer graphene, which is mathematically more challenging than the untwisted case we just considered. We show that the Bistritzer-MacDonald model \cite{bistritzer2011moire} for twisted bilayer graphene emerges as an effective PDE model describing the Schr\"odinger dynamics associated with a symbol {\em almost} satisfying Assumption \ref{assumption:a}; see also \cite{watson2023bistritzer} for an alternative derivation.

The main result of this section is Proposition \ref{prop:TBG} below, which establishes that the tight-binding (TB) dynamics for twisted bilayer graphene are accurately approximated by a first-order continuum Hamiltonian $\bH = \ow b_0$. 
The TB Hamiltonian $H_\delta = \ow a_\delta$ is a self-adjoint operator on $L^2 (\mathbb{R}^2; \mathbb{C}^4)$ given by
\begin{align*}
    H_\delta = \begin{pmatrix}
        H^{11}_\delta & \Hperp_\delta \vspace{0.2cm} \\
        (\Hperp_\delta)^* & H^{22}_\delta
    \end{pmatrix},
\end{align*}
where each $H_\delta^{ij} = \ow a_\delta^{ij}$ is a bounded operator on $L^2(\mathbb{R}^2; \mathbb{C}^2)$, with the symbols of the familiar form $a_\delta^{ij} (x, \xi) = a_{ij} (\delta x, \xi; \delta)$. 
Let $-\pi < \theta \le \pi$ and
define $a_{11} := a^{-\theta/2}$ and $a_{22} := a^{\theta/2}$ with $a^{\pm \theta/2}$ defined by \eqref{eq:a_theta}-\eqref{eq:a0_theta}. We assume that $\la_{j1}$ and $\lb_{j1}$ in \eqref{eq:shear} all vanish, so that the vectors $\la_j$ and $\lb_j$ are constant (see section \ref{sec:Haldane} for their definitions). Thus the operator $H_\delta^{jj}$ corresponding to graphene sheet $j$ is the Haldane Hamiltonian with small twist $(-1)^j \theta/2$ and no shear; see section \ref{sec:Haldane}. As before, we assume that the functions $M, t_1, t_2$ from \eqref{eq:a0_theta} all belong to $C^\infty_b (\mathbb{R}^2)$ with $|t_1 (X)| \ge c > 0$ bounded away from zero.

We fix $\beta \ne 0$ and make the important assumption that the {\em twist angle is small} and scales as
\begin{align}\label{assumption:theta}
    \theta = 2 \sin^{-1} (\beta \delta/2), \qquad 0 < \delta < \delta_0.
\end{align}

We next construct the symbol
\begin{align*}
    a_{12} = \begin{pmatrix}
        a_{12}^{AA} & a_{12}^{AB}\vspace{0.2cm}\\
        a_{12}^{BA} & a_{12}^{BB}
    \end{pmatrix}
\end{align*}
corresponding to the interlayer coupling Hamiltonian $H_\delta^{12}$. Set $\theta_j := (-1)^j \theta/2$ and define the Bravais lattice for each layer $j=1,2$ by $\lattice_j^0 := \rot_{\theta_j} \lattice$, with $\lattice = \mathbb{Z}v_1 + \mathbb{Z} v_2$ the Bravais lattice from section \ref{sec:Haldane}.
Since these lattices describe a discrete set of points but our model is on the continuum, we also define
\begin{align*}
    \lattice_1 (x) = \rot_{-\theta} (\lattice_2^0 + x), \qquad \lattice_2 (x) = \rot_\theta (\lattice_1^0 + x), \qquad x \in \mathbb{R}^2
\end{align*}
as the shifted lattices with respect to a given reference point. Observe that $\lattice_1 (x + r_2) = \lattice_1 (x)$ and $\lattice_2 (x + r_1) = \lattice_2 (x)$ for all $r_1 \in \lattice_1^0$ and $r_2 \in \lattice_2^0$, meaning that $\lattice_i$ is $\lattice_j^0$-periodic when $i \ne j$.
Fix $\delta_0, \gamma > 0$ and $0 < \tbgparam \le 1$. The interlayer coupling is modeled by the function 
$\hperp : \mathbb{R}^2 \to \mathbb{C}$. We assume a sufficiently rapid decay encoded in a Fourier transform
$\hathperp (k) := \int_{\mathbb{R}^2} e^{-ik \cdot x} \hperp (x)d x$ that satisfies
\begin{align}\label{eq:assumption_hperp}
    |\partial^\alpha \hathperp (k)| \le C_\alpha(\delta/\delta_0)^{\aver{k/\gamma}/\aver{K/\gamma} - (1-\tbgparam)|\alpha|}, \qquad k \in \mathbb{R}^2, \quad 0 < \delta < \delta_0
\end{align}
for all multi-indices $\alpha \in \mathbb{N}^2$ and some $\delta_0 > 0$, where $K := -\frac{4\pi}{3v} (0,1)$ as in Section \ref{sec:Haldane}.
\begin{remark}
    The decay assumption \eqref{eq:assumption_hperp} states that the two layers interact weakly as $\delta \to 0$, with an effective interaction strength of $O(\delta)$ when $|k| = |K|$. Physically, this corresponds to a large separation between the layers, typically of $O(|\log \delta|)$ in practice. We refer to \cite[Section 2.4]{watson2023bistritzer} and \cite[Section 2.3]{quinn2025} for more details and similar assumptions on interlayer hopping functions. In particular, \cite[Example 2.1]{quinn2025} considers
    \begin{align*}
        \hperp (x) = \frac{e^{-\gamma \sqrt{|x|^2 + \ell^2 (\delta)}}}{\sqrt{|x|^2 + \ell^2 (\delta)}}, \qquad \ell (\delta) := -\frac{1}{\sqrt{|K|^2 + \gamma^2}} \log \left(\delta/\delta_0\right), \qquad \delta_0 := \frac{2\pi}{\lambda_0 \sqrt{|K|^2 + \gamma^2}},
    \end{align*}
    for some fixed 
    $\gamma, \lambda_0 > 0$. The Fourier transform
    $\hathperp (k) = \frac{2\pi}{\sqrt{|k|^2+ \gamma^2}} \left( \delta/\delta_0 \right)^{\aver{k/\gamma}/\aver{K/\gamma}}$
    then satisfies \eqref{eq:assumption_hperp}, as
    \begin{align*}
        |\partial^\alpha \hathperp (k)| \le C_\alpha |\log \delta|^{|\alpha|}(\delta/\delta_0)^{\aver{k/\gamma}/\aver{K/\gamma}}, \qquad k \in \mathbb{R}^2, \quad {0 < \delta < \delta_0}.
    \end{align*}
\end{remark}
The interlayer Hamiltonian and its adjoint are defined by
\begin{align*}
    (H_\delta^{12} \psi)^\sigma (x) = \sum_{r_2 \in \lattice_2 (x)} \sum_{\sigma' \in \{A,B\}} \hperp^{\sigma \sigma'} (x-r_2) \psi^{\sigma'} (r_2), \quad ((H_\delta^{12})^* \psi)^\sigma (x) = \sum_{r_1 \in \lattice_1 (x)} \sum_{\sigma' \in \{A,B\}}\overline{\hperp^{\sigma \sigma'} (r_1 - x)} \psi^{\sigma'} (r_1), 
\end{align*}
where $x \in \mathbb{R}^2$ and $\sigma \in \{A,B\}$. 
The hopping functions are given by
\begin{align}\label{eq:os}
    \hperp^{\sigma \sigma'} (x) := \hperp (x + \os^{\sigma \sigma'}), \qquad \os^{\sigma \sigma'} := \os_1^\sigma - \os_2^{\sigma'}, \qquad \os_j^\sigma := \rot_{\theta_j} \os^\sigma, \qquad \os^A := (0,0), \qquad \os^B := \la_1,
\end{align}
where the Fourier transform of $\hperp$ satisfies \eqref{eq:assumption_hperp}. The operator $\Hperp_\delta$ is bounded on $L^2 (\mathbb{R}^2; \mathbb{C}^2)$, as established by the technical Lemma \ref{lemma:Hperp_bdd}.  
We then verify that $H_{\delta}^{12} = \ow a_\delta^{12}$, with the explicit symbol
\begin{align*}
    (a_\delta^{12})^{\sigma\sigma'} (x,\xi) = \ff \sum_{r \in \lattice} e^{-i \xi \cdot \varphi_r (x)} \hperp^{\sigma \sigma'} (\varphi_r (x)), \quad \ff := \frac{2}{1+\cos \theta}, \quad \varphi_r (x) := 2 (I + \rot_\theta)^{-1} ((I - \rot_\theta)x-\rot_{\theta/2} r).
\end{align*}
Indeed, with $a^{12}_\delta$ as above,
\begin{align*}
    (\ow a_\delta^{12} \psi)^\sigma (x) &= \frac{\ff}{(2\pi)^2}\int_{\mathbb{R}^4} e^{i \xi \cdot (x-y)} \sum_{r \in \lattice} \sum_{\sigma' \in \{A,B\}} e^{-i \xi \cdot \varphi_r ((x+y)/2)} \hperp^{\sigma \sigma'} (\varphi_r ((x+y)/2)) \psi^{\sigma'} (y) dy d\xi\\
    &=\ff\int_{\mathbb{R}^2} \sum_{r \in \lattice}\sum_{\sigma' \in \{A,B\}}\delta (x-y-\varphi_r ((x+y)/2)) \hperp^{\sigma \sigma'} (\varphi_r ((x+y)/2)) \psi^{\sigma'} (y) dy\\
    &=\ff\int_{\mathbb{R}^2} \sum_{r \in \lattice}\sum_{\sigma' \in \{A,B\}}\frac{\delta (y-\rot_\theta x - \rot_{\theta/2} r)}{\det  [I + (I+\rot_\theta)^{-1} (I - \rot_\theta)]} \hperp^{\sigma \sigma'} (\varphi_r ((x+y)/2)) \psi^{\sigma'} (y) dy,
\end{align*}
which by $\ff = \det  [I + (I+\rot_\theta)^{-1} (I - \rot_\theta)]$ and
$\varphi_r ((x + \rot_\theta x + \rot_{\theta/2} r)/2) = x - \rot_\theta - \rot_{\theta/2} r$ implies that
\begin{align*}
    (\ow a_\delta^{12} \psi)^\sigma (x) &= \sum_{r \in \lattice} \sum_{\sigma' \in \{A,B\}}\hperp^{\sigma \sigma'} (x - \rot_\theta x - \rot_{\theta/2} r) \psi^{\sigma'} (\rot_\theta x + \rot_{\theta/2} r)\\
    &= \sum_{r_2 \in \lattice_2 (x)} \sum_{\sigma' \in \{A,B\}} \hperp^{\sigma \sigma'} (x-r_2) \psi^{\sigma'} (r_2) = (H_\delta^{12} \psi)^\sigma (x).
\end{align*}
We now recast $a_\delta^{12}$ in order to apply Theorem \ref{thm:main_0}.
Writing $\hperp^{\sigma \sigma'}$ in terms of its Fourier transform, we obtain
\begin{align*}
    (a^{12}_\delta)^{\sigma \sigma'} (x,\xi) = \frac{\ff}{(2\pi)^2} \sum_{r \in \lattice} \int_{\mathbb{R}^2} e^{i (p-\xi) \cdot \varphi_r (x)} \hathperp^{\sigma \sigma'}(p) dp,
\end{align*}
which by the Poisson summation formula becomes
\begin{align*}
    (a^{12}_\delta)^{\sigma \sigma'} (x,\xi) = \frac{\ff}{|\Gamma|} \int_{\mathbb{R}^2} \sum_{q \in\lattice^*} \delta (2 \rot_{-\theta/2} (I + \rot_{-\theta})^{-1} (p-\xi) + q) e^{i (p-\xi) \cdot 2 (I+\rot_\theta)^{-1} (I - \rot_\theta) x} \hathperp^{\sigma \sigma'} (p) dp,
\end{align*}
with $|\Gamma|$ the area of the unit cell for $\lattice$, and $\lattice^* := \mathbb{Z} w_1 + \mathbb{Z} w_2$ the reciprocal lattice for one untwisted layer for $\begin{pmatrix}
     w_1 & w_2
\end{pmatrix} := 2\pi \begin{pmatrix}
    v_1 & v_2
\end{pmatrix}^{-1}$.
Evaluating the integral over $p$, this simplifies to
\begin{align*}
    (a^{12}_\delta)^{\sigma \sigma'} (x,\xi) &
    = \frac{1}{|\Gamma|} \sum_{q \in \lattice^*} e^{i q \cdot (\rot_{\theta/2} - \rot_{-\theta/2})x} \hathperp^{\sigma \sigma'} (\xi - \frac{1}{2} (\rot_{\theta/2} + \rot_{-\theta/2}) q).
\end{align*}
It follows from \eqref{assumption:theta} that
\begin{align}\nonumber
    (a^{12}_\delta)^{\sigma \sigma'} (x,\xi) &=\frac{1}{|\Gamma|} \sum_{q \in \lattice^*} e^{i \beta \delta q \cdot \rot_{\pi/2} x} \hathperp^{\sigma \sigma'} (\xi -(1 - \frac{1}{4}\beta^2 \delta^2)^{1/2}q),
    \\
    a_{12}^{\sigma \sigma'} (X, \xi; \delta) &= (a^{12}_\delta)^{\sigma \sigma'} (X/\delta,\xi) =\frac{1}{|\Gamma|} \sum_{q \in \lattice^*} e^{i \beta q \cdot \rot_{\pi/2} X} \hathperp^{\sigma \sigma'} (\xi -(1 - \frac{1}{4}\beta^2 \delta^2)^{1/2}q).\label{eq:a12_F}
\end{align}
We observe that $a_{12}$ does not quite satisfy Assumption \ref{assumption:a} since its derivatives with respect to $X$ are not bounded in $\xi$. Still, we will see that an approximation to $a_{12}$ does satisfy Assumption \ref{assumption:a}. Clearly, Assumption \ref{assumption:a2} does not hold either as $a_{12}$ does not belong to the class $\es$ of Definition \ref{def:es}.
However, we can 
still establish the validity of a first-order PDE model for twisted bilayer graphene at small twist angles.
\begin{proposition}\label{prop:TBG}
    Define $H_\delta = \ow a_\delta$ 
    as above, where $t_1 (X) \ge c > 0$ for all $X \in \mathbb{R}^2$, the interlayer hopping function $\hperp$ satisfies \eqref{eq:assumption_hperp}, and the relative twist angle $\theta$ satisfies \eqref{assumption:theta} for some $\beta \ne 0$. Set 
    \begin{align*}
        T_0 := \begin{pmatrix}
            1 & 1\\
            1 & 1
        \end{pmatrix} \qquad T_1 := \begin{pmatrix}
        1 & e^{i2\pi/3}\\
        e^{-i2\pi/3} & 1
    \end{pmatrix}
        \qquad T_2 := \begin{pmatrix}
        1 & e^{-i2\pi/3}\\
        e^{i2\pi/3} & 1
    \end{pmatrix}.
    \end{align*}
    Suppose that for $j=0,1,2$, there exist constants $\lambda_j \in \mathbb{C}$ such that $\hat{h} (\rot_{-2\pi j/3}K;\delta) = \delta \lambda_j$.
    Define
    \begin{align*}
        w_0 := (0,0), \qquad
        w_1 :=-\frac{4\pi\beta}{\sqrt{3}v} \left(\frac{\sqrt{3}}{2}, \frac{1}{2}\right), \qquad
        w_2 := \frac{4\pi\beta}{\sqrt{3}v} \left(-\frac{\sqrt{3}}{2}, \frac{1}{2}\right),
    \end{align*}
    and with $\tM := M - \frac{3\sqrt{3}}{2} t_2$, define the symbol $b_{0} (X,\zeta) = b_{0} \in S(\aver{\zeta})$ by 
    \begin{align*}
        b_{0} &:= \begin{pmatrix}
    b_{0}^{11} & b_{0}^{12}\vspace{0.2cm}\\
    (b_{0}^{12})^* & b_{0}^{22}
\end{pmatrix},\\
        b_{0}^{jj} (X,\zeta) &:= \frac{\sqrt{3} v}{2} t_1 (X) \zeta \cdot R_{\pi/2} \sigma +
    \tM (X) \sigma_3 + (-1)^j t_1 (X) \frac{\pi}{\sqrt{3}} \beta \sigma_2, \quad
    b_{0}^{12} (X) := \frac{1}{|\Gamma|}\sum_{j=0}^2 \lambda_j T_j e^{i w_j \cdot X}.
    \end{align*}
    Define $\varphi_\delta$ in terms of $H_\delta$ by \eqref{eq:microWP}, and take $\psi_\delta$ in terms of $b_{0}$ by \eqref{eq:macroH_0}-\eqref{eq:macro_0}-\eqref{eq:psi_def}.
    Then for any 
    \begin{align}\label{eq:s_eta}
        0 < \rate < \min \{\tbgparam, \; \aver{2K/\gamma}/\aver{K/\gamma} - 1 \}
    \end{align}
    and $N \ge 0$, there exists a constant $C > 0$ such that 
    \begin{align}\label{eq:convergence_tbg}
        \|(\psi_\delta-\varphi_\delta)(t, \cdot)\|_{H^N(\Rm^2; \mathbb{C}^4)} \leq C\delta^{1+\rate}t \left(1 + (\delta t)^N \right), \qquad t \ge 0, \quad 0 < \delta \le \delta_0/2.
    \end{align}
\end{proposition}
The proof of Proposition \ref{prop:TBG} is postponed to Appendix \ref{subsec:pf_TBG}. It consists of decomposing the TB symbol $a_\delta$ into two terms (one that satisfies Assumption \ref{assumption:a} and one that is negligible; see Lemma \ref{lemma:a12}) so that Theorem \ref{thm:main_0} can be applied. By including more terms in the symbol $b_0$ above and invoking Theorem \ref{thm:main}, the $O(\delta^\rate)$ convergence in \eqref{eq:convergence_tbg} over timescales $t \lesssim \delta^{-1}$ could be improved to any power of $\delta$. 
Such computationally tedious higher-order convergence for twisted bilayer graphene can be found in \cite{quinn2025}, and hence is not pursued here further.
\begin{remark}\label{remark:BM}
    The effective Hamiltonian $\bH = \ow b_0$ is an elliptic first-order differential operator (recall that $t_1$ is bounded away from zero by assumption).
    Besides the domain wall $\tM$, spatial dependence of $t_1$ and a permutation of the Pauli matrices in the $b_0^{jj}$, the only difference between our continuum model and the Bistritzer-MacDonald (BM) model \cite{bistritzer2011moire} comes from our wave-packet ansatz \eqref{eq:psi_def}. Indeed, we use the \emph{untwisted} degenerate point $K$ for each layer, whereas the BM model is derived from an ansatz of the form
    \begin{align}\label{eq:ansatz_j}
        \psi_\delta^j(t,x) = e^{i(K_j \cdot x-Et)} \delta^{\frac d2} \phi_j(\delta t,\delta x;\delta)
    \end{align}
    in layer $j$, with $K_j := \rot_{\theta_j} K$ the rotated Dirac point corresponding to layer $j$; see \cite{watson2023bistritzer} for more details. As a result of this difference, our model includes an extra constant term in the $b_0^{jj}$ (accounting for the fact that $a^{jj}_\delta (X, K) \ne 0$ for $\delta > 0$) and the vectors $w_j$ appearing in $b_0^{12}$ are shifted by $\beta \rot^\top_{\pi/2} K$, which is the leading-order approximation to $K_1 - K_2 = (\rot_{\theta_1} - \rot_{\theta_2}) K$. These two effective models are related by a unitary transformation since $\bH_{{\rm BM}} = U \bH U^*$ for $U$ the point-wise multiplication operator defined by
    \begin{align*}
        U := \begin{pmatrix}
            e^{-\frac{i\beta}{2} (\rot_{\pi/2}^\top K) \cdot x}I_2 & 0\\
            0 & e^{\frac{i\beta}{2} (\rot_{\pi/2}^\top K) \cdot x}I_2
        \end{pmatrix},
    \end{align*}
    and $\bH_{\rm BM} = \ow \tilde b_0$ is the modified BM model corresponding to the ansatz \eqref{eq:ansatz_j}, given explicitly by 
    \begin{align*}
        \tilde b_0 = \begin{pmatrix}
    \tilde b_{0}^{11} & \tilde b_{0}^{12}\vspace{0.2cm}\\
    (\tilde b_{0}^{12})^* & \tilde b_{0}^{22}
    \end{pmatrix}, \qquad 
    \tilde b_{0}^{jj} (X,\zeta) := \frac{\sqrt{3} v}{2} t_1 (X) \zeta \cdot R_{\pi/2} \sigma +
    \tM (X) \sigma_3, \qquad
    \tilde b_{0}^{12} (X) := \frac{1}{|\Gamma|}\sum_{j=0}^2 \lambda_j T_j e^{i \tilde w_j \cdot X},
    \end{align*}
    with $\tilde w_j := w_j - \beta \rot^\top_{\pi/2} K$. Observe that $\tilde w_0 = \frac{4\pi\beta}{3v} (1,0)$ and $\tilde w_j = \rot^\top_{2\pi/3} \tilde w_{j-1}$.
\end{remark}

%
\section{Numerical simulations}
\label{sec:num}
We present numerical simulations of propagating edge and bulk modes for a Bernal stacked gated bilayer graphene model described in section \ref{sec:multilayer} using the discrete (tight-binding) model~\eqref{eq:Hndelta} and the continuum model~\eqref{eq:macromulti} to illustrate and complement the analysis presented in previous sections. We choose $t_1(X) = 1$, and $t_2(X) = 0$, and hence our tight-binding Hamiltonian is
\begin{equation}\label{eq:tb_blg}
H_{\br \ell, \br'\ell'} = (-1)^\ell \omega m \left ( \frac{\delta}{\omega} \br \right )\delta_{\ell\ell'}\delta_{\br\br'} + \delta_{\ell\ell'}\delta_{\br-\br' \in\{\pm\mathfrak{a}_j\}_j} + \gamma \delta_{\ell \neq \ell'}\delta_{\br\br'}
\end{equation}
where $\br,\br'$ denote atomic positions in the multi-lattices of the two layers and $\ell,\ell'$ denote the layer indices. Throughout this section we set $\gamma = 0.121$, which is the proportion to $t_1 = 1$ for Bernal-stacked bilayer graphene calculated in \cite{McCann_2013}, $\delta = 0.01$ and $\omega = 0.25$. The mass function $m$ takes values in $[-1,1]$ describing a gating switch across a domain wall, and we choose a racetrack profile described by the function
\[
    m(x) = \tanh(r(x)) \qquad \text{with} \qquad r(x) = \begin{cases}
       \sqrt{(\vert x_1 \vert - \ell /2)^2+x_2^2} - w/2 & \text{ if } \vert x_1 \vert > \ell/2, \\
       \vert x_2 \vert - w/2 & \text{else.}
    \end{cases}
\]
Parameters $\ell$ and $w$ describe the length and radius of a racetrack forming the zero level set of the mass function, as seen in Figure~\ref{fig:racetrack}, and given in Table~\ref{tab:numparameters}.
 Note that the $\delta$-dependence in the Hamiltonian~\eqref{eq:tb_blg} does not conform precisely to Assumption~\ref{assumption:a} as we aim to present an example with realistic parameters, though the dependence is implicitly there under rescaling as $\omega$, $\delta/\omega$ and $\gamma$ have relatively small values. 
We will compare to the first and second order continuum models centered at valley $K$ given by
\begin{align}
    &\bH_1 =  \omega m\left ( \frac{\delta}{\omega}X \right )\delta_{\ell\ell'}\sigma_3\otimes I + I\otimes \frac{\sqrt{3}v}{2}(D_2\sigma_1 - D_1\sigma_2) + \gamma (\sigma_+\otimes \sigma_++\sigma_-\otimes \sigma_-)\\
    &\bH_2 = \bH_1 + \frac{v^2}{16}\bigl((D_2-iD_1)^2 I\otimes\sigma_+ + (D_2+iD_1)^2 I\otimes\sigma_-\bigr)
\end{align}
where recall $v$ is the length of the primitive lattice vectors.

\subsection{Spatial discretization and time-stepping methods}
Since our goal is to illustrate numerically the long-time behavior of the approximation error between the wave packet evolved according to the miscroscopic tight-binding model and the wave packet evolved according to the corresponding continuum model obtained by Taylor expansion, long-time accuracy of the time-stepping methods is important to minimize numerical error while keeping numerical costs reasonable for large domains.

\paragraph{Continuum model.}
Similar to the earlier work~\cite{bal2023mathematical}, computations are conducted on a finite rectangular domain of size $L_x \times L_y$ equipped with periodic boundary conditions. Localized solutions of the macroscopic Dirac equation are approximated in a pseudo-spectral approach using truncated symmetric Fourier series
\begin{equation}\label{eq:pseudospectral}
    {\boldsymbol \phi}(X,Y) \approx \sum_{k=-K_x}^{K_x} \sum_{k=-K_y}^{K_y} \widehat{\boldsymbol \phi}_{k,l} e^{2i\pi k X/L_x + 2i\pi l Y/L_y},
\end{equation}
where $\widehat{\boldsymbol \phi}_{k,l} \in \mathbb{C}^4$ and truncation parameters $K_x$ and $K_y$ are chosen appropriately to ensure convergence of the results. Derivation operators such as the monolayer Dirac terms are represented as diagonal matrices, and pointwise multiplication operators such as the gate potential $\Omega_2(X) \otimes I_2$ are applied efficiently using Fourier interpolation on a uniform real-space grid of size $3(K_x+1) \times 3(K_y+1)$ by using the discrete Fourier transform $\mathcal{F}$, avoiding aliasing of function products for sharp gate or interlayer coupling profiles $\Omega(X)$, $\Gamma_2(X)$ (at the cost of a small error regarding wave packet mass conservation).

Given that the time evolution of the macroscopic Dirac equation for constant-coefficient differential operators as well as multiplication operators can be computed exactly for the pseudo-spectral ansatz~\eqref{eq:pseudospectral}, we implement the 6-stage, 4-th order symplectic Runge-Kutta splitting approach proposed in~\cite{blanes2002practical}, allowing for relatively large time step $h_c$ compared for example to 2nd-order Strang splitting. Indeed, the numerical error of such time integrators scales with the commutator of the Dirac term with the gate and interlayer potential terms, hence is proportional to the small parameter $\delta$.
\begin{remark}
The resulting time integration is not quite symplectic due to the truncation error when applying multiplication operators, however, for smooth enough solutions this effect is negligible. 
\end{remark}
\paragraph{Tight-binding model.} We will focus numerical examples exclusively on periodic systems with mass terms that can vary in the center but are approximately constant on the boundary of a supercell compatible with the underlying crystal lattice, e.g., the racetrack profile described above. Hence our solutions can also be restricted to a finite cutout of the hexagonal multi-lattice 
\[
    \tilde \Lambda_{L_x,L_y} = \{ m v_1 + n v_2 + s \ \vert \  m,n \in \mathbb{Z}, \ \vert m \vert \leq L_x, \vert n \vert \leq L_y, \; s \in \{s_A,s_B\} \ \}
\]
given by~\eqref{eq:vs} and equipped with periodic boundary conditions.  In order to make the domain rectangular rather than a parallelogram, in practice we implement an augmented unit cell with four sites per cell, but for simplicity of presentation we will neglect this detail from the discussion as the only effect is to allow a rectangular domain cutout.
The corresponding Hamiltonian is a sparse matrix of size $8(2L_x+1)(2L_y+1) \times 8(2L_x+1)(2L_y+1)$. We employ the standard 4-th order Runge-Kutta method to integrate accurately the Schrödinger equation with time-step $h_{TB}$.

\paragraph{Error computation.}
Finally, at a given time $t \geq 0$ the above approaches yield numerical approximations to the the microscopic tight-binding solution $\psi^{L_x,L_y,h_{TB}}(t,\cdot)$ at the atomic positions $X_{m,n}^{\ell,\alpha} = m v_1 + n v_2 + \mathfrak{a}_1(\delta_{\alpha B}+\delta_{\ell 2})$, and to the macroscopic continuum solution $\phi^{K_x, K_y, h_c}(t,\cdot)$ in the form of the ansatz~\eqref{eq:pseudospectral}. We interpolate the macroscopic solution first to the uniform rectangular grid described above by a fast Fourier transform, and second to the lattice points using accurate cubic B-splines~\cite{interpol}. 
Numerical integration parameters (continuum and tight-binding time-steps, Fourier cutoffs $K_x$ and $K_y$, see Table~\ref{tab:numparameters})
are carefully chosen such that the model error (Theorem~\ref{thm:main}) dominates the computed mean square error.

\begin{table}
    \centering
    \begin{tabular}{|c|c|c|c|c|c|c|c|c|} \hline
      & $h_{TB}$ & $h_c$ & $K_x$ & $K_y$ & $L_x$ & $L_y$ & $\ell$ & $w$ \\ \hline 
      Edge modes, section~\ref{sec:edgestatenum}   & $1/8$ & $1/2$ & $128$ & $128$ & $577\sqrt{3} (\approx 1000)$ & $500$ & $500$ & $300$ \\ \hline
      Bulk mode, section~\ref{sec:bulkstatenum}   & $1/8$ & $1/2$ & $128$ & $128$ & $577\sqrt{3} (\approx 1000)$ & $1000$ & 0 & $500$ \\
        \hline
    \end{tabular}
    \caption{Numerical parameters for the simulations presented in paragraphs~\ref{sec:edgestatenum} and~\ref{sec:bulkstatenum}.}
    \label{tab:numparameters}
\end{table}
\[
    E_{K_x, K_y,L_x,L_y}^{h_c, h_{TB}}(t) = \left ( \sum_{\alpha \in \{A,B\}}\sum_{\ell=1}^2\sum_{m=-L_x}^{L_x} \sum_{n=-L_y}^{L_y} \vert \psi^{L_x,L_y,h_{TB}}_{\ell}(t,X_{m,n}^{\ell,\alpha}) - \phi^{K_x, K_y, h_c}_{\ell\alpha}(t,X_{m,n}^{\ell,\alpha})\vert^2 \right )^{1/2}.
\]
According to Theorem \ref{thm:main} and using Sobolev embedding, we expect for $t \lesssim \delta^{-1}$ to have
\begin{equation}
E_{p, K_x, K_y,L_x,L_y}^{h_c, h_{TB}}(t) = O(\delta^{p+1} t).
\end{equation}
We note that for this system $\rate = 1$ (cf. \eqref{eq:decomposition_higher_order} in Assumption \ref{assumption:a_higher_order}) and recall $p$ is the order of the continuum approximation $\bH_p$.

\subsection{Edge state propagation}\label{sec:edgestatenum}

\begin{figure}[t!]
    \centering
    \begin{subcaptionblock}{3.25in}
        \includegraphics[width=3.25in]{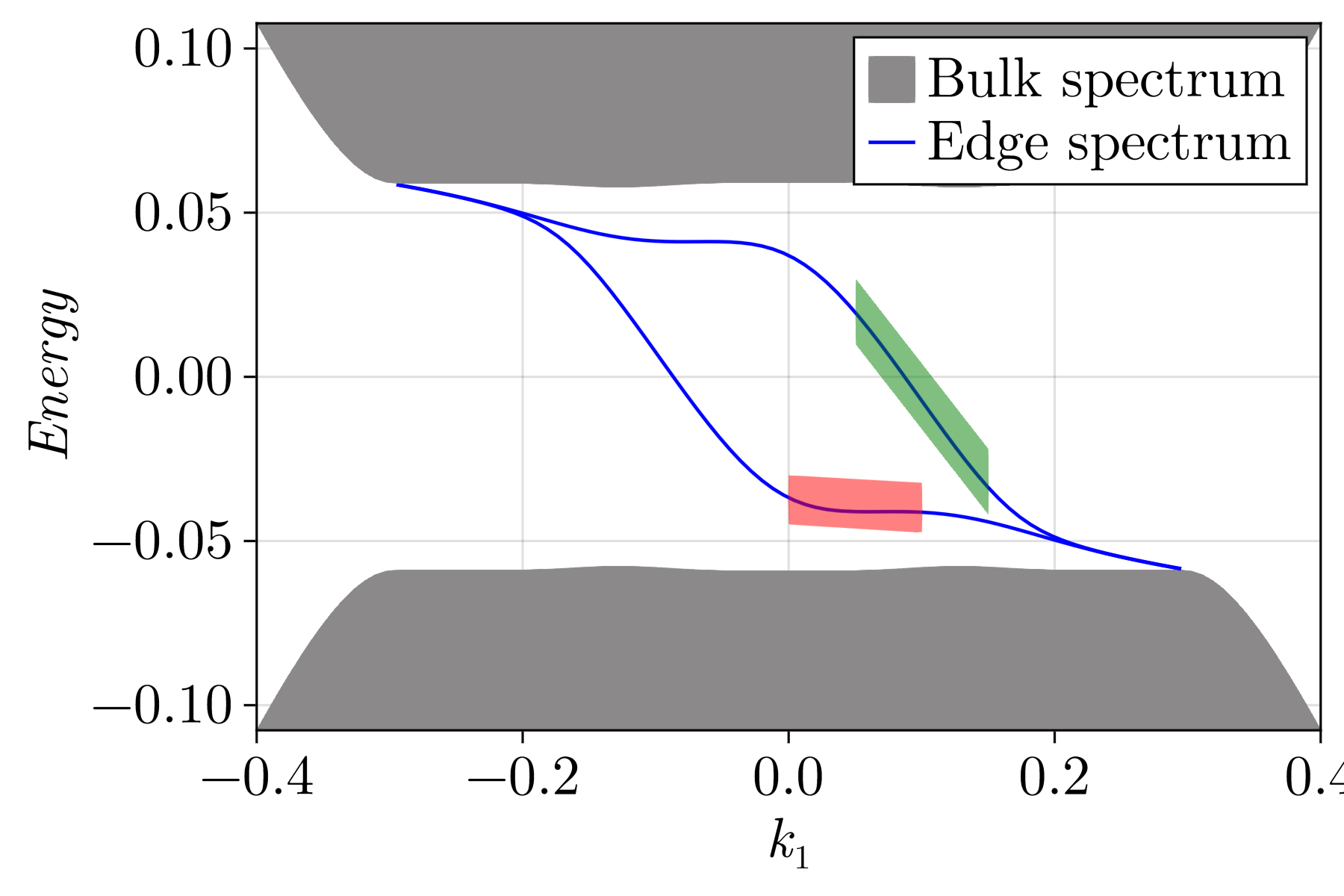}
    \caption{Numerical band structure for a straight edge profile, computed using the continuum Hamiltonian of order 1, $H_1$. Highlighted respectively in green and red is the support for the two localized edge modes studied in section~\ref{sec:edgestatenum}.}
    \label{fig:edgebands}
    \end{subcaptionblock}
    \hfill
    \begin{subcaptionblock}{3in}
        \centering
        \includegraphics[width=3in]{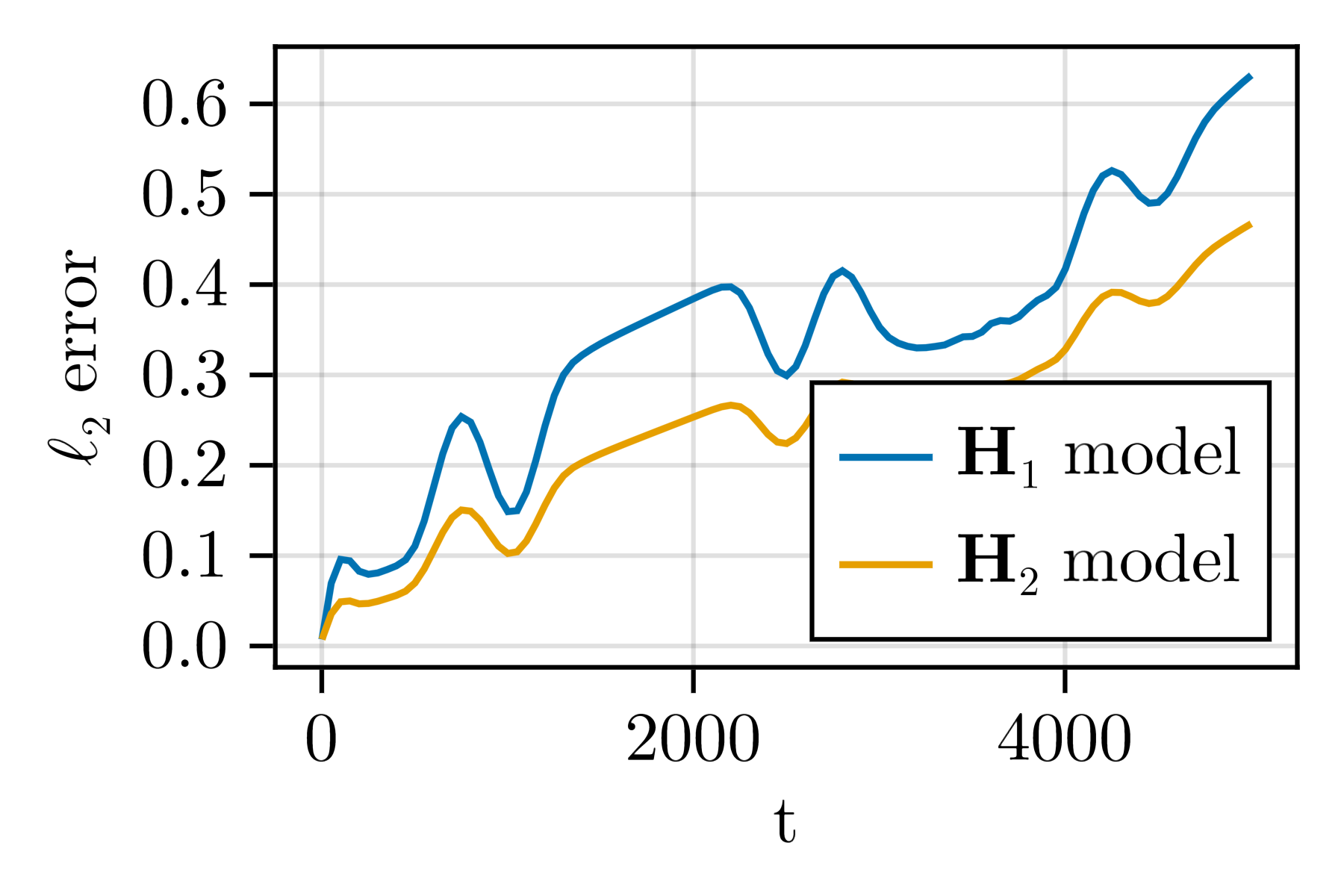}
        \caption{Error plot as a function of time for both the first and second order continuum models compared to tight-binding for the fast edge mode supported by the green momentum-energy region in Figure~\ref{fig:edgebands}.}
        \label{fig:fastedge_error}
    \end{subcaptionblock}
    \newline
    \begin{subcaptionblock}{3.25in}
        \centering
        \includegraphics[width=3.25in]{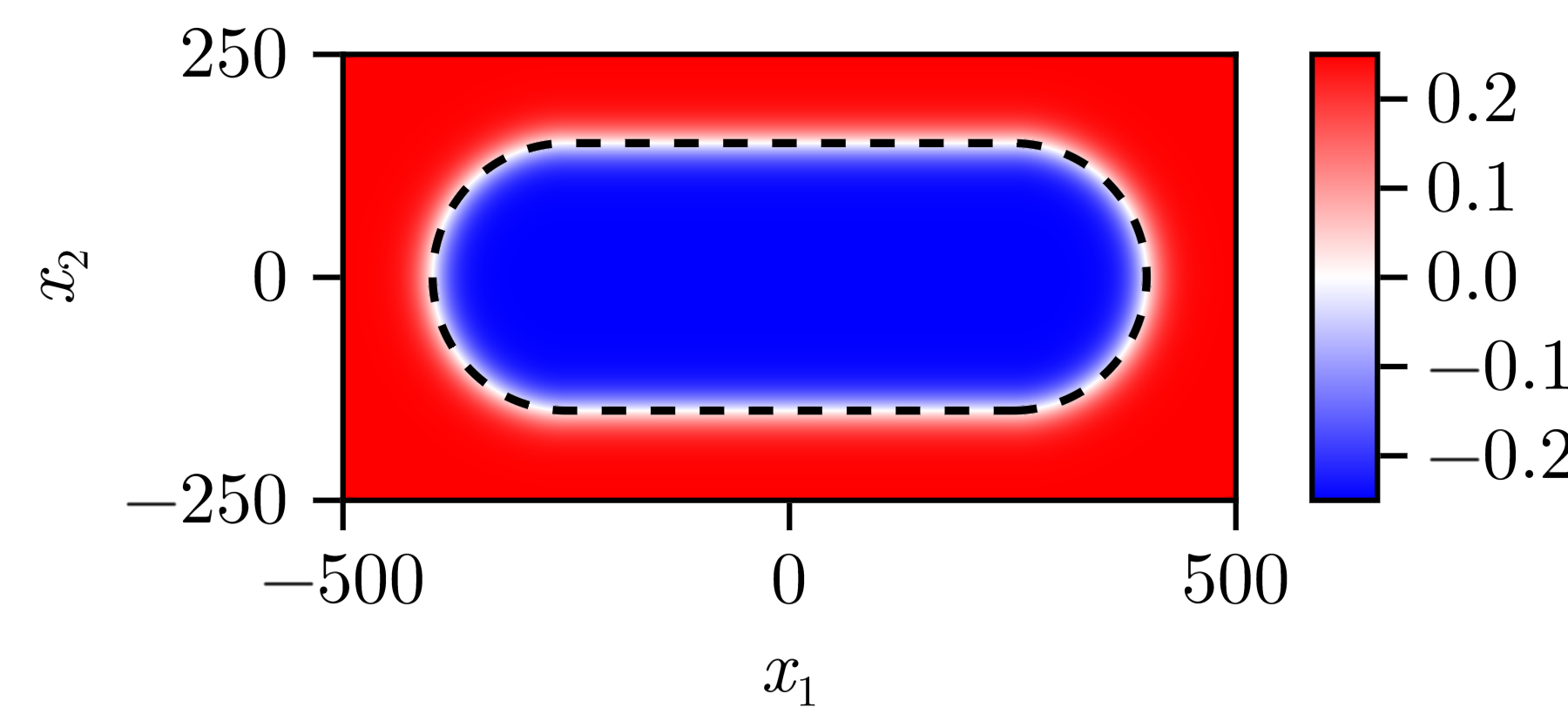}
        \bigskip
        
        \caption{Profile of the gating parameter $\omega(x_1,x_2) = \omega m \left ( \frac{\delta}{\omega} \br \right )$ used for the edge mode propagation simulations in section~\ref{sec:edgestatenum}.}
        \label{fig:racetrack}
    \end{subcaptionblock}
    \hfill
    \begin{subcaptionblock}{3in}
        \centering
        \includegraphics[width=3in]{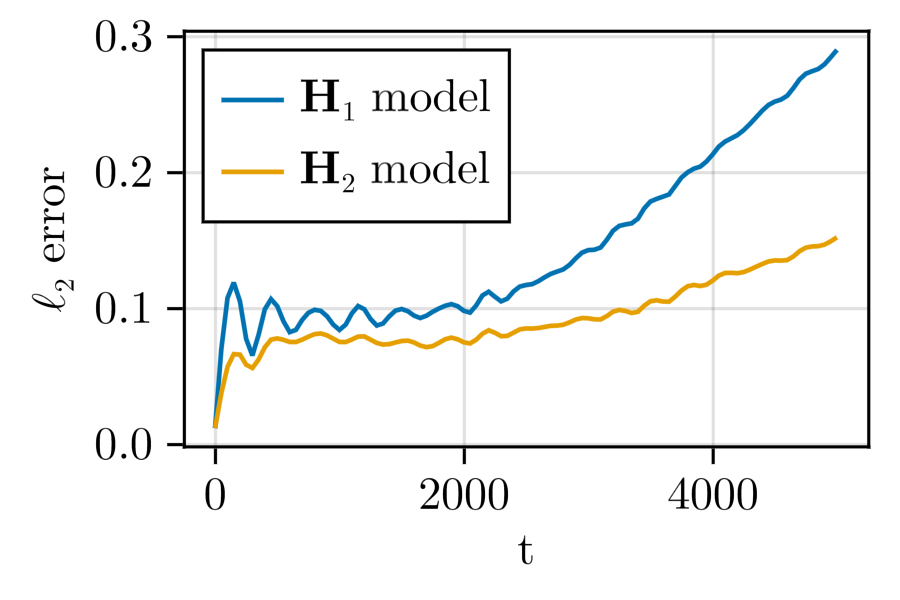}
        \caption{Error plot as a function of time for both the first and second order continuum models compared to tight-binding for the slow edge mode supported by the red momentum-energy region in Figure~\ref{fig:edgebands}.}
        \label{fig:slowedge_error}
    \end{subcaptionblock}
    \caption{One dimensional band structure, geometry and error estimates for the edge modes presented in section~\ref{sec:edgestatenum}.}
    \label{fig:edgestatenum}
\end{figure}
\subsubsection{Initial wave packet preparation.}
In Figure \ref{fig:edgebands}, we plot the edge band structure for bernal stacked bilayer graphene modeled by the first order continuum model $\bH_1$ with parameters described above for a straight edge profile along the $x_1$-axis $m(x_2) := \tanh(x_2)$, obtained by numerical diagonalization of the corresponding one-dimensional Hamiltonian parameterized by the momentum $k_1$ with appropriate boundary conditions. We highlight in particular two regions from this band structure from which an appropriate linear combination of eigenstates is used to prepare localized initial wave packets for both the tight-binding and continuum models, as analytic formulae for these edge modes are not available. 
We observe the group velocity is much larger in the green region than in the red one, and hence we refer to the former as the fast edge mode and the second as the slow edge mode. After Fourier transform in the $x_1$ direction, these wave packets are then translated and interpolated on the top edge of the racetrack profile, see Figure~\ref{fig:racetrack}, where the gating function $\omega(x)$ locally matches the $x_1$-invariant edge profile described above.

\paragraph{Fast edge mode}
\begin{figure}[ht!]
    \centering
    \includegraphics[width=.94\textwidth]{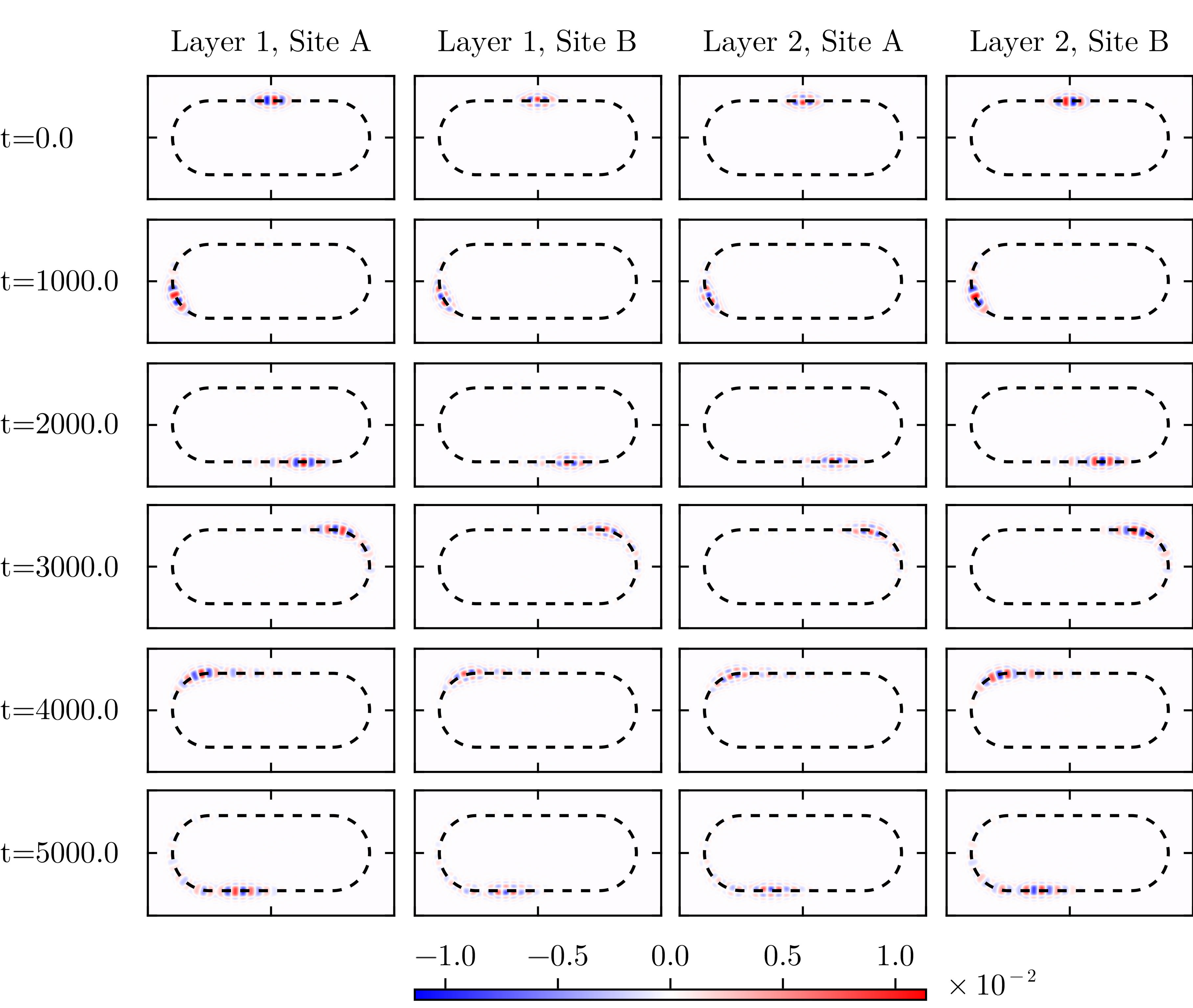}
    \caption{Snapshots of the real part of the fast edge mode envelope simulated\textbf{} using the tight-binding model.}
    \label{fig:fastedge}
\end{figure}
\begin{figure}[ht!]
    \centering
    \includegraphics[width=.94\textwidth]{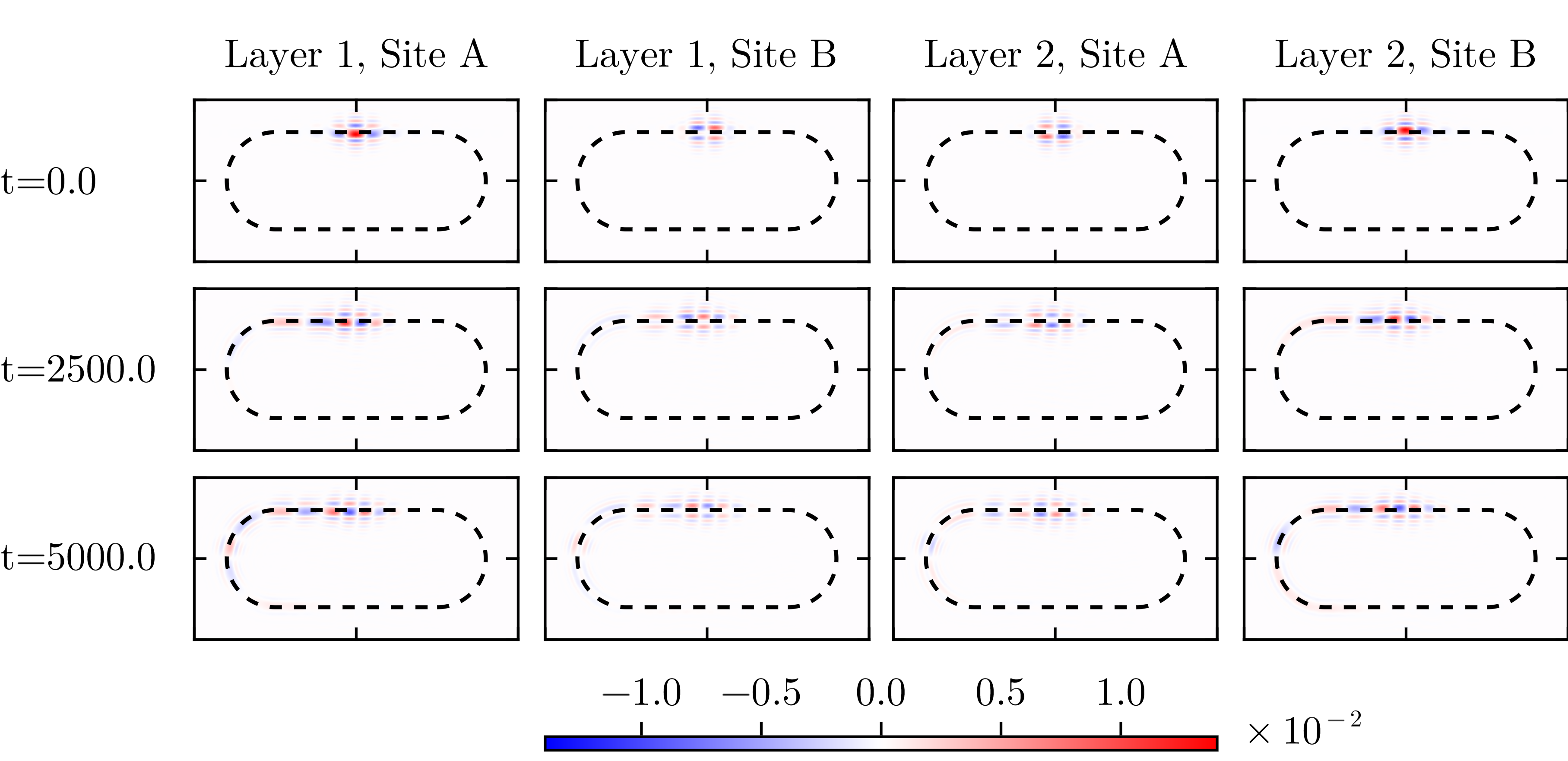}
    \caption{Snapshots of the real part of the slow edge mode envelope simulated using the tight-binding model.}
    \label{fig:slowedge}
\end{figure}
We plot 6 snapshots of the real part of the edge mode envelope $\widetilde{\psi}(t,x) = \psi(t,x) e^{-iK \cdot x}$, see ansatz~\eqref{eq:psi_def}, for all four components using the tight binding model in Figure~\ref{fig:fastedge}. 
Simulations using the continuum models are qualitatively undistinguishable from the tight binding results, and the corresponding plots are omitted.
We note that this edge mode has much more complex spatial structure compared to the massive Dirac equation with two components~\cite{bal2023edge}, with significant oscillations perpendicular to the edge and components 1B and 2A exhibiting an asymmetrical profile, unlike components 1A and 2B.

Propagation in time shows significant dispersion as the mode travels one and a half times around the racetrack at approximately unit speed, induced in particular by the curved profile of the edge~\cite{bal2023edge}, although the packet still propagates coherently for a long time.

Error plots comparing tight-binding to first and second order continuum models are presented in Figure~\ref{fig:fastedge_error}. Recall that solutions are normalized such that $\Vert \psi \Vert_{\ell_2} = 1$ for both tight-binding and continuum interpolated wave packets. While the error between the two models is significant and becomes relatively large at the end of the integration interval, we note that the timescale of simulations goes far beyond the theoretically predicted $T \simeq \delta^{-1}$ and qualitative agreement is still very good. The error curve displays the expected linear time dependence for small $t \ll \delta^{-1}$, then error growth slows down and exhibits a more oscillatory behavior (although this is evidently setup-dependent). Finally, we note the significant improvement between first and second order models which is uniform over the simulation time period.

\paragraph{Slow edge mode}
We plot 3 snapshots of the second edge mode, which is prepared from the momentum-energy window highlighted in red in Figure~\ref{fig:edgebands}. While the mode looks rather similar in its spatial structure to the fast edge mode studied above, the group velocity of this second mode is vanishingly small over a significant portion of the relevant part of the edge bands. Although the system is time integrated over the same interval as the previous example, the wavepacket barely moves, although a relatively small part of the packet advances slowly and turns onto the curved portion of the racetrack after $t \approx 2500$. Note that due to its relatively higher energy, the mode also oscillates at a relatively faster rate, even as it does not advance.

Error plots comparing evolution with the tight-binding model to first- and second-order continuum models are presented in Figure~\ref{fig:slowedge_error}. After an initial transient in the regime $t \ll \delta^{-1}$, which highlights the increased accuracy of the second order model in this region, the error plateaus at a relatively low level for both models for some time. 
Interestingly, the error starts increasing again at a higher rate, roughly linearly in time, around $t \approx 2500$; this coincides roughly with the time a portion of the wavepacket turns onto the curved portion of the interface, and analysis of the spatial profile of the error 
shows that the error is initially mostly composed of bulk modes due to the slight mismatch between models when preparing the initial wavepacket, until $t\approx 2500$ where a significant model error starts to accumulate on the curved section of the interface and eventually dominates the difference between the tight binding and continuum models.
\begin{figure}[b!]
    \centering
    \begin{subcaptionblock}{3in}
    \centering
    \includegraphics[width=3in]{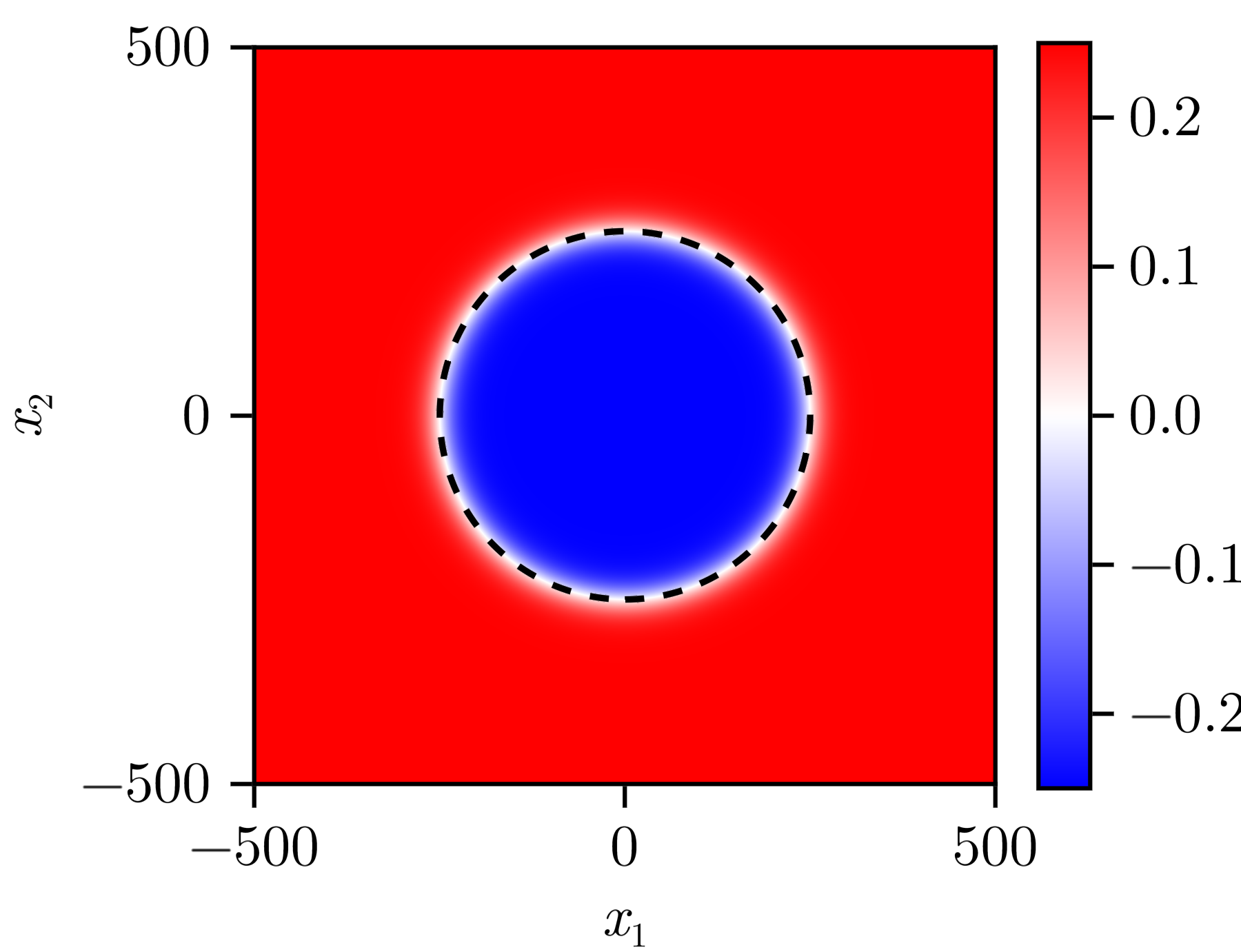}
    \caption{Profile of the gating parameter $\omega(x)$ used for the bulk mode propagation simulations in section~\ref{sec:bulkstatenum}.}
    \label{fig:circle}
    \end{subcaptionblock}
    \hfill
    \begin{subcaptionblock}{3.25in}
    \centering
        \includegraphics[width=3.25in]{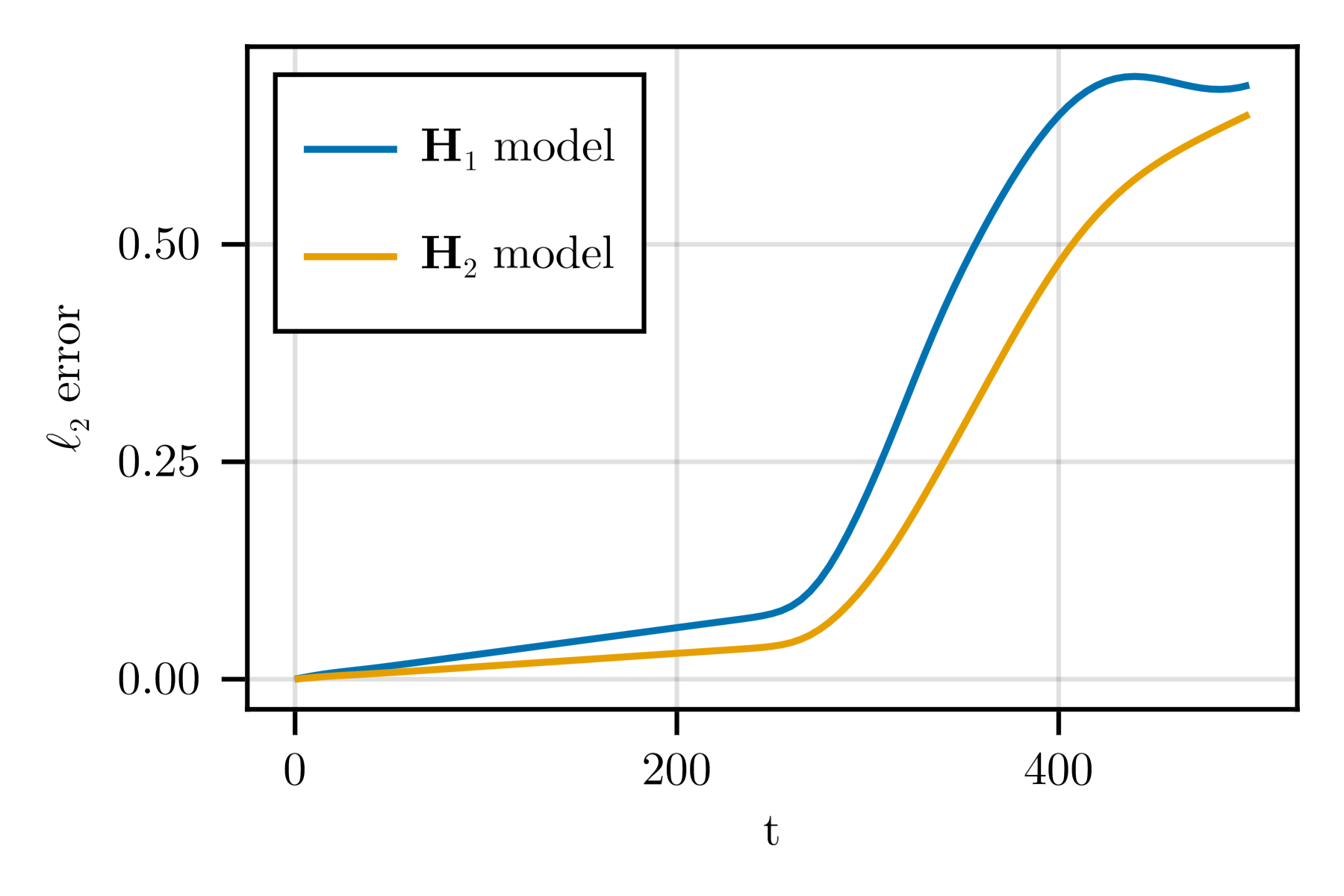}
        \caption{Error plot as a function of time for both the first and second order continuum models compared to tight-binding for the bulk edge mode.}
        \label{fig:bulkmode_error}
        \end{subcaptionblock}
        \caption{Geometry and error estimates for the bulk edge mode presented in section~\ref{sec:bulkstatenum}.}
        \label{fig:bulksetup}
\end{figure}

\begin{figure}[ht!]
    \centering
    \begin{subcaptionblock}{3.25in}
    \centering
    \includegraphics[width=3in]{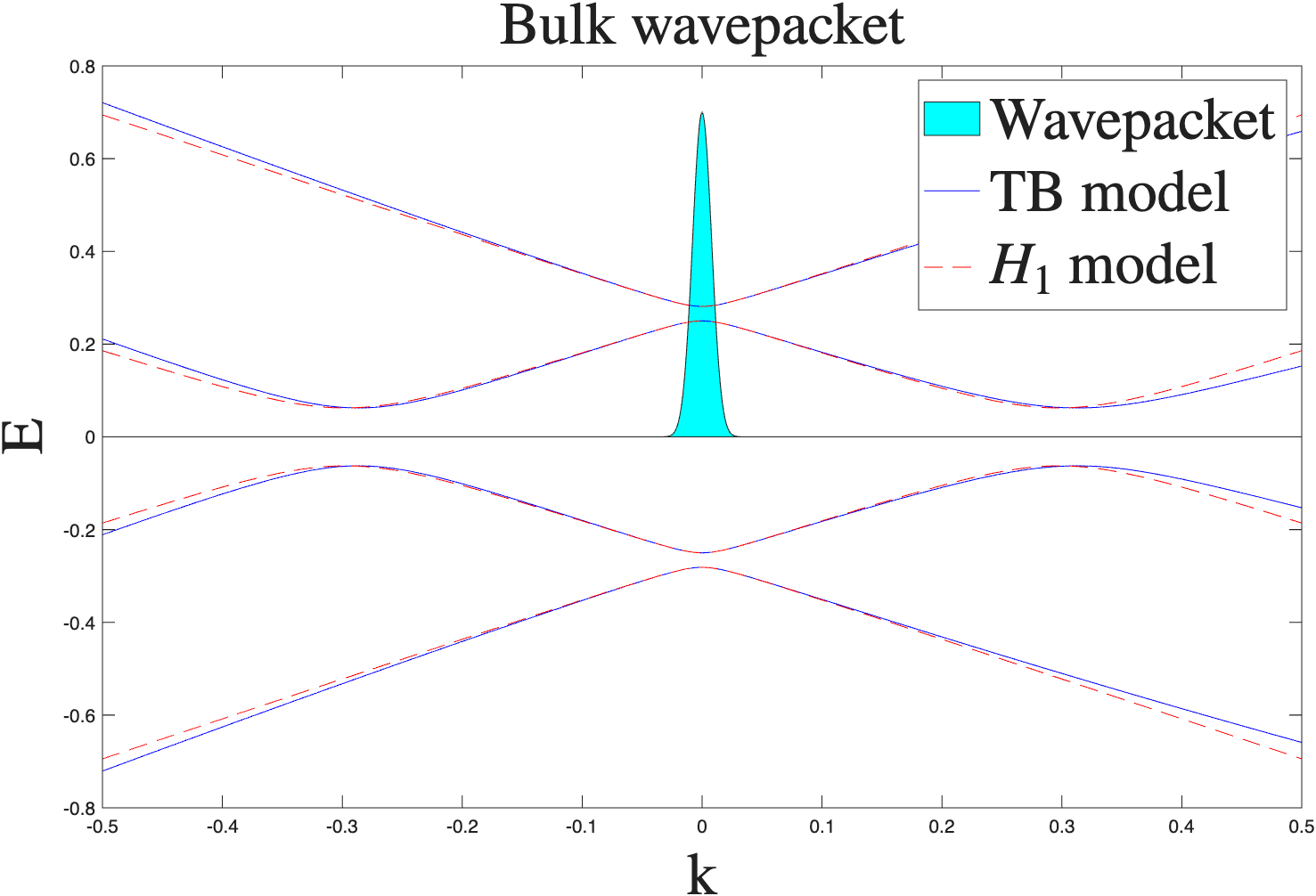}
    \caption{The gaussian profile of $\hat \phi_0(k)$ plotted against the band structure of $\bH_1$ and the tight-binding model along the k-space line $K + \hat e[-r,r]$ for $r = 1/2$ and $\hat e$ the unit vector in the $K'-K$ direction.}
    \label{fig:circle2}
    \end{subcaptionblock}
    \hfill
    \begin{subcaptionblock}{3in}
    \centering
        \includegraphics[width=3in]{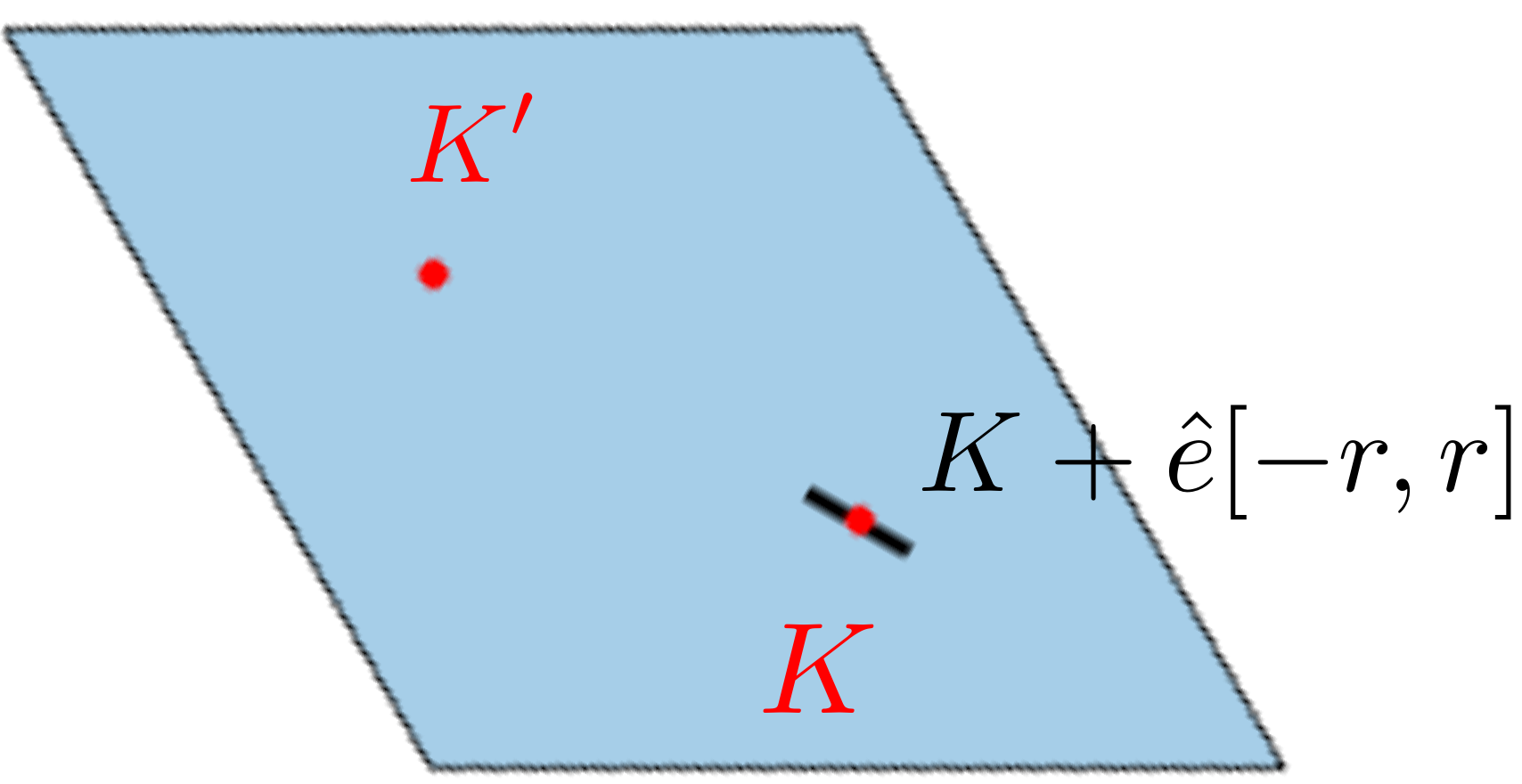}
        \vspace{.1in}
        \caption{The Brillouin Zone with $K$ points labeled and the interval $K+\hat e [-r,r]$ drawn.}
        \vspace{.3in}
        \label{fig:bulkmode_error2}
        \end{subcaptionblock}
        \caption{Wavepacket description for the initial bulk state.}
        \label{fig:bulk_state}
\end{figure}


\subsection{Bulk state propagation}\label{sec:bulkstatenum}

As a third and last example, we propose to study the propagation of a low-energy bulk mode, that is a mode which is not carefully prepared such that it stays confined to the interface. To highlight some interesting geometric features of the solution, we modify parameters for the racetrack such that it traces a rotationally invariant profile, see~Figure~\ref{fig:circle}. 

\noindent An initial Gaussian pulse of the form
\[
    \phi_0(x) = c_{s,l} e^{-x^2/2\sigma^2}
\]
is prepared with width $\sigma_0 = 20$ and parameters $c_{s,l}$ which are randomly chosen to avoid spurious symmetries.
We illustrate the Fourier transform of this wavepacket against the band structure of the bulk tight-binding model and $\bH_1$ in Figure \ref{fig:bulk_state} to show the initial state lies within the regime the continuum model is accurate.
Snapshots of the time integrated tight-binding model solution are presented in Figure~\ref{fig:bulksnapshots} and show the solution traveling outwards at roughly unit speed. At time $t \approx 250$ the interface and the wavepacket interact, with some of the wave being reflected towards the origin and forming a spiral pattern at $t=500$, and some of the packet transmitted and radiating outwards, showing some curious asymmetrical, nearly triangular features especially for components 1B and 2A. 
By comparison, the continuum (order 1) solution at the same time, although qualitatively in agreement, lacks such higher order features, with components 1B and 2A exhibiting a nearly rotationally invariant profile.

An error plot comparing evolution with the tight-binding model to first- and second-order continuum models is presented in Figure~\ref{fig:bulkmode_error}. Similarly to the slow edge mode, we observe a slow error accumulation until $t=250$, when the mode interacts with the interface, at which time the error starts to increase rapidly. This showcases how one of the main contributors to model error for the macroscopic continuum models is the spatially dependent coefficients. As in earlier experiments, the second order model is also uniformly more accurate, especially in the early regime, $t \leq 250$.

\section{High-order expansions, gauge transformations, and topology}
\label{sec:topo}
The results presented in Theorem \ref{thm:main} show that high-order expansions of the Weyl symbol of a degenerate Hamiltonian lead to more accurate descriptions of wave-packet transport. In this section, we demonstrate that such expansions do not always lead to more accurate spectral descriptions. In particular, we show that topological invariants naturally associated to the TB and macroscopic models $\bH_p = \ow b_{0p}$ in general depend on $p$. We even show that the invariants associated to $\bH_p$ may depend on ($\delta-$dependent unitary) gauge transformations of the TB model $H_\delta$.

\begin{figure}[ht!]
    \centering
    \begin{subcaptionblock}{.95\textwidth}
    \includegraphics[width=.95\textwidth]{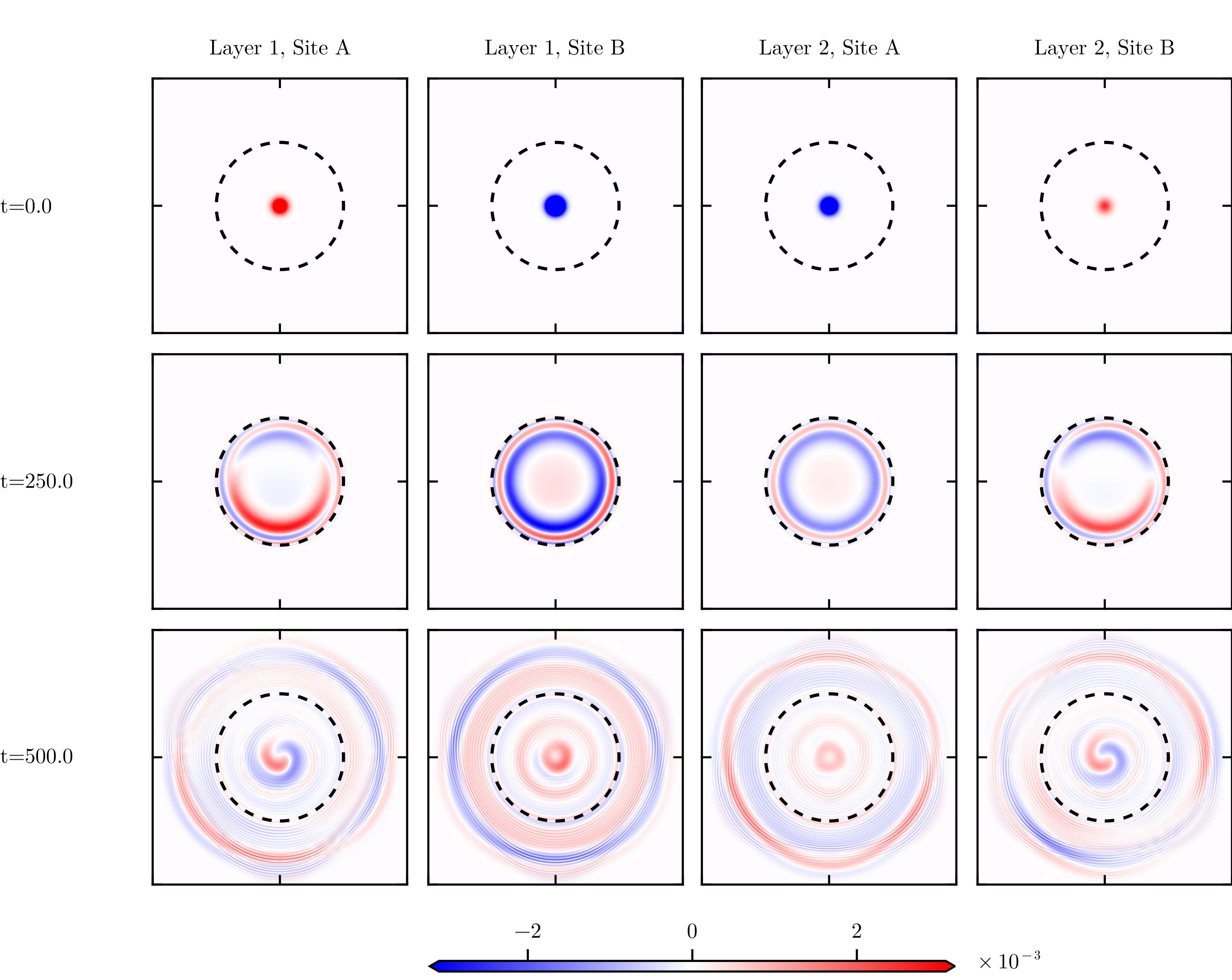}
    \caption{Snapshots of the tight-binding solution envelope.}
    \label{fig:bulksnapshots_tb}
    \end{subcaptionblock}
    \hfill
    \begin{subcaptionblock}{.95\textwidth}
    \includegraphics[width=.95\textwidth]{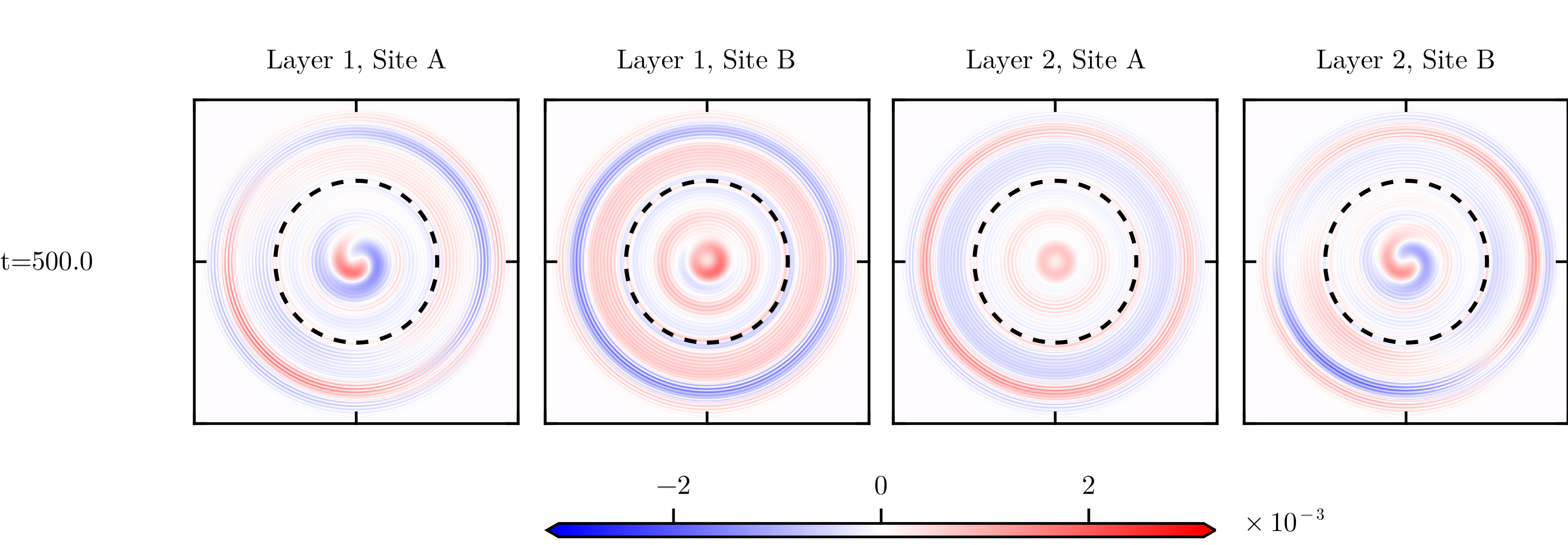}
    \caption{Snapshot of the solution computed using the continuum model of order 1, $\bH_1$.}
    \label{fig:bulksnapshots_cont}
    \end{subcaptionblock}
    \caption{Snapshot of the real part of the bulk mode envelope.}
    \label{fig:bulksnapshots}
\end{figure}
We illustrate this behavior on the Haldane model of section \ref{sec:Haldane}, and in particular the macroscopic Hamiltonians $\bH_1$ and $\bH_2$ given by \eqref{eq:macroHaldane}. The Haldane model \cite{haldane1988model} is the typical example displaying a topological quantum anomalous Hall effect; see \cite{BH} for detail on the topological insulators and topological invariants of many tight binding models. The topological classification of the effective (elliptic) differential operators appearing in this work has been analyzed in \cite{bal2022topological,bal2023topological,quinn2024approximations}; see also the review \cite{bal2024continuous} for details. Topological classifications of effective models of Rhombohedral graphene and related models of Floquet topological insulators are also available at \cite{bal2022multiscale,frazier2025quantum}.

We consider the following modified version of the Haldane model where $h \in \mathbb{R}$ and the $B$ sites are shifted by $-h\la_1=-\frac{hv}{\sqrt3}(1,0)$. Note that the choice $h=0$ corresponds to the Hamiltonians $\bH_1$ and $\bH_2$ in \eqref{eq:macroHaldane}, while $h=1$ corresponds to the $B$ sites located on top of the $A$ sites.
Dropping the subscript $\delta$ from $H_\delta$ for brevity, we define the $L^2$ unitarily equivalent Hamiltonian
\[ \tilde H = \tilde\tau H \tilde \tau^*,\qquad \tilde \tau = \begin{pmatrix}
    1 & 0 \\ 0 & {\tau_{h\la_1}}
\end{pmatrix}.\]
Thus, $H$ and $\tilde H$ are unitarily equivalent. When $h=1$, the latter is computationally convenient as nearest-neighbor shifts no longer appear since they are gauged away. A macroscopic approximation that follows from our framework, however, places A and B sites at the same spatial location, which is physically incorrect. This inaccuracy appears at second-order approximation in powers of $\delta$ as we now demonstrate.
Using \eqref{eq:a0} and the above transformation yields 
\begin{align*}
    \tilde{H} = \ow\tilde{a}_\delta, \qquad \tilde{a}_\delta (x,\xi) = \begin{pmatrix}
        a_\delta^{11} (x,\xi) & e^{-ih\la_1 \cdot \xi} a_\delta^{12} (x + \frac{h\la_1}{2},\xi)\\
        e^{ih\la_1 \cdot \xi} a_\delta^{21} (x + \frac{h\la_1}{2},\xi) & a_\delta^{22} (x + h\la_1, \xi)
    \end{pmatrix},
\end{align*}
where $a_\delta^{ij}$ denotes entry $(i,j)$ of $a_\delta$. Let us assume that $t_1(X)=t_1$ is constant to simplify calculations. Then
$
\tilde a_\delta (x,\xi) = \tilde a_0 (\delta x,\xi;\delta) + \delta \tilde a_1 (\delta x, \xi;\delta),
$
with 
\[ \begin{array}{rcl}
\tilde a_0(X,\xi;\delta) &=&  \begin{pmatrix}
        0 & e^{-ih\la_1 \cdot \xi} a_\delta^{12} (\xi)\\ 
        e^{ih\la_1 \cdot \xi} \overline{a_\delta^{12} (\xi)} & 0
    \end{pmatrix},
\\[5mm]
\tilde a_1(X,\xi;\delta) &=& \begin{pmatrix}
        M(X) + t_2(X) \sum_j \sin \lb_j\cdot\xi & 0\\ 
        0 & -M(X+\delta h\la_1) - t_2(X+\delta h \la_1) \sum_j \sin \lb_j\cdot\xi
    \end{pmatrix},
\end{array}\]
where we have used the shorthand $a_\delta^{12} (\xi) =a_\delta^{12} (x,\xi)$ since the latter is independent of $x$.
Performing Taylor expansions on this un-centered model, we observe that
\begin{align}\label{eq:tilde_b_01}
    \tilde b_{01}(X,\zeta)= \frac{\sqrt{3} v}{2} t_1 (\zeta_2 \sigma_1 - \zeta_1 \sigma_2) +
    \Big(M(X) - \frac{3\sqrt{3}}{2} t_2 (X) \Big) \sigma_3 = b_{01}(X,\zeta) 
\end{align}
so that the leading term in the Dirac equation is independent of $h$. 

Let us for concreteness assume that $t_1 > 0$ and that $\tM (X_2):= M (X) - \frac{3\sqrt{3}}{2} t_2 (X)$ depend only on $X_2$ with $\tM \in \fs (E_-, E_+)$ for some $E_- < 0 < E_+$. Here, $\fs(a,b)$ is the set of smooth {\em switch} functions $f(x)$ on $\Rm$ such that $f(x)=a$ for $x<x_0$ and $f(x)=b$ for $x>x_1$ for some $x_0<x_1$. The function $\tilde M(X_2)$ therefore generates a domain wall from one insulator (for energies close to $E=0$) as $X_2\to-\infty$ to another insulator (for the same energy range) as $X_2\to+\infty$. The interface $X_2\approx0$ separating the insulating phases may be conducting with a quantized asymmetric transport described by a topological invariant \cite{BH} referred to in this setting as a bulk-difference invariant (BDI) \cite{bal2022topological,bal2023topological,quinn2024approximations}. The BDI associated with symbols of the form $b_{0p} (X_2,\zeta; \delta) = \sum_{j=1}^3f^{p}_j (X_2,\zeta; \delta) \sigma_j$ takes the particular form
\begin{equation}\label{eq:BDI}
    \bdi [\ow b_{0p}] = -\sum_{\ell=1}^L \sign \det J_\ell \in \Zm,
\end{equation}
where $J_\ell \in \mathbb{R}^{3 \times 3}$ is the Jacobian matrix defined by $J_\ell^{mn} = \partial_m f^p_n (X_{2}^\ell, \zeta_1^\ell, \zeta_2^\ell)$, the variables are identified by {$(1,2,3) = (X_2,\zeta_1, \zeta_2)$}, and $\{(X_{2}^\ell, \zeta_1^\ell, \zeta_2^\ell) : \ell \in \{1, \dots, L\}\}$ is the set of points $(X_2, \zeta_1, \zeta_2)$ for which $f^p (X_2, \zeta_1, \zeta_2) = (0,0,0)$, with $f^p := (f^p_1, f^p_2, f^p_3)$. We refer to \cite[Proposition 2.7]{quinn2024approximations} for more details.

Assume without loss of generality that the function $\tM$ has one zero and has slope $1$ there. Then for the first-order model with symbol $\tilde{b}_{01}$ given by \eqref{eq:tilde_b_01}, the vector field $$f^1 (X_2, \zeta_1, \zeta_2) := (\frac{\sqrt{3} v}{2} t_1 \zeta_2, -\frac{\sqrt{3} v}{2} t_1 \zeta_1, \tM (X_2))$$ has exactly one zero, and the associated Jacobian matrix is
\begin{align*}
    J_1 = \begin{pmatrix}
        0 & 0 & 1\\
        0 & -\frac{\sqrt{3} v}{2} t_1 & 0\\
        \frac{\sqrt{3} v}{2} t_1 & 0 & 0
    \end{pmatrix}.
\end{align*}
It follows that
\begin{align*}
    \bdi [\ow \tilde b_{01}] = -1,
\end{align*}
which is the standard asymmetric current associated with the Dirac operator \cite{bal2022topological,quinn2024approximations}.

However, at second-order, we find
\begin{align*}
    \tilde b_{02} (X,\zeta;\delta) = \sum_{j=1}^3 \tilde f_j (X_2, \zeta) \sigma_j  - \delta h\la_1\cdot \left( \nabla M(X) - \frac{3\sqrt3}2  \nabla t_2(X) \right) \frac{I-\sigma_3}2,
\end{align*}
where
\begin{align*}
    \tilde f_1 (X_2,\zeta) &= t_1 \Big( \frac{\sqrt{3} v}{2} \zeta_2 + \delta \frac{v^2}{8} (\zeta_2^2 -\zeta_1^2 +4 h \zeta_1^2)\Big), \qquad \tilde f_2 (X_2,\zeta) = t_1 \Big( -\frac{\sqrt{3} v}{2} \zeta_1 + \delta \frac{v^2}{4} (1+2h) \zeta_1 \zeta_2 \Big), \\
    \tilde f_3 (X_2, \zeta) &= \tM (X_2).
\end{align*}
We observe an $O(\delta)$ term proportional to $I-\sigma_3$ in $\delta\tilde b_{02}$, which is negligible compared to the spectral gap modeled by $(-E_0,E_1)$ and to the term $\tilde f_3$ and hence may be neglected in the computation of the BDI \cite{quinn2024approximations}.

The BDI associated to $\tilde b_{02}$ is determined in \eqref{eq:BDI} by the zeros of $(\tilde f_1, \tilde f_2,\tilde f_3)$.
Since $(\tilde f_1, \tilde f_2)$ is independent of $X_2$ and $\tilde f_3$ independent of $\zeta$, it suffices to analyze the zeros of
the vector field $\zeta \mapsto (\tilde f_1, \tilde f_2)$, two of which are independent of $h$ and located at
\begin{align*}
    \zeta^1 := (0,0), \qquad \zeta^2:= \left( 0, -\frac{4 \sqrt{3}}{\delta v} \right).
\end{align*}
If $-1/2 < h < 1/4$, then the vector field $\zeta \mapsto (\tilde f_1, \tilde f_2)$ has two additional zeros, given by
\begin{align*}
    \zeta^3 := (\zeta_{1,h}, \zeta_{2,h}), \qquad \zeta^4 := (-\zeta_{1,h}, \zeta_{2,h}), \qquad \zeta_{1,h} := \frac{6}{\delta v} \sqrt{\frac{1+\frac{4}{3} h}{1-4h}} \frac{1}{1+2h}, \qquad \zeta_{2,h}:= \frac{2\sqrt{3}}{\delta v} \frac{1}{1+2h}.
\end{align*}
The Jacobian matrices $J_{ij} = \partial_i \tilde{f}_j$ associated with these zeros are
\begin{align*}
    J (\zeta^1) &= \begin{pmatrix}
        0 & -\frac{\sqrt{3} v}{2} t_1\\
        \frac{\sqrt{3} v}{2} t_1 & 0
    \end{pmatrix}, \qquad
    J (\zeta^2) = \begin{pmatrix}
        0 & -\frac{3\sqrt{3}}{2} v t_1 (1 + \frac{4}{3} h)\\
        -\frac{\sqrt{3}}{2} v t_1 & 0
    \end{pmatrix},\\
    J (\zeta^3) &= \begin{pmatrix}
        -\frac{3}{2} v t_1 \frac{\sqrt{(1-4h)(1+\frac{4}{3}h)}}{1+2h} & 0 \\
        \frac{\sqrt{3}}{2} v t_1 (1 + \frac{1}{1+2h}) & \frac{3}{2} v t_1 \sqrt{\frac{1+\frac{4}{3} h}{1-4h}}
    \end{pmatrix}, \qquad
        J (\zeta^4) = \begin{pmatrix}
            \frac{3}{2} v t_1 \frac{\sqrt{(1-4h)(1+\frac{4}{3}h)}}{1+2h} & 0 \\
        \frac{\sqrt{3}}{2} v t_1 (1 + \frac{1}{1+2h}) & -\frac{3}{2} v t_1 \sqrt{\frac{1+\frac{4}{3} h}{1-4h}}
        \end{pmatrix},
\end{align*}
with the variables identified by $(1,2) = (\zeta_1, \zeta_2)$. The associated topological charges are
\begin{align*}
    -\sgn \det J (\zeta^1) = -1, \qquad -\sgn \det J (\zeta^2) = +1, \qquad -\sgn \det J (\zeta^3) = +1, \qquad -\sgn \det J (\zeta^4) = +1,
\end{align*}
and thus
\begin{align*}
    \bdi [\ow \tilde b_{02}] = -\sum_{\ell=1}^4 \sgn \det J (\zeta^\ell) = +2, \qquad -1/2 < h < 1/4.
\end{align*}
But if $h > 1/4$, then the zeros $\zeta^{3,4}$ are no longer real-valued, meaning that the edge current is
\begin{align*}
    \bdi [\ow \tilde b_{02}] = -\sum_{\ell=1}^2 \sgn \det J (\zeta^\ell) = 0, \qquad h > 1/4.
\end{align*}

To summarize, making the $h$-dependence of the operator $\bH_p^h := \ow \tilde b_{0p}$ explicit, we have shown that
$$\bdi [\tilde \bH_1^h] = -1, \qquad h \in \mathbb{R}$$ is independent of $h$, while
$$\bdi [\tilde \bH_2^h] = \begin{cases}
    +2, & -1/2 < h < 1/4\\
    0, & h> 1/4
\end{cases}$$
jumps at $h=1/4$. The above example illustrates that the edge current invariant is not preserved when we add higher order terms to the continuum model, and that it even depends on the gauge used in the definition of unitarily equivalent Hamiltonians. Although the higher order terms improve the accuracy of the continuum model in the vicinity of $K$, they induce spurious zeros (and topological charge) of the Hamiltonian's symbol at high frequencies of order $O(\delta^{-1})$.

Note that high-order approximations with the correct topological characteristics may be obtained by regularizing the correction term, for instance by defining a modification $\check \bH_2$ of $\bH_2$ with symbol given by
\[
\check b_{02} (X, \zeta;\delta)
    = 
    b_{01} (X, \zeta;\delta) + \frac{v^2}{8} t_1 (X) \frac{\delta ((\zeta_2^2 - \zeta_1^2) \sigma_1 + 2\zeta_1 \zeta_2 \sigma_2)}{1+\alpha^2\delta \aver{\zeta}}.
\]
The regularization by the $\alpha-$dependent term preserves the asymptotic expansion of $H_\delta$ to second order and for $\alpha^2$ large enough (namely larger than {$v/2\sqrt3$}) does not generate any spurious zeros as one may verify. 
As a consequence the BDI of the corresponding operator would be $\bdi[\check \bH_2]=-1$ for all $\alpha$ sufficiently large, which is the correct invariant for the Haldane model.

This shows that the macroscopic models $\tilde \bH_2^h$ and $\check \bH_2$ of $H_\delta$ (up to a unitary transformation) all lead to second-order approximations for the temporal propagation of wave packets as underlined in Theorem \ref{thm:main} while providing different topological invariants (spectral information). Only $\check \bH_2$ has the correct invariant.



\section*{Acknowledgments} 
GB's research was funded in part by NSF grant DMS-230641 and ONR grant N00014-26-1-2017. DM’s research was supported by AFRL grant FA9550-24-1-0177. PC's research was supported in part by the Simons Foundation Travel Support for
Mathematicians award MPS-TSM-00966604. The Flatiron Institute is a division of the Simons Foundation.

\bibliographystyle{siam}
\bibliography{refs}

\appendix
\numberwithin{equation}{section}

\section{Pseudodifferential operators}

\subsection{Notation and the Calder\'on-Vaillancourt Theorem}\label{subsec:notation}
We first introduce notation on pseudo-differential operators used throughout the text; see, e.g. \cite{Bony, DS, Hormander, Zworski} for additional details. Let $\mathcal{S} (\mathbb{R}^d) \otimes \mathbb{C}^n$ be the Schwartz space of vector-valued functions and 
$\mathcal{S}' (\mathbb{R}^d) \otimes \mathbb{C}^n$ its dual. Let $\mathbb{M}_n$ denote the space of Hermitian $n \times n$ matrices. Given a 
symbol $a(x,\xi) \in \mathcal{S}' (\mathbb{R}^d \times \mathbb{R}^d) \otimes \mathbb{M}_n$, we define the Weyl quantization of $a$ as the operator 
\begin{align}\label{eq:weylquanth}
    \ow(a) \psi (x) :=
    \frac{1}{(2\pi)^d} \int_{\mathbb{R}^{2d}}
    e^{i(x-y)\cdot \xi}
    a(\frac{x+y}{2}, \xi) \psi (y) dy d\xi,
    \qquad
    \psi \in \mathcal{S} (\mathbb{R}^d) \otimes \mathbb{C}^n.
\end{align}
When $a$ is polynomial in $\xi$, it follows that $\ow (a)$ is a differential operator. 

For $(x, \xi) = X\in \mathbb{R}^{2d}$, we define $\aver{X} := \sqrt{1 + |X|^2}$. A function $\fm: \mathbb{R}^{2d} \rightarrow [0,\infty)$ is called an {\em order function} if there exist constants $C_0 > 0$, $N_0 > 0$ such that
$\fm(X) \le C_0 \aver{X-Y}^{N_0} \fm(Y)$ for all $X,Y \in \mathbb{R}^{2d}$. Moreover, if $\fm_1$ and $\fm_2$ are order functions, then so are $\fm_1^{-1}$ and $\fm_1 \fm_2$.  Note that $\aver{X}^p$ is an order function. 

\begin{definition}
We say that $a \in S(\fm)$ ({\em symbol class} associated to the order function $\fm$) if for every $\alpha \in \mathbb{N}^{2d}$, there exists $C_\alpha > 0$ such that $|\partial^\alpha a (X)| \le C_\alpha \fm(X)$
for all $X \in \mathbb{R}^{2d}$. 
If $a = a (X; \delta)$ also depends on some parameter $\delta \in \mathbb{R}$, we say that $a \in S (\fm)$ 
if all the constants $C_\alpha$ can be chosen independent of $\delta$. 
\end{definition}

We will also use the following H\"ormander classes of symbols defined in \cite{hormander1971fourier} and \cite[Chapter 18]{Hormander}. 
\begin{definition}
    For $m \in \mathbb{R}$, let $S^m$ denote the set of $a \in C^\infty (\mathbb{R}^d \times \mathbb{R}^d)$ such that for every $(\alpha, \beta) \in \mathbb{N}^d \times \mathbb{N}^d$,
\begin{align*}
    |\partial^\alpha_\xi \partial^\beta_x a (x,\xi)| \le C_{\alpha, \beta} \aver{\xi}^{m - |\alpha|}, \qquad (x, \xi) \in \mathbb{R}^d \times \mathbb{R}^d.
\end{align*}
Let $\cs^m$ denote the set of $a \in C^\infty (\mathbb{R}^d \times \mathbb{R}^d)$ such that for every $(\alpha, \beta) \in \mathbb{N}^d \times \mathbb{N}^d$,
\begin{align*}
    |\partial^\alpha_\xi \partial^\beta_x a (x,\xi)| \le C_{\alpha, \beta} \aver{\xi}^{m + |\beta|}, \qquad (x, \xi) \in \mathbb{R}^d \times \mathbb{R}^d.
\end{align*}
If $a = a (x,\xi; \delta)$ also depends on some parameter $\delta \in \mathbb{R}$, we say that $a \in S^m$ (or $a \in \cs^m$) 
if all the constants $C_{\alpha, \beta}$ can be chosen independent of $\delta$.
\end{definition}



Finally, we recall the following Calder\'on-Vaillancourt result \cite{calderon1972class}:
\begin{lemma}\label{lemma:S_bd}
    If $a = a(x,\xi)$ belongs to $S(\aver{\xi}^m)$, then for any $N \ge 0$, $\ow a : H^{N+m} \to H^N$ is bounded. Moreover, there exists $\ncv> 0$ (independent of $N$) such that
    \begin{align*}
        \norm{\ow (a) f}_{H^N (\mathbb{R}^d)} \le C \max_{|\alpha+\beta| \le \ncv+N} \norm{\aver{\xi}^{-m}\partial^\alpha_x \partial^\beta_\xi a (x,\xi)}_{L^\infty (\mathbb{R}^{2d})} \norm{f}_{H^{N+m} (\mathbb{R}^d)}
    \end{align*}
    uniformly in $a \in S(\aver{\xi}^m)$ and $f \in H^{N+m} (\mathbb{R}^d)$.
\end{lemma}

\subsection{Symbols of exponential type}\label{subsec:exp}
We collect two useful properties of symbols of the form $a_\delta (x, \xi) =a(\delta x, \xi; \delta)$, where $a$ is of exponential type; recall Definition \ref{def:es}.
\begin{lemma}\label{lemma:es_bd}
    Suppose $a (X,\xi;\delta) = a \in \es$ uniformly in $0<  \delta \le 1$ and define $a_\delta (x, \xi) := a (\delta x, \xi; \delta)$. Then there exists $\delta_0 > 0$ such that for any $N \ge 0$, the operator $\ow a_\delta : H^{N+d+1} \to H^N$ is bounded with norm uniform in 
    $0 < \delta \le \delta_0$.
\end{lemma}
\begin{proof}
    Let $\phase (X, \xi; \delta) = \phase \in S^1$ and $\sym (X, \xi; \delta) = \sym \in S(1)$ uniformly in $\delta$, and define
    \begin{align*}
        \phase_\delta (x, \xi) := \phase (\delta x, \xi; \delta), \qquad \sym_\delta (x, \xi) := \sym (\delta x, \xi; \delta).
    \end{align*}
    It suffices to show that $A := \ow e^{i \phase_\delta} \sym_\delta$ is bounded $H^{N+d+1} \to H^N$ uniformly in $\delta$. We write
    \begin{align*}
        Af (x) = \dint_{\Rm^{2d}} e^{i(x-y)\cdot\xi} e^{i \phase(\delta \frac{x+y}{2}, \xi; \delta)} \sym (\delta \frac{x+y}2,\xi; \delta) f(y) \dfrac{dy d\xi}{(2\pi)^d} = \dint_{\Rm^{2d}} e^{i B} \sym (\delta x - \frac{\delta z}2,\xi; \delta) f(x-z) \dfrac{dz d\xi}{(2\pi)^d},
    \end{align*}
    where we changed variables $z := x-y$ and use the shorthand $B :=z \cdot \xi + \phase (\delta x - \frac{\delta z}{2}, \xi; \delta)$.
    Using that
    \begin{align*}
        \nabla_z B = \xi - \frac{\delta}{2} \nabla_X \phase (\delta x - \frac{\delta z}{2}, \xi; \delta)
    \end{align*}
    with $\nabla_X \phase (X, \xi; \delta) \le C \aver{\xi}$ uniformly in all variables, it follows that for some $c, \delta_0 >0$,
    \begin{align}\label{eq:growth_B}
        |\nabla_z B| \ge c \aver{\xi} - c^{-1}, \qquad (x,z,\xi;\delta) \in \mathbb{R}^{3d} \times (0,\delta_0].
    \end{align}
    Similarly, $\nabla_\xi B = z + \nabla_\xi \phase (\delta x - \frac{\delta z}{2}, \xi; \delta)$
    with (each entry of) $\nabla_\xi \phase$ in $S^0$. It follows that (for possibly a smaller choice of $c, \delta_0 > 0$)
    \begin{align}\label{eq:growth_B2}
        |\nabla_\xi B| \ge c \aver{z} - c^{-1}, \qquad (x,z,\xi;\delta) \in \mathbb{R}^{3d} \times (0,\delta_0].
    \end{align}
    Now, let $\chi_0 \in C^\infty_c (\mathbb{R}^d)$ such that $\chi_0 (z) = 1$ 
    for all $|z| \le c^{-2}$. This way, $\nabla_z B \ne 0$ whenever $\xi \in \supp (1-\chi_0)$ and $\nabla_\xi B \ne 0$ whenever $z \in \supp (1-\chi_0)$. Set $\chi_1 := 1-\chi_0$ and for $j,k \in \{0,1\}$, define the operators $A_{jk}$ by
    \begin{align*}
        A_{jk} f(x) = \dint_{\Rm^{2d}} \chi_j (z) \chi_k (\xi) e^{i B} \sym (\delta x - \frac{\delta z}2,\xi; \delta) f(x-z) \dfrac{dz d\xi}{(2\pi)^d},
    \end{align*}
    so that $A = \sum_{j,k = 0}^1 A_{jk}$.
    By the regularity of $B$ and $\sigma$, it follows from Cauchy-Schwarz (applied to the integral over $x$) that $A_{00} : H^N \to H^N$ is bounded uniformly in $\delta$. To handle unbounded $\xi$, we write
    \begin{align*}
        A_{01} f(x) = \dint_{\Rm^{2d}} \chi_0 (z) \chi_1 (\xi) \sym (\delta x - \frac{\delta z}2,\xi; \delta) f(x-z) \left( \frac{1}{|\nabla_z B|^2} (-i \nabla_z B \cdot \nabla_z )\right)^n e^{i B} \dfrac{dz d\xi}{(2\pi)^d},
    \end{align*}
    which after integrating by parts becomes
    \begin{align*}
        A_{01} f(x) = \dint_{\Rm^{2d}}e^{i B} \left(i \nabla_z \cdot \nabla_z B\frac{1}{|\nabla_z B|^2}\right)^n  \chi_0 (z) \chi_1 (\xi) \sym (\delta x - \frac{\delta z}2,\xi; \delta) f(x-z) \dfrac{dz d\xi}{(2\pi)^d}.
    \end{align*}
    The boundedness of $A_{01}: H^{N+d+1} \to H^N$ then follows from taking $n:= d+1$ and applying \eqref{eq:growth_B} to ensure integrability in $\xi$. 
    We can similarly establish the boundedness of the $A_{1k}$ by integrating by parts in $\xi$ and applying \eqref{eq:growth_B2}.
\end{proof}
From the above proof, we can also deduce the following
\begin{lemma}\label{lemma:chi}
    Take $a_\delta$ 
    as in Lemma \ref{lemma:es_bd}, and let $0 \le \chi \in C^\infty_c (\mathbb{R}^d)$ such that $\chi \equiv 1$ in a neighborhood of the origin. Fix $\nu > 0$ and set $\chi_\delta (z) := \chi (\delta^\nu z)$.
    Then with $H_\delta := \ow a_\delta$, the operator $H_{00}^\delta$ defined by
    \begin{align*}
        H_{00}^\delta f(x) = \dint_{\Rm^{2d}} \chi_\delta (z) e^{iz \cdot \xi} a (\delta x - \frac{\delta z}2,\xi; \delta) f(x-z) \dfrac{dz d\xi}{(2\pi)^d}
    \end{align*}
    satisfies
    \begin{align*}
        \norm{(H^\delta_{00} - H_\delta) f}_{H^N} \le C_{n,N} \delta^n \norm{f}_{H^{N+d+1}}, \qquad n,N \in \mathbb{N}_0, \quad 0 < \delta \le 1, \quad f \in H^N.
    \end{align*}
\end{lemma}
\begin{proof}
    This result follows from integrating by parts as in the proof of Lemma \ref{lemma:es_bd} and applying 
    \eqref{eq:growth_B} and \eqref{eq:growth_B2}.
\end{proof}

\section{Proofs of main results}\label{sec:proofs}
\subsection{Results used in proof of Theorem \ref{thm:main_0}}\label{subsec:proofs_main}
\begin{lemma}\label{lemma:Sobolev_0}
    Suppose Assumption \ref{assumption:a} holds, and let $\phi_0 \in \cS (\mathbb{R}^d; \mathbb{C}^n)$.
    For any $N \ge 0$, there exists $C > 0$ such that the solution $\phi$ to \eqref{eq:macro_0} satisfies $\norm{\phi (T,\cdot \; ; \delta)}_{H^N} \le C$ uniformly in $T \ge 0$ and $0 < \delta \le \delta_0$.
\end{lemma}
\begin{proof}
    This result follows immediately from the ellipticity condition \eqref{eq:equivalent_norms} and Lemma \ref{lemma:Sobolev_0_appendix} (with $\Lambda = (i+\bH)^N$ and $r \equiv 0$ there).
\end{proof}
\begin{lemma}\label{lemma:res}
    Suppose Assumption \ref{assumption:a} holds. Take $H_\delta$ as in \eqref{eq:microWP} and $\psi_\delta$ as in \eqref{eq:psi_def}. Then
    \begin{align*}
        \norm{(D_t + H_\delta) \psi_\delta (t, \cdot)}_{H^N} \le C_N \delta^{1+\rate}, \qquad t \ge 0, \quad 0 < \delta \le \delta_0, \quad N \ge 0.
    \end{align*}
\end{lemma}
\begin{proof}
  We recall that $(D_t + H_\delta) \psi_\delta(t,x)$ is given in \eqref{eq:residual}.
    With $C_0>0$ defined in Assumption \ref{assumption:a} {and $B_r \subset \mathbb{R}^d$ the ball of radius $r$ centered at the origin}, let $\chi \in C^\infty_c (B_{C_0})$ such that $\chi \equiv 1$ in $B_{\frac{1}{2}C_0}$. For $0 < \delta \le \delta_0$, define $\chi_\delta( \zeta):= \chi (\delta^\eta \zeta)$. 
    Assumption \ref{assumption:a} then implies that
    \begin{align}\label{eq:2terms_0}
        (D_t+H_\delta) \psi_\delta (t,x) = e^{i(K\cdot x-Et)} \left( \delta^{1+\rate} \ow \tilde{b}_{1} [\delta^{\frac d2}\phi(\delta t,\cdot \; ;\delta)](\delta x) + \ow \tilde{b}_{1}^{{\rm out}} [\delta^{\frac d2}\phi(\delta t,\cdot\; ;\delta)](\delta x) \right),
    \end{align}
    where
    \begin{align*}
        \tilde{b}_{1} (X,\zeta; \delta) := \chi_\delta (\zeta) b_{1}(X,\zeta; \delta), \quad \tilde{b}_{1}^{{\rm out}} (X,\zeta;\delta) := (1-\chi_\delta (\zeta))\left( a(X,K+\delta \zeta;\delta) - E I_n - \delta b_{0}(X, \zeta; \delta)\right).
    \end{align*}
    With $N' \in \mathbb{N}$ and $\hat{\zeta} := \zeta/|\zeta|$ the unit vector in the direction of $\zeta$,
    we next write
    \begin{align*}
        \ow \tilde{b}_{1}^{{\rm out}} [\delta^{\frac d2}\phi(\delta t,\cdot\; ;\delta)](\delta x) &= \frac{1}{(2\pi)^d} \int_{\mathbb{R}^{d} \times \{|\zeta| > \frac{1}{2}C_0\delta^{-\eta}\}}e^{i(\delta x-\delta y)\cdot \zeta} (1-\chi_\delta(\zeta)) \\
        & \qquad \times \left( a(\frac{\delta x+\delta y}2,K+\delta \zeta;\delta) - E I_n - \delta b_{0}(\frac{\delta x+\delta y}2, \zeta; \delta)\right)
        \delta^{\frac d2} \phi(\delta t,\delta y;\delta) {\rm d} \delta y {\rm d}\zeta\\
        &= \frac{1}{(2\pi)^d} \int_{\mathbb{R}^{d} \times \{|\zeta| > \frac{1}{2}C_0\delta^{-\eta}\}}\left(\frac{i}{|\zeta|} \hat{\zeta} \cdot \nabla_{\delta y} \right)^{N'}\left[ e^{i(\delta x-\delta y)\cdot \zeta}\right] (1-\chi_\delta(\zeta)) \\
        & \qquad \times \left( a(\frac{\delta x+\delta y}2,K+\delta \zeta;\delta) - E I_n - \delta b_{0}(\frac{\delta x+\delta y}2, \zeta; \delta)\right)
        \delta^{\frac d2} \phi(\delta t,\delta y;\delta) {\rm d} \delta y {\rm d}\zeta\\
        &= \frac{1}{(2\pi)^d} \int_{\mathbb{R}^{d} \times \{|\zeta| > \frac{1}{2}C_0\delta^{-\eta}\}}e^{i(\delta x-\delta y)\cdot \zeta}(1-\chi_\delta(\zeta)) \\
        & \hspace{-2.5cm} \times \left(-\frac{i}{|\zeta|} \hat{\zeta} \cdot \nabla_{\delta y} \right)^{N'}\left[ \left( a(\frac{\delta x+\delta y}2,K+\delta \zeta;\delta) - E I_n - \delta b_{0}(\frac{\delta x+\delta y}2, \zeta; \delta)\right)
        \delta^{\frac d2} \phi(\delta t,\delta y;\delta)\right] {\rm d} \delta y {\rm d}\zeta,
    \end{align*}
    where we have integrated by parts in $\delta y$ to obtain the last equality.
    Note that by \eqref{eq:reg_a}, for any $\alpha, \beta \in \mathbb{N}_0^d$,
    \begin{align*}
        \left| \partial_{\delta y}^\alpha \partial_\zeta^\beta \left( (1-\chi_\delta (\zeta))\left(-\frac{i}{|\zeta|} \hat{\zeta} \cdot \nabla_{\delta y} \right)^{N'} a(\frac{\delta x+\delta y}2,K+\delta \zeta;\delta) \right) \right| \le C_{\alpha,\beta,N'} \delta^{\eta N'-\nu_0-\nu_1 (|\alpha|+N') - \nu_2 |\beta|},
    \end{align*}
    while the fact that $b_{0} \in S^1$ implies that
    \begin{align*}
        \left| \partial_{\delta y}^\alpha \partial_\zeta^\beta \left( (1-\chi_\delta (\zeta))\left(-\frac{i}{|\zeta|} \hat{\zeta} \cdot \nabla_{\delta y} \right)^{N'} \left( E I_n + \delta b_{0} (\frac{\delta x+\delta y}2, \zeta;\delta) \right) \right) \right| \le C_{\alpha,\beta,N'} \delta^{\eta N'} \aver{\zeta}
    \end{align*}
    uniformly in $(x,y,\zeta; \delta) \in \mathbb{R}^{3d} \times (0, \delta_0]$. Therefore, we apply Lemma \ref{lemma:S_bd} with $|\alpha+\beta| = \ncv$ to conclude that
    \begin{align*}
        \| \ow \tilde{b}_{1}^{{\rm out}} [\phi(\delta t,\cdot\; ; \delta)]\|_{H^N} \le C_{N,N'} \delta^{(\eta - \nu_1) N'-\nu_0 - \nu (\ncv+N)}\| \phi (\delta t, \cdot \; ;\delta)\|_{H^{N+N'+1}}
    \end{align*}
    for any $N \ge 0$,
    where $\nu = \max\{ \nu_1, \nu_2\}$.
    Recalling that $N'$ was arbitrary and $\eta - \nu_1 > 0$ by Assumption \ref{assumption:a}, and applying Lemma \ref{lemma:Sobolev_0}, 
    we conclude that for any $s > 0$ there exists a positive constant $C_s$ such that
    \begin{align*}
        \| \ow \tilde{b}_{1}^{{\rm out}} [\phi(\delta t,\cdot\; ; \delta)]\|_{H^N} \le C_s \delta^{s}, 
        \qquad t \ge 0, \quad 0 < \delta \le \delta_0. 
    \end{align*}
    On the other hand, since $\chi$ is smooth and compactly supported, $\tilde{b}_{1} (\cdot, \cdot\; ; \delta) \in S(\aver{\zeta}^{2})$ uniformly in $0 < \delta \le \delta_0$. Thus Lemmas \ref{lemma:S_bd} and \ref{lemma:Sobolev_0} imply that
    \begin{align*}
        \| \ow \tilde{b}_{1} [\phi(\delta t,\cdot\; ; \delta)]\|_{H^N} \le C \| \phi(\delta t,\cdot\; ; \delta)\|_{H^{N+2}} \le C, \qquad t \ge 0, \quad 0 < \delta \le \delta_0.
    \end{align*}
    We combine \eqref{eq:2terms_0} with the above estimates to complete the proof.
\end{proof}

\subsection{Proof of Theorem \ref{thm:main2}}\label{subsec:proof_main_2}
This proof is similar to that of Theorem \ref{thm:main_0}; the main additional challenge is to handle the spatial dependence of the degenerate point. We again begin by establishing the regularity of the solution to the effective equation.
\begin{lemma}\label{lemma:Sobolev_0_2}
    Suppose Assumption \ref{assumption:a2} holds, and let $\phi_0 \in \cS (\mathbb{R}^d; \mathbb{C}^n)$.
    For any $N \ge 0$, there exists $C > 0$ such that the solution $\phi$ to \eqref{eq:macro_0} satisfies $\norm{\phi (T,\cdot \; ; \delta)}_{H^N} \le C$ uniformly in $T \ge 0$ and $0 < \delta \le \delta_0$.
\end{lemma}
\begin{proof}
    The operator $\bH = \ow b_0$ satisfies \eqref{eq:equivalent_norms}, as the regularity conditions on $b_0$ in Assumptions \ref{assumption:a} and \ref{assumption:a2} are exactly the same. The result then follows from Lemma \ref{lemma:Sobolev_0_appendix}.
\end{proof}
\begin{lemma}\label{lemma:res_2}
    Suppose Assumption \ref{assumption:a2} holds, and take $H_\delta$ and $\psi_\delta$ as in \eqref{eq:H_delta} and \eqref{eq:ansatz_2}. Then for any $\eps > 0$,
    \begin{align*}
        \norm{(D_t + H_\delta) \psi_\delta (t, \cdot)}_{H^N} \le C_N \delta^{1+\rate -\eps}, \qquad t \ge 0, \quad 0 < \delta \le \delta_0, \quad N \ge 0.
    \end{align*}
\end{lemma}
\begin{proof}
    Let $0< \nu < \eta$ and 
    $\chi_\delta (z) := \chi (\delta^\nu z)$, and define the operator $H_{00}^\delta$ by
    \begin{align*}
        H_{00}^\delta f(x) = \dint_{\Rm^{2d}} \chi_\delta (z) e^{iz \cdot \xi} a (\delta x - \frac{\delta z}2,\xi; \delta) f(x-z) \dfrac{dz d\xi}{(2\pi)^d}.
    \end{align*}
    Again,
    $u_\delta := \psi_\delta - \varphi_\delta$ satisfies
    \(
        (D_t + H_\delta) u_\delta(t,x) = (D_t + H_\delta)\psi_\delta (t,x),\) and  \(u_\delta (0, \cdot)= 0\)
    where by Lemma \ref{lemma:chi},
    \begin{align*}
        H_\delta \psi_\delta(t,x) &= H^\delta_{00} \psi_\delta (t,x) + r_\delta^0 (t,x),\\
        \norm{r_\delta^0 (t, \cdot)}_{H^{N}} &\le C_{n,N} \delta^n \norm{\psi_{\delta} (t, \cdot)}_{H^{N+d+1}} \le C_{n,N} \delta^n \norm{\phi (\delta t, \cdot \; ; \delta)}_{H^{N+d+1}}.
    \end{align*}
    Note that the last inequality in the bound for $r_\delta$ follows from the uniform boundedness (in $\delta$) of the map $x \mapsto e^{i(\frac{1}{\delta} \vp (\delta x) -Et)}$ and all of its derivatives. 
    Introduce $g_\delta (X, Z) := e^{\frac{i}{\delta} (Z \cdot \nabla B( X - \frac{Z}{2}) + B (X-Z) - B(X))}$
   and  write
    \begin{align}\label{eq:H00}
    \begin{split}
        & H_{00}^\delta \psi_\delta (t,x) = e^{-iEt} \dint_{\Rm^{2d}} \chi_\delta (z) e^{iz \cdot \xi} a (\delta x - \frac{\delta z}2,\xi; \delta) e^{\frac{i}{\delta} \vp (\delta (x-z))} \delta^{d/2} \phi (\delta t, \delta (x-z); \delta) \dfrac{dz d\xi}{(2\pi)^d}
        \\&
        = e^{i (\frac{1}{\delta} \vp (\delta x)-Et)} \dint_{\Rm^{2d}} \chi_\delta (z) e^{iz \cdot (\xi-K)} a (\delta x - \frac{\delta z}2,\xi; \delta) 
        e^{\frac{i}{\delta} (B(\delta (x-z))- B(\delta x))} \delta^{d/2} \phi (\delta t, \delta (x-z); \delta) \dfrac{dz d\xi}{(2\pi)^d}\\
        &= e^{i (\frac{1}{\delta} \vp (\delta x)-Et)} \dint_{\Rm^{2d}} \chi_\delta (z) e^{iz \cdot (\xi-\dpx (\delta x - \frac{\delta z}{2}))} a (\delta x - \frac{\delta z}2,\xi; \delta) 
        g_\delta (\delta x, \delta z) \delta^{d/2} \phi (\delta t, \delta (x-z); \delta) \dfrac{dz d\xi}{(2\pi)^d}.
    \end{split}
    \end{align}
    We will now show that $g_\delta$ can be approximated by $1$ in the above integrand, up to $O (\delta^{2 - 3\nu})$ error.
    By Taylor's theorem, we have
    \begin{align*}
        B(X) &= B(X - \frac{\delta z}{2}) + \frac{\delta z}{2} \cdot \nabla B (X - \frac{\delta z}{2}) + \frac{\delta^2}{4} z \cdot \nabla^2 B (X - \frac{\delta z}{2}) z + \delta^3 s_1 (X, z; \delta),\\
        B(X - \delta z) &= B(X - \frac{\delta z}{2}) - \frac{\delta z}{2} \cdot \nabla B (X - \frac{\delta z}{2}) + \frac{\delta^2}{4} z \cdot \nabla^2 B (X - \frac{\delta z}{2}) z + \delta^3 s_2 (X, z; \delta),
    \end{align*}
    where the $s_j$ are smooth functions of $(X,z)$ satisfying
    \begin{align*}
        |\partial^\alpha_X \partial^\beta_z s_j (X,z;\delta)| \le C_{\alpha,\beta} \aver{z}^3 \le C_{\alpha,\beta} \delta^{-3 \nu}, \qquad X \in \mathbb{R}^d, \quad z \in \operatorname{supp}(\chi_0^\delta), \quad 0 < \delta \le \delta_0
    \end{align*}
    for any multi-indices $\alpha, \beta \in \mathbb{N}_0^d$.
    Writing $g_\delta (\delta x, \delta z) = e^{i \delta^2 (s_2 (\delta x,z;\delta) - s_1 (\delta x,z;\delta))}$, it follows that
    \begin{align*}
       g_\delta (\delta x,\delta z) - 1 = \delta^{2-3\nu} R (\delta x,z;\delta), 
       \qquad x \in \mathbb{R}^d, \quad z \in \operatorname{supp}(\chi_0^\delta), \quad 0 < \delta \le \delta_0
    \end{align*}
    for some function $R$ satisfying
    \begin{align*}
        |\partial^\alpha_X \partial^\beta_z R (X,z;\delta)| \le C_{\alpha,\beta}, \qquad X \in \mathbb{R}^d, \quad z \in \operatorname{supp}(\chi_0^\delta), \quad 0 < \delta \le \delta_0.
    \end{align*}
    Therefore, using that $\nabla_x (a (\delta x - \frac{\delta z}{2})) = \delta \nabla_X a (\delta x - \frac{\delta z}{2})$ is bounded by $C\delta \aver{\xi}$, we conclude that
    \(H^\delta_{00} \psi_\delta (t,x) = \fri_\delta^0 (t,x) + r_\delta^1 (t,x),\)
    where
    \begin{align}\label{eq:fri_0_delta}
        \fri_\delta^0 (t,x) &:= e^{i (\frac{1}{\delta} \vp (\delta x)-Et)} \dint_{\Rm^{2d}} \chi_\delta (z) e^{iz \cdot (\xi-\dpx (\delta x - \frac{\delta z}{2}))} a (\delta x - \frac{\delta z}2,\xi; \delta)\delta^{d/2} \phi (\delta t, \delta (x-z); \delta) \dfrac{dz d\xi}{(2\pi)^d}
    \end{align}
    replaces $g_\delta$ by $1$ in \eqref{eq:H00}, and
    \(\norm{r_\delta^1 (t, \cdot)}_{H^N} \le C_N \delta^{2-3\nu} \norm{\phi (\delta t, \cdot \; ; \delta)}_{H^{N+d+1}}\).
    Note that we again had to integrate by parts in $z$ to obtain decay in $\xi$, which is why we lose $d+1$ derivatives.
    Applying Lemma \ref{lemma:chi} to the symbol $a_\delta (x, \xi + \dpx (\delta x - \frac{\delta z}{2}))$, we obtain that
    \(
        \fri_\delta^0 (t,x) = \fri_\delta (t,x) + r_\delta^2 (t,x),
    \)
    where
    \begin{align*}
        \fri_\delta (t,x) &:= e^{i (\frac{1}{\delta} \vp (\delta x)-Et)} \dint_{\Rm^{2d}} e^{iz \cdot (\xi-\dpx (\delta x - \frac{\delta z}{2}))} a (\delta x - \frac{\delta z}2,\xi; \delta)\delta^{d/2} \phi (\delta t, \delta (x-z); \delta) \dfrac{dz d\xi}{(2\pi)^d}
    \end{align*}
    replaces $\chi_\delta$ in \eqref{eq:fri_0_delta} by $1$, and
    \begin{align*}
        \norm{r^2_\delta (t, \cdot)}_{H^N} \le C_{n,N} \delta^n \norm{\phi (\delta t, \cdot \; ; \delta)}_{H^{N+d+1}}, \qquad n, N \in \mathbb{N}_0, \quad 0 < \delta \le \delta_0.
    \end{align*}
    Thus we have shown that
    \begin{align*}
        H_\delta \psi_\delta (t,x) = \fri_\delta (t,x) + \sum_{j=0}^2 r_\delta^j (t,x), \qquad \Big\|\sum_{j=0}^2 r_\delta^j (t,\cdot)\Big\|_{H^N} \le C_N \delta^{2-3\nu} \Big\|\phi (\delta t, \cdot \; ; \delta)\Big\|_{H^{N+d+1}}.
    \end{align*}
We now evaluate
\begin{align}\label{eq:Dt}  
        &D_t \psi_\delta (t,x) = -E \psi_\delta - \delta e^{i (\frac{1}{\delta} \vp (\delta x) - Et)} \delta^{d/2} \bH \phi (\delta t, \delta x; \delta) = \\
        & -e^{i (\frac{1}{\delta} \vp (\delta x) - Et)} \dint_{\Rm^{2d}} e^{iz \cdot (\xi-\dpx (\delta x - \frac{\delta z}{2}))} \left(E + \delta b_{0} \left(\delta x - \frac{\delta z}2,\frac{\xi-\dpx (\delta x - \frac{\delta z}{2})}{\delta}; \delta\right)\right)
        \delta^{d/2} \phi (\delta t, \delta (x-z); \delta) \dfrac{dz d\xi}{(2\pi)^d},
       \nonumber
\end{align} 
which implies that
\begin{align*}
        D_t \psi_\delta (t,x) &+ \fri_\delta (t,x) = e^{i (\frac{1}{\delta} \vp (\delta x) - Et)} \dint_{\Rm^{2d}} e^{iz \cdot (\xi-\dpx (\delta x - \frac{\delta z}{2}))}\\
        &\hspace{0cm} \left(a (\delta x - \frac{\delta z}{2}, \xi; \delta) - E - \delta b_{0} \left(\delta x - \frac{\delta z}2,\frac{\xi-\dpx (\delta x - \frac{\delta z}{2})}{\delta}; \delta\right)\right)\delta^{d/2} \phi (\delta t, \delta (x-z); \delta) \dfrac{dz d\xi}{(2\pi)^d}.
    \end{align*}
    After changing variables $\zeta := \frac{\xi-\dpx (\delta x - \frac{\delta z}{2})}{\delta}$, this becomes
    \begin{align*}
        D_t \psi_\delta (t,x) &+ \fri_\delta (t,x) = e^{i (\frac{1}{\delta} \vp (\delta x) - Et)} \dint_{\Rm^{2d}} e^{i \delta z \cdot \zeta}\\
        &\hspace{0cm} \left(a (\delta x - \frac{\delta z}{2},\dpx (\delta x - \frac{\delta z}{2}) + \delta \zeta; \delta) - E - \delta b_{0} \left(\delta x - \frac{\delta z}2,\zeta; \delta\right)\right)\delta^{d/2} \phi (\delta t, \delta (x-z); \delta) \dfrac{\delta^d dz d\zeta}{(2\pi)^d}.
    \end{align*}
    Define $\chi_\delta^0 := \chi_\delta$ and $\chi_\delta^1 := 1 - \chi_\delta$, and for $j=0,1$ set
    \begin{align*}
        \Delta_\delta^j (t,x) 
        &:= e^{i (\frac{1}{\delta} \vp (\delta x) - Et)} \dint_{\Rm^{2d}} e^{i z \cdot \zeta} \chi_\delta^j (\zeta) \\
        &\hspace{0cm} \left(a (\delta x - \frac{z}{2},\dpx (\delta x - \frac{z}{2}) + \delta \zeta; \delta) - E - \delta b_{0} \left(\delta x - \frac{z}2,\zeta; \delta\right)\right)\delta^{d/2} \phi (\delta t, \delta x-z; \delta) \dfrac{dz d\zeta}{(2\pi)^d},
    \end{align*}
    so that $D_t \psi_\delta + \fri_\delta = \Delta^0_\delta + \Delta^1_\delta$. Using that $\nu < \eta$, it follows from the decomposition \eqref{eq:decomposition2} that
    \begin{align*}
        \Delta_\delta^0 (t,x) &= \delta^{1+\rate} e^{i (\frac{1}{\delta} \vp (\delta x) - Et)} \dint_{\Rm^{2d}} e^{i z \cdot \zeta} \chi_\delta (\zeta) b_{1} \left(\delta x - \frac{z}2,\zeta; \delta\right)\delta^{d/2} \phi (\delta t, \delta x-z; \delta) \dfrac{dz d\zeta}{(2\pi)^d}.
    \end{align*}
    Using that $$\left| \partial^\alpha_x \partial^\beta_\zeta \left( b_{1} \left(\delta x - \frac{z}2,\zeta; \delta\right)\right)\right| \le C_{\alpha,\beta} \delta^{|\alpha|} \aver{\zeta}^{2+|\alpha|},$$
    it follows from standard arguments (involving integration by parts in $\zeta$ to obtain decay of the integrand in $z$, and recalling that the integrand is compactly supported in $\zeta$) that
    \begin{align*}
        \norm{\Delta_\delta^0 (t, \cdot)}_{H^N} &\le 
        C \delta^{1+\rate} \norm{\phi (\delta t, \cdot \; ; \delta))}_{H^N} \sup \{ C_{\alpha,\beta} \delta^{|\alpha|} \aver{\zeta}^{2+|\alpha|} \; : \; |\alpha| \le N, \, |\beta| \le d+1, \, \zeta \in \supp \chi^0_\delta\}\\
        &\le C \delta^{1+\rate -2\nu} \norm{\phi (\delta t, \cdot \; ; \delta))}_{H^N}.
    \end{align*}
    It remains to control $\Delta^1_\delta = \Delta^{1,0}_\delta - \Delta^{1,1}_\delta$, where
    \begin{align*}
        \Delta^{1,0}_\delta (t,x) &:= e^{i (\frac{1}{\delta} \vp (\delta x) - Et)} \dint_{\Rm^{2d}} e^{i z \cdot \zeta} \chi_\delta^1 (\zeta) a (\delta x - \frac{z}{2},\dpx (\delta x - \frac{z}{2})+\delta\zeta;\delta)\delta^{d/2} \phi (\delta t, \delta x-z; \delta) \dfrac{dz d\zeta}{(2\pi)^d},\\
        \Delta^{1,1}_\delta (t,x) &:= e^{i (\frac{1}{\delta} \vp (\delta x) - Et)} \dint_{\Rm^{2d}} e^{i z \cdot \zeta} \chi_\delta^1 (\zeta) (E+\delta b_{0} (\delta x - \frac{z}{2}, \zeta; \delta))\delta^{d/2} \phi (\delta t, \delta x-z; \delta) \dfrac{dz d\zeta}{(2\pi)^d}.
    \end{align*}
    By assumption on $a$, for $\Delta_\delta^{1,0}$ it suffices to consider
    \begin{align*}
        \tilde{\Delta}_\delta^{1,0} (t,x) &:= e^{i (\frac{1}{\delta} \vp (\delta x) - Et)} \dint_{\Rm^{2d}} e^{i z \cdot \zeta} \chi_\delta^1 (\zeta) e^{i \phase (\delta x - \frac{z}{2},\dpx (\delta x - \frac{z}{2})+\delta\zeta;\delta)} \\
        &\hspace{3cm} \sym (\delta x - \frac{z}{2},\dpx (\delta x - \frac{z}{2})+\delta\zeta;\delta)\delta^{d/2} \phi (\delta t, \delta x-z; \delta) \dfrac{dz d\zeta}{(2\pi)^d},
    \end{align*}
    where $\phase \in S^1$ and $\sym \in S(1)$ uniformly in $\delta$. Since $\dpx$ is bounded, we know that
    \begin{align*}
        |\nabla_z (\phase (\delta x - \frac{z}{2},\dpx (\delta x - \frac{z}{2})+\delta\zeta;\delta))| \le C \aver{\delta \zeta}, 
    \end{align*}
    and thus there exists a constant $c>0$ such that
    \begin{align*}
        |\nabla_z Q| \ge c \aver{\zeta} - c^{-1}, \qquad Q := z \cdot \zeta + \phase (\delta x - \frac{z}{2},\dpx (\delta x - \frac{z}{2})+\delta\zeta;\delta)).
    \end{align*}
    Hence for sufficiently small $\delta$ (recalling that $\chi_\delta^1$ vanishes in a large neighborhood of the origin), we have
    \begin{align*}
        \tilde{\Delta}_\delta^{1,0} (t,x) &:= e^{i (\frac{1}{\delta} \vp (\delta x) - Et)} \dint_{\Rm^{2d}} \chi_\delta^1 (\zeta) \sym (\delta x - \frac{z}{2},\dpx (\delta x - \frac{z}{2})+\delta\zeta;\delta)\delta^{d/2} \phi (\delta t, \delta x-z; \delta) 
        {\mathfrak V}^n (e^{iQ})\dfrac{dz d\zeta}{(2\pi)^d},
    \end{align*}
    using ${\mathfrak V}=\frac{-i}{|\nabla_z Q|^2}\nabla_z Q \cdot \nabla_z$,
    which after integrating by parts becomes
    \begin{align*}
        \tilde{\Delta}_\delta^{1,0} (t,x) &:= e^{i (\frac{1}{\delta} \vp (\delta x) - Et)} \dint_{\Rm^{2d}} e^{iQ} \chi_\delta^1 (\zeta) {\mathfrak V}^n  
        \left( \sym (\delta x - \frac{z}{2},\dpx (\delta x - \frac{z}{2})+\delta\zeta;\delta)\delta^{d/2} \phi (\delta t, \delta x-z; \delta) \right) \dfrac{dz d\zeta}{(2\pi)^d}.
    \end{align*}
    Similarly, using that $(-i\frac{z}{|z|^2} \cdot \nabla_\zeta)^{d+1} e^{iz \cdot \zeta} = e^{iz \cdot \zeta}$ when $z \ne 0$, we integrate by parts in $\zeta$ to establish that the above integrand decays super-algebraically in $z$ (uniformly in $\delta$). We conclude that
    \begin{align}\label{eq:tD}
    \begin{split}
        \norm{\tilde{\Delta}_\delta^{1,0} (t,\cdot)}_{H^N}^2 &\le C \delta^d \int_{\mathbb{R}^{5d}} \chi_\delta^1 (\zeta) \chi_\delta^1 (\zeta') \aver{\zeta}^{-n} \aver{\zeta'}^{-n} \aver{z}^{-d-1} \aver{z'}^{-d-1} \\
        &\qquad \sup_{|\alpha| \le N+n} |\partial^\alpha_z \phi (\delta t, \delta x - z; \delta)|
        \sup_{|\alpha'| \le N+n} |\partial^{\alpha'}_{z'} \phi (\delta t, \delta x - z'; \delta)| dx dz dz' d\zeta d \zeta '\\
        &\le C \norm{\phi (\delta t, \cdot; \delta)}_{H^{N+n}}^2 \left( \int_{\mathbb{R}^d} \chi^1_\delta (\zeta) \aver{\zeta}^{-n} d \zeta \right)^2 \le C \delta^{2\nu (n-d)} \norm{\phi (\delta t, \cdot; \delta)}_{H^{N+n}}^2
        \end{split}
    \end{align}
    for any $n \ge d+1$.

    A parallel argument establishes that for any $n \ge d+2$ and any $N \in \mathbb{N}_0$, 
    \begin{align*}
        \norm{\Delta_\delta^{1,1} (t,\cdot)}_{H^N}^2 \le C \delta^{2\nu (n-d-1)} \norm{\phi (\delta t, \cdot; \delta)}_{H^{N+n}}^2.
    \end{align*}
    Compared with \eqref{eq:tD}, this comes with a slightly worse rate of convergence due to the fact that $b_{0} \in S^1$ is not bounded in $\xi$.

    Putting together the above estimates, we have shown that for any $n \ge d+2$ and $N \in \mathbb{N}_0$,
    \begin{align*}
        \norm{D_t \psi_\delta (t, \cdot) + H_\delta \psi_\delta (t, \cdot)}_{H^N} &\le C (\delta^{1+\rate -2\nu}\norm{\phi (\delta t, \cdot \; ; \delta))}_{H^N}\\
        &\hspace{2cm} + \delta^{2-3\nu} \norm{\phi (\delta t, \cdot \; ; \delta))}_{H^{N+d+1}} + \delta^{\nu (n-d-1)} \norm{\phi (\delta t, \cdot; \delta)}_{H^{N+n}}).
    \end{align*}
    By Lemma 
    \ref{lemma:Sobolev_0_2}, all Sobolev norms of $\phi$ are bounded uniformly in time, hence
    \begin{align*}
        \norm{D_t \psi_\delta (t, \cdot) + H_\delta \psi_\delta (t, \cdot)}_{H^N} &\le C \delta^{1+\rate -2\nu}.
    \end{align*}
    Since $\nu$ was arbitrary, the result is complete. 
\end{proof}

\begin{proof}[Proof of Theorem \ref{thm:main2}]
By Lemma \ref{lemma:res_2}, $u_\delta := \psi_\delta - \varphi_\delta$ satisfies $u_\delta (0, \cdot) \equiv 0$ and
    \begin{align*}
        \norm{(D_t + H_\delta) u_\delta (t, \cdot)}_{H^N}\le C_N \delta^{1+\rate-\eps}, \qquad t \ge 0, \quad 0 < \delta \le \delta_0, \quad N \ge 0.
    \end{align*}
    The result then follows from Assumption \ref{assumption:commutator2} and Lemma \ref{lemma:u_delta} (with $M=0$ and $\mu \leftarrow \mu -\eps$ there).
\end{proof}

\subsection{Proof of Theorem \ref{thm:main}}\label{subsec:pf_higher_order}
We begin as in the previous sections, establishing the regularity of the solution to the higher-order effective equation. 
\begin{lemma}\label{lemma:Sobolev_bd_higher_order}
    Suppose Assumption \ref{assumption:a_higher_order} holds, and let $\phi_0 \in \cS (\mathbb{R}^d; \mathbb{C}^n)$.
    Fix $N \ge 0$ and 
    set $M = \lceil N/\rate \rceil$. Then 
    there exists $C > 0$ such that $\norm{\phi (T,\cdot \; ; \delta)}_{H^N} \le C (1 + T^{M})$ uniformly in $T \ge 0$ and $0 < \delta \le \delta_0$.
\end{lemma}
\begin{proof}
    Set $B := \bH_p - \bH$ with $\bH$ defined in \eqref{eq:macroH_0}, and define the sequence $(\phi_n)_{n \in \mathbb{N}}$ by
    \[ (D_T+\bH)\phi_1=0,\quad \phi_1(0, \cdot)=\phi_0,\qquad (D_T+\bH)\phi_{n+1}+B\phi_n=0,\quad \phi_{n+1}(0, \cdot)\equiv 0,\quad n\geq1.\]
    It follows from \eqref{eq:equivalent_norms} and Lemma \ref{lemma:Sobolev_0_appendix} with $\Lambda = (i+\bH)^q$ that for any $q \ge 0$, $$\norm{\phi_1 (T, \cdot)}_{H^q} \le C, \qquad \norm{\phi_{n+1} (T, \cdot)}_{H^q} \le C T \sup_{0 \le S \le T} \norm{B \phi_n (S, \cdot)}_{H^q},\quad n\geq1$$ uniformly in $T \ge 0$ and $0 < \delta \le \delta_0$. Assumption \ref{assumption:a_higher_order} implies that for any $q \ge 0$, the operator $B$ has norm bounded by $C\delta^\rate$ as a map from $H^{q+p+1}$ to $H^q$. It follows from induction that
    \begin{align*}
        \norm{\phi_{n+1} (T, \cdot)}_{H^q} \le C (\delta^\mu T)^n, \qquad T \ge 0, \quad 0 < \delta \le \delta_0
    \end{align*}
    for all $n \in \mathbb{N}$ and $q \ge 0$.
    Now, fix $M \in \mathbb{N}$, define $u_M := \phi - \sum_{j=1}^M \phi_j$, and observe that
    \begin{align*}
        (D_T + \bH_p) u_M = -B \phi_M, \qquad u_M (0, \cdot) \equiv 0.
    \end{align*}
    Applying \eqref{eq:equivalent_norms_p} and Lemma \ref{lemma:Sobolev_0_appendix} with $\Lambda = (i+\bH_p)^{N/p}$ and $C_1 = c \delta^N$ and $C_2 = C$ there, 
    it follows that
    \begin{align*}
        \norm{u_M (T, \cdot)}_{H^N} \le C \delta^{-N} T \sup_{0 \le S \le T} \norm{B \phi_M (S, \cdot)}_{H^N} \le C \delta^{-N} \delta^\mu T \sup_{0 \le S \le T} \norm{\phi_M (S, \cdot)}_{H^{N+p+1}} \le C \delta^{-N} (\delta^\mu T)^M
    \end{align*}
    uniformly in $T \ge 0$ and $0 < \delta \le \delta_0$. Choosing $M \ge N/\rate$, it follows that
    \begin{align*}
        \norm{u_M (T, \cdot)}_{H^N} \le C T^M, \qquad T \ge 0, \quad 0 < \delta \le \delta_0.
    \end{align*}
    We have thus verified that
    \begin{align*}
        \norm{\phi (T, \cdot)}_{H^N} \le \norm{u_M (T, \cdot)}_{H^N} + \sum_{j=1}^M \norm{\phi_j (T, \cdot)}_{H^N} \le C (T^M + 1),
    \end{align*}
    and the proof is complete.
\end{proof}
\begin{lemma}\label{lemma:res_higher_order}
    Suppose Assumption \ref{assumption:a_higher_order} holds. Take $H_\delta$ as in \eqref{eq:microWP} and $\psi_\delta$ as in \eqref{eq:psi_def}. Let $N \ge 0$, take $\ncv$ as in Lemma \ref{lemma:S_bd}, set $\nu = \max\{\nu_1, \nu_2\}$, and define $$\nbp^{N} := \lceil (\rate + p + \nu_0 + \nu (\ncv + N))/(\eta - \nu_1) \rceil, \quad \mbp^N := \lceil (N + \max\{\nbp^N,p\} + 1)/\rate\rceil.$$ Then
    \begin{align*}
        \norm{(D_t + H_\delta) \psi_\delta (t, \cdot)}_{H^N} \le C_N \delta^{p+\rate} (1 + (\delta t)^{\mbp^N}), \qquad t \ge 0, \quad 0 < \delta \le \delta_0, \quad N \ge 0.
    \end{align*}
\end{lemma}
\begin{proof}
    Following the proof of Lemma \ref{lemma:res}, we find that
    \begin{align}\label{eq:2terms_higher_order}
        (D_t+H_\delta) \psi_\delta (t,x) = e^{i(K\cdot x-Et)} \left( \delta^{p+\rate} \ow \tilde{b}_{1} [\delta^{\frac d2}\phi(\delta t,\cdot \; ;\delta)](\delta x) + \ow \tilde{b}_{1}^{{\rm out}} [\delta^{\frac d2}\phi(\delta t,\cdot\; ;\delta)](\delta x) \right),
    \end{align}
    where
    \begin{align*}
        \tilde{b}_{1} (X,\zeta; \delta) := \chi_\delta (\zeta) b_{1p}(X,\zeta; \delta), \quad \tilde{b}_{1}^{{\rm out}} (X,\zeta;\delta) := (1-\chi_\delta (\zeta))\left( a(X,K+\delta \zeta;\delta) - E I_n - \delta b_{0p}(X, \zeta; \delta)\right).
    \end{align*}
    As before, for any $N, N' \ge 0$ and with $\ncv$ as in Lemma \ref{lemma:S_bd}, we have
    \begin{align*}
        \| \ow \tilde{b}_{1}^{{\rm out}} [\phi(\delta t,\cdot\; ; \delta)]\|_{H^N} \le C_{N,N'} \delta^{(\eta - \nu_1) N'-\nu_0 - \nu (\ncv+N)}\| \phi (\delta t, \cdot \; ;\delta)\|_{H^{N+N'+1}}, \qquad t \ge 0, \quad 0 < \delta \le \delta_0,
    \end{align*}
    where $\nu = \max\{ \nu_1, \nu_2\}$. Choosing $N' := \nbp^N$ and applying Lemma \ref{lemma:Sobolev_bd_higher_order}, this becomes
    \begin{align}\label{eq:b1_out}
        \| \ow \tilde{b}_{1}^{{\rm out}} [\phi(\delta t,\cdot\; ; \delta)]\|_{H^N} 
        \le C_{N} \delta^{p+\rate} (1 + (\delta t)^{M_1}), \qquad t \ge 0, \quad 0 < \delta \le \delta_0,
    \end{align}
    where $M_1 := \lceil (N + \nbp^N + 1)/\rate\rceil$.
    Since $\chi$ is smooth and compactly supported, $\tilde{b}_{1} (\cdot, \cdot\; ; \delta) \in S(\aver{\zeta}^{p+1})$ uniformly in $0 < \delta \le \delta_0$. Hence Lemmas \ref{lemma:S_bd} and \ref{lemma:Sobolev_bd_higher_order} imply that
    \begin{align}\label{eq:b1_tilde}
        \| \ow \tilde{b}_{1} [\phi(\delta t,\cdot\; ; \delta)]\|_{H^N} \le C \| \phi(\delta t,\cdot\; ; \delta)\|_{H^{N+p+1}} \le C (1 + (\delta t)^{M_2}), \qquad t \ge 0, \quad 0 < \delta \le \delta_0,
    \end{align}
    where $M_2 := \lceil (N + p + 1)/\rate\rceil$. 
    The result then follows from \eqref{eq:2terms_higher_order}-\eqref{eq:b1_out}-\eqref{eq:b1_tilde}.
\end{proof}
\begin{proof}[Proof of Theorem \ref{thm:main}]
    Set $u_\delta := \psi_\delta - \varphi_\delta$, then $u_\delta (0, \cdot) \equiv 0$ and $(D_t + H_\delta) u_\delta = (D_t + H_\delta) \psi_\delta$, hence
    \begin{align*}
        \norm{(D_t + H_\delta) u_\delta (t, \cdot)}_{H^N}\le C_N \delta^{p+\rate} (1 + (\delta t)^{\mbp^N}), \qquad t \ge 0, \quad 0 < \delta \le \delta_0, \quad N \ge 0
    \end{align*}
    by Lemma \ref{lemma:res_higher_order}.
    Since $H_\delta$ satisfies Assumption \ref{assumption:commutator}, the result follows immediately from Lemma \ref{lemma:u_delta_reg}.
\end{proof}

\subsection{Proof of Proposition \ref{prop:TBG}}\label{subsec:pf_TBG}
We first collect two useful boundedness properties of $\Hperp_\delta$.
\begin{lemma}\label{lemma:Hperp_bdd}
    Suppose \eqref{eq:assumption_hperp} and \eqref{assumption:theta} hold, and set $\nabla_x^\theta := \rot_\theta^\top \nabla_x$.
    Then for any $0 < \delta < \delta_0$ and $N \ge 0$, 
    $$\Hperp_\delta :H^N (\mathbb{R}^2; \mathbb{C}^2)\to H^N(\mathbb{R}^2; \mathbb{C}^2), \qquad 
    \nabla_x \Hperp_\delta - \Hperp_\delta \nabla_x^\theta : H^{N}(\mathbb{R}^2; \mathbb{C}^2) \to H^N(\mathbb{R}^2; \mathbb{C}^2 \oplus \mathbb{C}^2)$$ are bounded operators with
    \begin{align*}
        \norm{\Hperp_\delta}_{H^N \to H^N} \le C, \qquad \norm{\nabla_x \Hperp_\delta - \Hperp_\delta \nabla_x^\theta}_{H^{N} \to H^N} \le C \delta, \qquad 0 < \delta \le \delta_0/2.
    \end{align*}
\end{lemma}
Recall that derivatives of the symbol 
$$a_{12}^{\sigma \sigma'} (X, \xi; \delta) =\frac{1}{|\Gamma|} \sum_{q \in \lattice^*} e^{i \beta q \cdot \rot_{\pi/2} X} \hathperp^{\sigma \sigma'} (\xi -(1 - \frac{1}{4}\beta^2 \delta^2)^{1/2}q)$$ 
from \eqref{eq:a12_F} do not satisfy the decay requirements of Assumption \ref{assumption:a}.
We thus split up the above sum over $q$ into two parts, one of which satisfies Assumption \ref{assumption:a}, and the other one which is small. To make this precise, let $0 < \tpn < 1$ and define $\latticeIn := \{q \in \lattice^* : |q| < \delta^{-\tpn}\}$ and $\latticeOut := \lattice^* \setminus \latticeIn$. We then have
\begin{lemma}\label{lemma:a12}
    The symbol $a_{12}$ given by \eqref{eq:a12_F} can be decomposed as $a_{12} = \atrunc + \ares$, where
    \begin{align*}
        (\atrunc)^{\sigma \sigma'} (X,\xi; \delta) &=\frac{1}{|\Gamma|} \sum_{q \in \latticeIn} e^{i \beta q \cdot \rot_{\pi/2} X} \hathperp^{\sigma \sigma'} (\xi -(1 - \frac{1}{4}\beta^2 \delta^2)^{1/2}q)
        \\
        \ares^{\sigma \sigma'}(X,\xi; \delta) &= \frac{1}{|\Gamma|} \sum_{q \in \latticeOut} e^{i \beta q \cdot \rot_{\pi/2} X} \hathperp^{\sigma \sigma'} (\xi -(1 - \frac{1}{4}\beta^2 \delta^2)^{1/2}q).
    \end{align*}
    Defining $\adres (x, \xi; \delta) := \ares (\delta x, \xi; \delta)$, we have the bounds
    \begin{align}\nonumber
         \forall \alpha, \beta \in \mathbb{N}_0^2,\quad |\partial^\alpha_X \partial^\beta_\xi (\atrunc)^{\sigma \sigma'} (X, \xi; \delta)| & \le C_{\alpha,\beta} \delta^{-2\tpn - \tpn |\alpha|- (1-\tbgparam)|\beta|}, \qquad (X,\xi; \delta) \in \mathbb{R}^{4} \times (0, \delta_0/2]
         \\
         \forall q, N \ge 0,\quad 
         \norm{\ow \adres f}_{H^N (\mathbb{R}^2; \mathbb{C}^2)} &\le C \delta^q \norm{f}_{H^{N+m} (\mathbb{R}^2; \mathbb{C}^2)}, \quad 0 < \delta \le \frac12\delta_0, \ \ f \in H^{N+m}(\mathbb{R}^2; \mathbb{C}^2),
         \label{eq:ares_bd}
    \end{align}
    for some $m=m(q,N)$ sufficiently large.
\end{lemma}
We postpone the proofs of Lemmas \ref{lemma:Hperp_bdd} and \ref{lemma:a12} to Appendix \ref{subsec:lemma_a12}. With these lemmas in hand, we are now ready to establish the validity of our first-order continuum model for twisted bilayer graphene.
\begin{proof}[Proof of Proposition \ref{prop:TBG}]
Take $\atrunc$ and $\adres$ as in Lemma \ref{lemma:a12}. Set $\adtrunc (x, \xi) := \atrunc (\delta x, \xi; \delta)$ and 
$$
\Htrunc_\delta := \ow \atruncall_\delta, \quad \Hres_\delta := \ow \aresall_\delta \qquad \atruncall_\delta := \begin{pmatrix}
    a_\delta^{11} & \adtrunc\\
    (\adtrunc)^* & a_\delta^{22}
\end{pmatrix},\quad
    \aresall_\delta := \begin{pmatrix}
        0 & \adres\\
        \adres^* & 0
    \end{pmatrix}.
$$
so that $H_\delta = \Htrunc_\delta + \Hres_\delta$.
Observe that
\begin{align*}
    \Hres_\delta \psi_\delta (t,x) = 
    e^{-iEt} \Hres_\delta \check \psi_\delta (t,x), \qquad \check \psi_\delta (t,x) := 
    \delta^{d/2} e^{iK \cdot x} \phi (\delta t, \delta x; \delta),
\end{align*}
where for any $q \ge 0$, Lemma \ref{lemma:a12} implies the existence of some $m \ge 0$ such that
\begin{align}\label{eq:check}
    \norm{\Hres_\delta \check \psi_\delta (t, \cdot)}_{H^N} \le C \delta^q \norm{\check \psi_\delta (t, \cdot)}_{H^{N+m}} \le C \delta^{q} \norm{\phi (\delta t, \cdot \; ; \delta)}_{H^{N+m}} \le C \delta^q, \qquad t \ge 0, \quad 0 < \delta \le \delta_0/2,
\end{align}
where the last inequality follows from Lemma \ref{lemma:Sobolev_0}. 

We next control $(D_t + \Htrunc_\delta) \psi_\delta$, which will be done by applying Lemma \ref{lemma:res}.
Following the derivation of \eqref{eq:small_twist}-\eqref{eq:twist_corr}, 
the symbols $a_\delta^{jj}$ and
\[
b_0^{jj} (X,\zeta) =\frac{\sqrt{3} v}{2} t_1 (X) (\zeta_2 \sigma_1 - \zeta_1 \sigma_2) +
    \tM (X) \sigma_3 + (-1)^j t_1 (X) \frac{\pi}{\sqrt{3}} \beta \sigma_2
\]
satisfy Assumption \ref{assumption:a} for $\rate = 1$ with respect to the degenerate point $K =-\frac{4\pi}{3v} (0,1)$. Note that compared to \eqref{eq:twist_corr}, the third term on the above right-hand side acquired a factor of $(-1)^j/2$ since the twist angle for layer $j$ is $\theta_j = (-1)^j \theta/2 = (-1)^j \beta \delta /2 + O(\delta^2)$.

It remains to show that the off-diagonal symbol $\atrunc$ is well approximated by $b_{0}^{12}$. Let $$\NN := K - \{K, \rot_{2\pi/3} K, \rot_{4\pi/3} K\}$$ denote the set of three nearest neighbors to $K$ in $\lattice^*$.
One can then verify that
\begin{align*}
    (b_{0}^{12})^{\sigma \sigma'} (X) = \frac{1}{\delta|\Gamma|} \sum_{q \in \NN} e^{i \beta q \cdot \rot_{\pi/2} X}\hathperpnot^{\sigma \sigma'} (K-q),
\end{align*}
where $\hathperpnot^{\sigma \sigma'}$ is the Fourier transform of the function $x \mapsto \hperp (x + \os^{\sigma \sigma'}_0)$, with
\begin{align}\label{eq:os_0}
    \os^{\sigma \sigma'}_0 := \os^{\sigma} - \os^{\sigma'}
\end{align}
and the $\os^{\sigma}$ defined in \eqref{eq:os}.
It follows that
\begin{align}\label{eq:b11}
    (\atrunc)^{\sigma \sigma'} (X,K + \delta \zeta ; \delta) - \delta (b_{0}^{12})^{\sigma \sigma'}(X) = \res_{0,\delta}^{\sigma \sigma'}(X) + \res_{1,\delta}^{\sigma \sigma'} (X, \zeta) + \res_{2,\delta}^{\sigma \sigma'} (X, \zeta),
\end{align}
where 
\begin{align}\label{eq:res_delta}
\begin{split}
\res_{0,\delta}^{\sigma \sigma'} (X) &:= \frac{1}{|\Gamma|}\sum_{q \in \NN} e^{i \beta q \cdot \rot_{\pi/2} X}(\hathperp^{\sigma \sigma'} (K-q) -\hathperpnot^{\sigma \sigma'} (K-q)),\\
    \res_{1,\delta}^{\sigma \sigma'} (X, \zeta) &:= \frac{1}{|\Gamma|} \sum_{q \in \NN} e^{i \beta q \cdot \rot_{\pi/2} X}(\hathperp^{\sigma \sigma'} (\delta \zeta + K -(1 - \frac{1}{4}\beta^2 \delta^2)^{1/2}q) -\hathperp^{\sigma \sigma'} (K-q)),\\
    \res_{2,\delta}^{\sigma \sigma'} (X, \zeta) &:= \frac{1}{|\Gamma|} \sum_{q \in \latticeIn \setminus \NN} e^{i \beta q \cdot \rot_{\pi/2} X}\hathperp^{\sigma \sigma'} (\delta \zeta + K -(1 - \frac{1}{4}\beta^2 \delta^2)^{1/2}q).
\end{split}
\end{align}
Observe that for all multi-indices $\alpha, \beta \in \mathbb{N}_0^2$ and any $\eta < 1$ and $\eps', C_0 > 0$, we have
\begin{align*}
    |\partial^\alpha_X \partial^\beta_\zeta \res_{2,\delta}^{\sigma \sigma'} (X, \zeta)| &\le C_{\alpha,\beta} \sum_{q \in \latticeIn \setminus \NN}|q|^{|\alpha|} \delta^{|\beta|} \left| \partial^{\beta} \hathperp^{\sigma \sigma'} (\delta \zeta + K -(1 - \frac{1}{4}\beta^2 \delta^2)^{1/2}q) \right|\\
    &\le C_{\alpha,\beta} (\delta/\delta_0)^{\tbgparam |\beta|} \sum_{q \in \latticeIn \setminus \NN} |q|^{|\alpha|} (\delta/\delta_0)^{\aver{(\delta \zeta + K -(1 - \frac{1}{4}\beta^2 \delta^2)^{1/2}q)/\gamma}/\aver{K/\gamma}}
\end{align*}
uniformly in $X \in \mathbb{R}^d$ and $|\zeta| \le C_0 \delta^{-\eta}$, with the second inequality following from our assumption \eqref{eq:assumption_hperp} on the decay of $\hathperp$.
We next write $|q|^{|\alpha|} = |1 - \frac{1}{4}\beta^2 \delta^2|^{-|\alpha|/2} |(1 - \frac{1}{4}\beta^2 \delta^2)^{1/2} q|^{|\alpha|}$ and use the uniform boundedness of $\delta \zeta$ 
to conclude that for any $\eps' > 0$, 
\begin{align*}
    \sum_{q \in \latticeIn \setminus \NN} |q|^{|\alpha|} (\delta/\delta_0)^{\aver{(\delta \zeta + K -(1 - \frac{1}{4}\beta^2 \delta^2)^{1/2}q)/\gamma}/\aver{K/\gamma}} \le C \delta^{\aver{2K/\gamma}/\aver{K/\gamma} - \eps'}, \qquad X \in \mathbb{R}^d, \quad |\zeta| \le C_0 \delta^{-\eta}
\end{align*}
uniformly in $0 < \delta \le \delta_0/2$. Indeed, the polynomial growth in $(1 - \frac{1}{4}\beta^2 \delta^2)^{1/2} q$ of $|q|^{|\alpha|}$ is more than compensated by the exponential decay of $(\delta/\delta_0)^{\aver{(\delta \zeta + K -(1 - \frac{1}{4}\beta^2 \delta^2)^{1/2}q)/\gamma}/\aver{K/\gamma}}$, and
\begin{align}\label{eq:odl}
    K - (1 - \frac{1}{4}\beta^2 \delta^2)^{1/2}q = K-q + \odl q, \qquad \odl := 1 - (1 - \frac{1}{4}\beta^2 \delta^2)^{1/2}= O(\delta),
\end{align}
with $|K - q| \ge 2 |K|$ and $|\odl q| \le C\delta^{1-\eps} = o_{\delta \to 0}(1)$ for all $q \in \latticeIn \setminus \NN$. Thus we have shown that
\begin{align*}
    |\partial^\alpha_X \partial^\beta_\zeta \res_{2,\delta}^{\sigma \sigma'} (X, \zeta)| \le C_{\alpha, \beta} (\delta/\delta_0)^{\tbgparam |\beta|}\delta^{\aver{2K/\gamma}/\aver{K/\gamma} - \eps'}, \qquad X \in \mathbb{R}^d, \quad |\zeta| \le C_0 \delta^{-\eta},\quad 0 < \delta \le \delta_0/2.
\end{align*}

On the other hand, we use that $|q - K| = |K|$ for all $q \in \NN$ and again apply 
\eqref{eq:assumption_hperp} to conclude that
\begin{align*}
    \hathperp (\xi - (1 - \frac{1}{4}\beta^2 \delta^2)^{1/2}q) - \hathperp (K-q) = (\xi - K + \odl q) \cdot \hrem (\xi;q), \qquad q \in \NN, \quad |\xi - K| \le C_0 \delta^{1-\eta},
\end{align*}
where 
$\odl$ is defined in \eqref{eq:odl}
and 
$\hathperp^{{\rm rem}} (\cdot \; ;q) \in C^\infty (\mathbb{R}^2)$ satisfies $$|\partial^\alpha_\xi \hathperp^{{\rm rem}} (\xi;q)| \le C_\alpha \sup \{\partial^\beta \hathperp (\xi - (1 - \odl )q) :|\beta| \le |\alpha|+1,\; |\xi - K| \le C_0 \delta^{1-\eta}\}\le C_\alpha \delta^{1-(1-\tbgparam)(|\alpha|+1)-\eps'}$$
uniformly in $q \in \NN$ and $|\xi - K| \le C_0 \delta^{1-\eta}$. 
It follows that
\begin{align*}
    |\partial^\alpha_X \partial^\beta_\zeta \res_{1,\delta}^{\sigma \sigma'} (X, \zeta)| &\le C_{\alpha,\beta} \delta^{2+|\beta|-(1-\tbgparam)(|\beta|+1)-\eps'} \aver{\zeta} = C_{\alpha,\beta} \delta^{1+\tbgparam (1+|\beta|)-\eps'} \aver{\zeta}, \qquad X \in \mathbb{R}^2, \quad |\zeta| \le C_0 \delta^{-\eta}.
\end{align*}

To control $\res_{0,\delta}$, observe that
$\hathperp^{\sigma \sigma'} (k) = e^{ik \cdot \os^{\sigma \sigma'}} \hathperp (k)$ and $\hathperpnot^{\sigma \sigma'} (k) = e^{ik \cdot \os^{\sigma \sigma'}_0} \hathperp (k)$, where the definitions \eqref{eq:os} and \eqref{eq:os_0} together with the assumption \eqref{assumption:theta} imply that $|\os^{\sigma \sigma'} - \os_0^{\sigma \sigma'}| \le C \delta$ uniformly in $0 < \delta \le \delta_0/2$.
The definition \eqref{eq:res_delta} of $\res_{0,\delta}$ 
gives
\begin{align*}
    \res_{0,\delta} (X) = \frac{1}{|\Gamma|}\sum_{q \in \NN} e^{i \beta q \cdot \rot_{\pi/2} X}\hathperp (K-q) (e^{i(K-q) \cdot \os^{\sigma \sigma'}} - e^{i(K-q) \cdot \os_0^{\sigma \sigma'}}),
\end{align*}
while our assumption that $\hat{h} (\rot_{-2\pi j/3}K;\delta) = O(\delta)$ for each $j=0,1,2$ implies that each $\hathperp (K-q)$ on the above right-hand side is of $O(\delta)$. Therefore, $|\partial^\alpha_X \res_{0,\delta}^{\sigma \sigma'} (X)| \le C_\alpha \delta^2$ for all $X\in \mathbb{R}^2$ and $0 < \delta \le \delta_0/2$.

Combining our above estimates on $\res_{0,\delta}, \res_{1,\delta}$ and $\res_{2,\delta}$, we conclude by \eqref{eq:b11} that
\begin{align*}
    \atrunc (X,K + \delta \zeta ; \delta) - \delta b_{0}^{12}(X) = \delta^{1+\rate} b_{1} (X, \zeta; \delta),\qquad X \in \mathbb{R}^2, \quad |\zeta| \le C_0 \delta^{-\eta}
\end{align*}
for some $b_{1} \in S(\aver{\zeta})$ and any $\rate < \min \{\aver{2K/\gamma}/\aver{K/\gamma}-1, \tbgparam\}$.

Note that by Lemma~\ref{lemma:a12}, $\atruncall$ satisfies \eqref{eq:reg_a} in Assumption \ref{assumption:a} with $\nu_0 = 2\tpn$, $\nu_1 = \tpn$ and $\nu_2 = 1 - \tbgparam$.
We have thus shown that each $a_\delta^{jj}$ and $b_{0}^{jj}$ satisfy Assumption \ref{assumption:a}, and that for any $\tpn < \eta< 1$ and $\rate$ satisfying \eqref{eq:s_eta}, the interlayer coupling symbol also admits a decomposition
$$\atrunc (X,\xi;\delta) = \delta b_{0}^{12}(X) + \delta^{1+\rate} b_{1}(X, \frac{\xi-K}\delta;\delta), \qquad x \in \mathbb{R}^d, \quad |\xi - K| \le C_0 \delta^{1-\eta},$$
where $b_{0}^{12} \in S (1) \subset S (\aver{\zeta})$ and $b_{1} \in S(\aver{\zeta})$ uniformly in $\delta$. The ellipticity of the $b_{0}^{jj}$ and uniform boundedness of $b_{0}^{12}$ imply that $b_{0}$ is also elliptic. 
Therefore, the full symbols $\atruncall_\delta, b_{0}$ satisfy Assumption \ref{assumption:a}, implying that
\begin{align*}
    \|(D_t + \Htrunc_\delta) \psi_\delta (t, \cdot)\|_{H^N} \le C \delta^{1+\rate},\qquad t \ge 0, \quad 0 < \delta \le \delta_0/2
\end{align*}
by Lemma \ref{lemma:res}. Combining this with \eqref{eq:check}, we conclude that
\begin{align*}
    \|(D_t + H_\delta) \psi_\delta (t, \cdot)\|_{H^N} \le C \delta^{1+\rate},\qquad t \ge 0, \quad 0 < \delta \le \delta_0/2.
\end{align*}
Now, Lemma \ref{lemma:Hperp_bdd} implies that for any $0 < \delta \le \delta_0/2$ and $N \ge 0$, the operators
\begin{align*}
    H_\delta : H^N (\mathbb{R}^2; \mathbb{C}^4) \to H^N (\mathbb{R}^2; \mathbb{C}^4), \qquad [\diag (\nabla_x I_2, \nabla^\theta_x I_2), H_\delta]: H^N (\mathbb{R}^2; \mathbb{C}^4) \to H^N (\mathbb{R}^2; \mathbb{C}^4 \oplus \mathbb{C}^4)
\end{align*}
are bounded, with
\begin{align*}
    \norm{H_\delta}_{H^N \to H^N} \le C, \qquad \norm{[\diag (\nabla_x I_2, \nabla^\theta_x I_2), H_\delta]}_{H^N \to H^N} \le C \delta, \qquad 0 < \delta \le \delta_0/2.
\end{align*}
Since $u_\delta := \psi_\delta - \varphi_\delta$ satisfies $u_\delta (0, \cdot) = 0$ and
    $
        (D_t + H_\delta) u_\delta(t,x) = (D_t + H_\delta)\psi_\delta (t,x),
    $
we conclude by
    Lemma \ref{lemma:u_delta_reg} and Remark \ref{remark:twisted_grad} 
    that 
    \begin{align*}
        \norm{u_\delta (t, \cdot)}_{H^N (\mathbb{R}^2; \mathbb{C}^4)} \le C \delta^{1+\rate} t (1 + (\delta t)^N), \qquad t \ge 0, \quad 0 < \delta \le \delta_0/2,
    \end{align*}
    as desired.
\end{proof}

\subsection{Proof of Lemmas \ref{lemma:Hperp_bdd} and \ref{lemma:a12}}\label{subsec:lemma_a12}
\begin{proof}[Proof of Lemma \ref{lemma:Hperp_bdd}]
We first prove that $\Hperp_\delta$ is bounded on $L^2 (\mathbb{R}^2; \mathbb{C}^2)$, with $\norm{\Hperp_\delta}_{L^2 \to L^2} \le C$ uniformly in $\delta$. By the regularity \eqref{eq:assumption_hperp} of $\hperp$, this argument easily extends to the spaces $H^N (\mathbb{R}^2; \mathbb{C}^2)$.

For simplicity of the proof presentation, we suppress the orbital notation and treat $h_\delta$ as a matrix-valued function and $\psi$ below as a vector-valued function. Then we have
        \begin{equation*}
            \begin{split}
                 H_\delta^{12}\psi(x) &= \sum_{r_1 \in \Lambda_1^0}h_\delta(x-R_\theta (r_1 + x)) \psi( R_\theta(r_1 + x)) 
                 = \sum_{r_2 \in \Lambda_2^0} h_\delta((I-R_\theta)x - r_2) \psi(R_\theta x + r_2).
            \end{split}
        \end{equation*}
        Let $\Gamma_2$ denote the unit cell of the lattice $\lattice_2^0$. 
        For any $x \in \mathbb{R}^2$ and $r_2 \in \lattice_2^0$, there exists a unique $y_{r_2} (x) \in \lattice^0_2$ such that $(I - \rot_\theta) x -y_{r_2} (x) \in \Gamma_2 + r_2$. Moreover, for all $x \in \mathbb{R}^2$, $\{y_{r_2} (x) : r_2 \in \lattice_2^0\} = \lattice_2^0$. Therefore,
        \begin{align}\label{eq:Hperp_y}
            \Hperp_\delta \psi (x) = \sum_{r_2 \in \lattice_2^0}\hperp ((I -\rot_\theta) x - y_{r_2} (x)) \psi (\rot_\theta x+ y_{r_2} (x)).
        \end{align}
        By the decay \eqref{eq:assumption_hperp} of $\hperp$ and definition of $y_{r_2}$, it is clear that
        \begin{align*}
            \sum_{r_2 \in \lattice^0_2} \sup_{x \in \mathbb{R}^2} |\hperp ((I -\rot_\theta) x - y_{r_2} (x))| \le C < \infty
        \end{align*}
        uniformly in 
        $0 < \delta \le \delta_0/2$. 
        This means the operators $\fh^{r_2}_{\delta}$ defined as point-wise multiplication by the function $\hperp ((I -\rot_\theta) x - y_{r_2} (x))$ are each bounded in $L^2$ and satisfy
        \begin{align}\label{eq:sum_hperp}
            \sum_{r_2 \in \lattice^0_2} \norm{\fh^{r_2}_{\delta}}_{L^2 \to L^2}\le C < \infty, \qquad 0 < \delta \le \delta_0/2.
        \end{align}
        Now, observe that for any $r_2, r_2' \in \lattice_2^0$,
        \begin{align}\label{eq:level_sets}
            y_{r_2}^{-1} (r_2') := \{x \in \mathbb{R}^2 : y_{r_2} (x) = r_2'\} = (I - \rot_\theta)^{-1} (\Gamma_2 + r_2 + r_2').
        \end{align}
        For any fixed $r_2 \in \lattice_2^0$, the sets $y_{r_2}^{-1} (r_2')$ tile the plane, meaning that $$\cup_{r_2' \in \lattice_2^0} y_{r_2}^{-1} (r_2') = \mathbb{R}^2, \qquad y_{r_2}^{-1} (r_2') \cap y_{r_2}^{-1} (r_2'') = \emptyset \hspace{0.2cm} \text{if}\hspace{0.2cm} r_2' \ne r_2''.$$
        Thus for any $r_2 \in \lattice_2^0$,
        \begin{align*}
            \int_{\mathbb{R}^2} |\psi (\rot_\theta x + y_{r_2} (x))|^2 dx = \sum_{r_2' \in \lattice_2^0} \int_{y^{-1}_{r_2}(r_2')}|\psi (\rot_\theta x + y_{r_2} (x))|^2 dx = \sum_{r_2' \in \lattice_2^0} \int_{y^{-1}_{r_2}(r_2')}|\psi (\rot_\theta x + r_2')|^2 dx.
        \end{align*}
        After the change of variables $z = \rot_\theta x+ r_2'$, this becomes
        \begin{align*}
            \int_{\mathbb{R}^2} |\psi (\rot_\theta x + y_{r_2} (x))|^2 dx = \sum_{r_2' \in \lattice_2^0}\int_{Y_{r_2} (r_2')}|\psi (z)|^2 dz, \qquad Y_{r_2} (r_2') := \rot_\theta y^{-1}_{r_2}(r_2') + r_2'.
        \end{align*}
        From \eqref{eq:level_sets} and our assumption \eqref{assumption:theta} on the twist angle, 
        it follows that provided $\delta > 0$ is sufficiently small, 
        the sets $Y_{r_2} (r_2')$ and $Y_{r_2} (r_2'')$ are disjoint for all $|r_2'- r_2''| > v$, with $v > 0$ the lattice spacing for $\lattice_2^0$ (recall section \ref{sec:Haldane}).
        Indeed, 
        since each entry of $(I-\rot_\theta)^{-1}$ is of $O (\delta^{-1})$, we know that non-adjacent sets $R_\theta y_{r_2}^{-1} (r_2')$ and $R_\theta y_{r_2}^{-1} (r_2'')$ are separated by $O(\delta^{-1} |r_2'-r_2''|)$, meaning that the corresponding $Y_{r_2} (r_2') = \rot_\theta y^{-1}_{r_2}(r_2') + r_2'$ and $Y_{r_2} (r_2'') = \rot_\theta y^{-1}_{r_2}(r_2'') + r_2''$ also have a positive separation of $O(\delta^{-1} |r_2'-r_2''|)$. The sets $R_\theta y_{r_2}^{-1} (r_2')$ and $R_\theta y_{r_2}^{-1} (r_2'')$ are non-adjacent if $|r_2' - r_2''| > v$.
        
        We also know that the triangular lattice can be decomposed as
        \(
            \lattice_2^0 = \lattice^1 \cup \lattice^2 \cup \lattice^3,
        \)
        where $\lattice^1$ is a larger triangular lattice with spacing $\sqrt{3} v$, and $\lattice^2$ and $\lattice^3$ are shifts of $\lattice^1$ that satisfy
        \(
            \lattice^i \cap \lattice^j = \emptyset \hspace{0.2cm} \text{if}\hspace{0.2cm} i \ne j.
        \)
        We conclude that for all $r_2 \in \lattice_2^0$,
        \begin{align*}
            \int_{\mathbb{R}^2} |\psi (\rot_\theta x + y_{r_2} (x))|^2 dx \le \sum_{j=1}^3 \sum_{r_2' \in \lattice^j}\int_{Y_{r_2} (r_2')}|\psi (z)|^2 dz \le 3 \norm{\psi}_{L^2}^2,
        \end{align*}
        which combined with \eqref{eq:Hperp_y} and \eqref{eq:sum_hperp} yields the desired bound of
        \begin{align}\label{eq:TBG_bdd}
            \norm{\Hperp_\delta \psi}_{L^2}^2 \le C \norm{\psi}_{L^2}^2, \qquad 0 < \delta \le \delta_0/2.
        \end{align}

    The boundedness of $\nabla_x \Hperp_\delta - \Hperp_\delta \nabla_x^\theta$ 
    is established by a parallel argument. 
    For any 
    $\psi \in L^2 (\mathbb{R}^2; \mathbb{C}^2)$,
    \begin{align*}
        (\nabla_x \Hperp_\delta - \Hperp_\delta \nabla_x^\theta) [\psi]^\sigma (x) &= \sum_{r \in \lattice} \sum_{\sigma' \in \{A,B\}} (I - \rot_\theta)^\top \nabla \hperp^{\sigma \sigma'} ((I - \rot_\theta) x - \rot_{\theta/2} r) \psi^{\sigma'} (\rot_\theta x + \rot_{\theta/2} r),
    \end{align*}
    with
    $I - \rot_\theta = O(\delta)$ by \eqref{assumption:theta}. It then follows from the proof of \eqref{eq:TBG_bdd} and regularity of $\hperp$ that
    \begin{align*}
        \norm{(\nabla_x \Hperp_\delta - \Hperp_\delta \nabla_x^\theta) [\psi]}_{L^2} \le C \delta \norm{\psi}_{L^2}, \qquad 0 < \delta \le \delta_0/2.
    \end{align*}
    As before, this bound naturally extends to the spaces $H^N (\mathbb{R}^2; \mathbb{C}^2)$.
\end{proof}
\begin{proof}[Proof of Lemma \ref{lemma:a12}]
    Observe that for any $\alpha, \beta \in \mathbb{N}_0^2$,
\begin{align*}
    &|\partial^\alpha_X \partial^\beta_\xi \atrunc (X, \xi; \delta)| \le C \sum_{q \in \latticeIn} |q|^{|\alpha|} |\partial^\beta_\xi \hathperp^{\sigma \sigma'} (\xi -(1 - \frac{1}{4}\beta^2 \delta^2)^{1/2}q; \delta)| \\
    \le & \ C \sum_{q \in \latticeIn} |q|^{|\alpha|} (\delta/\delta_0)^{\aver{(\xi -(1 - \frac{1}{4}\beta^2 \delta^2)^{1/2}q)/\gamma}/\aver{K/\gamma} - (1-\tbgparam)|\beta|}
    \le C\sum_{q \in \latticeIn} \delta^{-\tpn |\alpha|} \delta^{-(1-\tbgparam)|\beta|} \le C \delta^{-2\tpn - \tpn |\alpha| - (1-\tbgparam)|\beta|},
\end{align*}
where the third and fourth inequalities respectively follow from the decay \eqref{eq:assumption_hperp} of $\hathperp$ and definition of $\latticeIn$ to bound $|q|^{|\alpha|}$, while the last inequality follows from the fact that $\latticeIn$ has $O (\delta^{-2 \tpn})$ elements.

It remains to verify \eqref{eq:ares_bd} which is unfortunately technical.
Fix $q,N \ge 0$ and let $f \in \cS (\mathbb{R}^2; \mathbb{C}^2)$. Then
\begin{align*}
    \ow \adres f (x) &= \int_{\mathbb{R}^{4}} e^{i \xi \cdot (x-y)} \ares (\frac{\delta x+\delta y}{2}, \xi; \delta) f(y) \frac{dy d\xi}{(2\pi)^2}
    =
     \int_{\mathbb{R}^{4}} e^{i \xi \cdot z} \ares (\delta x - \frac{\delta z}{2}, \xi; \delta) f(x-z) \frac{dy d\xi}{(2\pi)^2},
\end{align*}
which implies that
\begin{align*}
    &\ow \adres [f]^\sigma (x) 
    = \int_{\mathbb{R}^{4}} \sum_{q \in \latticeOut} \sum_{\sigma' \in \{A,B\}} e^{i \xi \cdot z} \!\!\! e^{i \beta q \cdot \rot_{\pi/2}(\delta x - \delta z/2)} \hathperp^{\sigma \sigma'} (\xi - (1 - \frac{1}{4}\beta^2 \delta^2)^{\frac12} q) f^{\sigma'}(x-z)\frac{dz d\xi}{(2\pi)^2 |\Gamma|}.
\end{align*}
Swapping the order of the sums and integrals and changing variables $\xi \leftarrow \xi - \delta \beta \rot_{\pi/2}^\top q/2$, we obtain
\begin{align*}
    \ow \adres [f]^\sigma (x) = \sum_{q \in \latticeOut} \sum_{\sigma' \in \{A,B\}} \int_{\mathbb{R}^{4}} e^{i \xi \cdot z} e^{i \delta \beta q \cdot \rot_{\pi/2}x} \hathperp^{\sigma \sigma'} (\xi - (I-\delta \rotq_\delta)q) f^{\sigma'}(x-z) \frac{dz d\xi}{(2\pi)^2 |\Gamma|},
\end{align*}
where $\rotq_\delta := \frac{1}{2}\beta \rot_{\pi/2}^\top + \delta^{-1}(1 - (1 - \frac{1}{4}\beta^2 \delta^2)^{1/2}) I = O(1)$ as $\delta \to 0$.
Let $\chi_0 \in C^\infty_c (\mathbb{R}^2)$ such that $\chi_0 \equiv 1$ in the unit ball, and define $\chi_1 := 1 - \chi_0$. Assume that $0 \le \chi_0 \le 1$. For $\sigma, \sigma' \in \{A,B\}$ and $j,k \in \{0,1\}$, define
\begin{align}\label{eq:f_decomp}
    \ow (\adres) [f]^\sigma (x) & = \sum_{\sigma' \in \{A,B\}} \sum_{j,k =0}^1 g_{j,k}^{\sigma,\sigma'} (x) 
    \\
    g^{\sigma, \sigma'}_{j,k} (x) &:= 
    \sum_{q \in \latticeOut} \int_{\mathbb{R}^{4}}\chi_j(z) \chi_k (\xi) e^{i \xi \cdot z} e^{i \delta \beta q \cdot \rot_{\pi/2}x} \hathperp^{\sigma \sigma'} (\xi - (I-\delta \rotq_\delta)q) f^{\sigma'}(x-z) \frac{dz d\xi}{(2\pi)^2 |\Gamma|}.\nonumber
\end{align}
Dropping the subscripts $\sigma, \sigma'$ for brevity, we find that for any $\alpha \in \mathbb{N}_0^2$ with $|\alpha| \le N$,
\begin{align*}
    |\partial^\alpha_x g_{0,0} (x) | \le C \sup_{|\alpha_1| + |\alpha_2| \le N} \int_{\Omega^2} \sum_{q \in \latticeOut} \delta^{|\alpha_1|} |q|^{|\alpha_1|} \left|\hathperp (\xi - (I-\delta \rotq_\delta)q) \right| \left| \partial^{\alpha_2} f(x-z) \right| dz d\xi,
\end{align*}
where $\Omega \subset \mathbb{R}^2$ is a bounded set containing $\supp (\chi_0)$.
By the decay \eqref{eq:assumption_hperp} of $\hathperp$, this becomes
\begin{align*}
    |\partial^\alpha_x g_{0,0} (x) | \le C \sup_{|\alpha_1| + |\alpha_2| \le N} \int_{\Omega^2} \sum_{q \in \latticeOut} \delta^{|\alpha_1|} |q|^{|\alpha_1|} (\delta/\delta_0)^{\aver{(\xi - (I-\delta \rotq_\delta)q)/\gamma}/\aver{K/\gamma}} \left| \partial^{\alpha_2} f(x-z) \right| dz d\xi.
\end{align*}
Since $\xi$ is restricted to a bounded set, it follows that for some $\eta_1, \eta_2 > 0$,
\begin{align*}
    |\partial^\alpha_x g_{0,0} (x) | \le C \sup_{|\alpha_1| + |\alpha_2| \le N} \int_{\Omega^2} \sum_{q \in \latticeOut} \delta^{|\alpha_1|} |q|^{|\alpha_1|} (\delta/\delta_0)^{\eta_1 |q| - \eta_2} \left| \partial^{\alpha_2} f(x-z) \right| dz d\xi.
\end{align*}
Using that $\sum_{q \in \latticeOut}|q|^{|\alpha_1|} (\delta/\delta_0)^{\eta_1 |q|} \le C (\delta/\delta_0)^{\tilde\eta_1 \delta^{-\tpn}}$ for any $0 < \tilde\eta_1 < \eta_1$, it follows that
\begin{align*}
    |\partial^\alpha_x g_{0,0} (x) | \le C \sup_{|\alpha_1| + |\alpha_2| \le N} \int_{\Omega^2} \delta^{|\alpha_1|} (\delta/\delta_0)^{\eta_1 \delta^{-\tpn} - \eta_2} \left| \partial^{\alpha_2} f(x-z) \right| dz d\xi,
\end{align*}
where we have dropped the tilde from $\tilde \eta_1$. Since $\Omega^2$ is bounded, we conclude that
\begin{align}\label{eq:bd_g00}
    \norm{g_{0,0}}_{H^N} \le C (\delta/\delta_0)^{\eta_1 \delta^{-\tpn} - \eta_2} \norm{f}_{H^N}
\end{align}
decays super-algebraically in $\delta$.

We next consider $g_{1,0}$, which by integration by parts satisfies
\begin{align*}
    &g_{1,0}^{\sigma, \sigma'} (x) = \int_{\mathbb{R}^{4}} \left( \left( \frac{z \cdot \nabla_\xi}{i|z|^2}\right)^m e^{i \xi\cdot z} \right) \chi_1 (z) \chi_0 (\xi)  \sum_{q \in \latticeOut} e^{i \delta \beta q \cdot \rot_{\pi/2}x} \hathperp^{\sigma \sigma'} (\xi - (I-\delta \rotq_\delta)q)  f^{\sigma'}(x-z) \frac{dz d\xi}{(2\pi)^2 |\Gamma|}\\
    = & \int_{\mathbb{R}^{4}} \chi_1 (z)e^{i \xi\cdot z} \sum_{q \in \latticeOut} e^{i \delta \beta q \cdot \rot_{\pi/2}x} f^{\sigma'}(x-z) 
    \left( \frac{i z \cdot \nabla_\xi}{|z|^2}\right)^m \left( \chi_0 (\xi) \hathperp^{\sigma \sigma'} (\xi - (I-\delta \rotq_\delta)q) \right) \frac{dz d\xi}{(2\pi)^2 |\Gamma|}
\end{align*}
for any $m \in \mathbb{N}_0$.
It follows that for any $\alpha \in \mathbb{N}_0^2$ with $|\alpha| \le N$ (again dropping the superscripts $\sigma, \sigma'$),
\begin{align*}
    |\partial^\alpha_x g_{1,0} (x) | &\le C \sup_{\substack{|\alpha_1| + |\alpha_2| \le N, \\
    |\beta| \le m}} \int_{\mathbb{R}^2 \times \Omega} \chi_1 (z)\sum_{q \in \latticeOut} \delta^{|\alpha_1|} |q|^{|\alpha_1|} |z|^{-m} 
    \left|\partial^\beta \hathperp (\xi - (I-\delta \rotq_\delta)q) \right| \left| \partial^{\alpha_2} f(x-z) \right| dz d\xi,
\end{align*}
which by the decay \eqref{eq:assumption_hperp} of $\hathperp$ becomes
\begin{align*}
    |\partial^\alpha_x g_{1,0} (x) | &\le C \sup_{|\alpha_1| + |\alpha_2| \le N} \int_{\mathbb{R}^2 \times \Omega} \chi_1 (z)\sum_{q \in \latticeOut} \delta^{|\alpha_1|} |q|^{|\alpha_1|} |z|^{-m} \\
    &\hspace{4cm}(\delta/\delta_0)^{\aver{(\xi - (I-\delta \rotq_\delta)q)/\gamma}/\aver{K/\gamma}- (1-\tbgparam)m} \left| \partial^{\alpha_2} f(x-z) \right| dz d\xi.
\end{align*}
Choosing $m = 3$ so that the above function of $z$ is integrable, the same arguments that led to \eqref{eq:bd_g00} imply that (for possibly a different choice of positive constants $C, \eta_1, \eta_2$)
\begin{align}\label{eq:bd_g10}
    \norm{g_{1,0}}_{H^N} \le C (\delta/\delta_0)^{\eta_1 \delta^{-\tpn} - \eta_2} \norm{f}_{H^N}.
\end{align}

Next, we consider $g_{0,1}$. Again, integration by parts (this time in $z$) implies
\begin{align*}
    g_{0,1}^{\sigma,\sigma'} (x) &= \frac{1}{(2\pi)^2|\Gamma|} \int_{\mathbb{R}^{4}} \left( \left( \frac{\xi \cdot \nabla_z}{i|\xi|^2}\right)^m e^{i \xi\cdot z} \right) \chi_0 (z) \chi_1 (\xi)
    \sum_{q \in \latticeOut} e^{i \delta \beta q \cdot \rot_{\pi/2}x} \hathperp^{\sigma \sigma'} (\xi - (I-\delta \rotq_\delta)q) f^{\sigma'}(x-z) dz d\xi
    \\&
    = \frac{1}{(2\pi)^2 |\Gamma|} \int_{\mathbb{R}^{4}} e^{i \xi\cdot z} \chi_1 (\xi)\sum_{q \in \latticeOut} e^{i \delta \beta q \cdot \rot_{\pi/2}x} \hathperp^{\sigma \sigma'} (\xi - (I-\delta \rotq_\delta)q)
    \left( \frac{i\xi \cdot \nabla_z}{|\xi|^2}\right)^m \left( \chi_0 (z) f^{\sigma'}(x-z)\right) dz d\xi
\end{align*}
for any $m \in \mathbb{N}_0$. Thus, for any $\alpha \in \mathbb{N}_0^2$ with $|\alpha| \le N$, we have
\begin{align*}
    |\partial^\alpha_x g_{0,1} (x)| &\le C \sup_{\substack{|\alpha_1|+|\alpha_2| \le N\\
    |\beta| \le m}} \int_{\Omega \times \mathbb{R}^2} \chi_1 (\xi) |\xi|^{-m} \delta^{|\alpha_1|} 
    \sum_{q \in \latticeOut} |q|^{|\alpha_1|} |\hathperp (\xi - (I-\delta \rotq_\delta)q)| |\partial^{\alpha_2 + \beta} f (x-z)| dz d\xi.
\end{align*}
It follows from \eqref{eq:assumption_hperp} that
\begin{align}\label{eq:g_bd}
\begin{split}
    |\partial^\alpha_x g_{0,1} (x)| &\le C \sup_{\substack{|\alpha_1|+|\alpha_2| \le N\\
    |\beta| \le m}} \int_{\Omega \times \mathbb{R}^2} \aver{\xi}^{-m} \delta^{|\alpha_1|} 
    \sum_{q \in \latticeOut} |q|^{|\alpha_1|} (\delta/\delta_0)^{\aver{(\xi - (I-\delta \rotq_\delta)q)/\gamma}/\aver{K/\gamma}} |\partial^{\alpha_2 + \beta} f (x-z)| dz d\xi.
\end{split}
\end{align}
We next 
separate the contributions to the above integral of $\xi \in \Omega_\delta$ and $\xi \in \mathbb{R}^2 \setminus \Omega_\delta =: \Omega_\delta^c$, where 
\begin{align}\label{eq:Omega_delta}
    \Omega_\delta := \{\xi \in \mathbb{R}^2: |\xi| < \frac{1}{2} \delta^{-\tpn}\}.
\end{align}
The first contribution is then
\begin{align*}
    \iin (x) &:= \sup_{\substack{|\alpha_1|+|\alpha_2| \le N\\
    |\beta| \le m}} \int_{\Omega \times \Omega_\delta} \aver{\xi}^{-m} \delta^{|\alpha_1|} 
    \sum_{q \in \latticeOut} |q|^{|\alpha_1|} (\delta/\delta_0)^{\aver{(\xi - (I-\delta \rotq_\delta)q)/\gamma}/\aver{K/\gamma}} |\partial^{\alpha_2 + \beta} f (x-z)| dz d\xi\\
    &\le \sup_{\substack{|\alpha_1|+|\alpha_2| \le N\\
    |\beta| \le m}} \int_{\Omega \times \Omega_\delta} \aver{\xi}^{-m} \delta^{|\alpha_1|} 
    \sum_{q \in \latticeOut} |q|^{|\alpha_1|} (\delta/\delta_0)^{\aver{q/4\gamma}/\aver{K/\gamma}} |\partial^{\alpha_2 + \beta} f (x-z)| dz d\xi,
\end{align*}
where we used the fact that $|\xi - (I-\delta \rotq_\delta)q| \ge |q| - \delta |\rotq_\delta q| - \frac{1}{2} \delta^{-\tpn} \ge \frac{1}{2} |q| - \delta |\rotq_\delta q|\ge\frac{1}{4} |q|$ for all $\xi \in \Omega_\delta$ and $q \in \latticeOut$. As before, there exist some $C, \eta_1 > 0$ such that $\sum_{q \in \latticeOut}|q|^{|\alpha_1|} (\delta/\delta_0)^{\aver{q/4\gamma}/\aver{K/\gamma}} \le C (\delta/\delta_0)^{\eta_1 \delta^{-\tpn}}$,
and thus
\begin{align*}
    \iin (x) &\le C\sup_{|\alpha'| \le N+m} \int_{\Omega \times \Omega_\delta} \aver{\xi}^{-m} (\delta/\delta_0)^{\eta_1 \delta^{-\tpn}} |\partial^{\alpha'} f (x-z)| dz d\xi.
\end{align*}
Choosing $m=3$ so that $\aver{\xi}^{-m} \in L^1 (\mathbb{R}^2)$, we conclude that $\norm{\iin}_{L^2} \le C(\delta/\delta_0)^{\eta_1 \delta^{-\tpn}} \norm{f}_{H^{N+3}}$.
The remaining contribution to the right-hand side of \eqref{eq:g_bd} is
\begin{align*}
    \iout (x) &:= \sup_{\substack{|\alpha_1|+|\alpha_2| \le N\\
    |\beta| \le m}} \int_{\Omega \times \Omega_\delta^c} \aver{\xi}^{-m} \delta^{|\alpha_1|} 
    \sum_{q \in \latticeOut} |q|^{|\alpha_1|} (\delta/\delta_0)^{\aver{(\xi - (I-\delta \rotq_\delta)q)/\gamma}/\aver{K/\gamma}} |\partial^{\alpha_2 + \beta} f (x-z)| dz d\xi.
\end{align*}
Observe that for any $|\alpha_1| \le N$,
\begin{align}\label{eq:sum_q_unif_bd}
    \sum_{q \in \latticeOut} |q|^{|\alpha_1|} (\delta/\delta_0)^{\aver{(\xi - (I-\delta \rotq_\delta)q)/\gamma}/\aver{K/\gamma}} \le C
\end{align}
uniformly in $\xi \in \mathbb{R}^2$ and $0 < \delta \le \delta_0/2$, hence
\begin{align*}
    \iout (x) \le C \sup_{|\alpha'| \le N+m} \int_{\Omega \times \Omega_\delta^c} \aver{\xi}^{-m} |\partial^{\alpha'} f (x-z)| dz d\xi.
\end{align*}
It follows from the boundedness of $\Omega$ that
\begin{align*}
    \norm{\iout}_{L^2} \le C \norm{f}_{H^{N+m}} \int_{\Omega_\delta} \aver{\xi}^{-m} d\xi \le C\norm{f}_{H^{N+m}} \delta^{\tpn (m-2)}
\end{align*}
for any $m \ge 3$. Combining our estimates on $\iin$ and $\iout$, and using the fact that $|\partial^\alpha_x g_{0,1} (x)| \le C (\iin (x) + \iout (x))$, we obtain that for any $m \ge 3$,
\begin{align}\label{eq:bd_g_01}
    \norm{g_{0,1}}_{H^N} \le C\norm{f}_{H^{N+m}} \delta^{\tpn (m-2)}.
\end{align}

It remains to consider $g_{1,1}$. Integrating by parts as before, we have using ${\mathfrak V}=\frac{\xi \cdot \nabla_z}{i|\xi|^2}$ that 
\begin{align*}
    g_{1,1}^{\sigma, \sigma'} (x) &=  \int_{\mathbb{R}^{4}} \left( {\mathfrak V}^{m_1} e^{i \xi\cdot z} \right) \chi_1 (z) \chi_1 (\xi)
    \sum_{q \in \latticeOut} e^{i \delta \beta q \cdot \rot_{\pi/2}x} \hathperp^{\sigma \sigma'} (\xi - (I-\delta \rotq_\delta)q) f^{\sigma'}(x-z) \frac{dz d\xi}{(2\pi)^2|\Gamma|}
    \\
    &=  \int_{\mathbb{R}^{4}} e^{i \xi\cdot z} \chi_1 (\xi)\sum_{q \in \latticeOut}e^{i \delta \beta q \cdot \rot_{\pi/2}x} \hathperp^{\sigma \sigma'} (\xi - (I-\delta \rotq_\delta)q)
    {\mathfrak V}^{m_1} \left( \chi_1 (z) f^{\sigma'}(x-z)\right) \frac{dz d\xi}{(2\pi)^2|\Gamma|}
    \\
    &=  \int_{\mathbb{R}^{4}} \left( {\mathfrak V}^{m_2} e^{i \xi\cdot z}\right) \chi_1 (\xi)\sum_{q \in \latticeOut} e^{i \delta \beta q \cdot \rot_{\pi/2}x} \hathperp^{\sigma \sigma'} (\xi - (I-\delta \rotq_\delta)q)
    {\mathfrak V}^{m_1} \left( \chi_1 (z) f^{\sigma'}(x-z)\right) \frac{dz d\xi}{(2\pi)^2|\Gamma|}
    \\
    &= \int_{\mathbb{R}^{4}}  e^{i \xi\cdot z}{\mathfrak V}^{m_2} \Bigg ( \chi_1 (\xi)\sum_{q \in \latticeOut} e^{i \delta \beta q \cdot \rot_{\pi/2}x} \hathperp^{\sigma \sigma'} (\xi - (I-\delta \rotq_\delta)q)
    {\mathfrak V}^{m_1} \left( \chi_1 (z) f^{\sigma'}(x-z)\right) \Bigg) 
    \frac{dz d\xi}{(2\pi)^2|\Gamma|}
\end{align*}
for any $m_1, m_2 \in \mathbb{N}_0$. It follows that for any $\alpha \in \mathbb{N}_0^2$ with $|\alpha| \le N$, using the notation ${\mathfrak w}=\aver{z}^{-m_2} \aver{\xi}^{-m_1} \delta^{|\alpha_1|}$,
\begin{align*}
    |\partial^{\alpha} g_{1,1}^{\sigma, \sigma'} (x)| &\le C \sup_{\substack{|\alpha_1| + |\alpha_2| \le N\\
    |\beta_j| \le m_j}}
    \int_{\mathbb{R}^4} 
    {\mathfrak w}
    \sum_{q \in \latticeOut}|q|^{|\alpha_1|} |\partial^{\beta_2}_\xi \hathperp^{\sigma \sigma'} (\xi - (I-\delta \rotq_\delta)q)| |\partial^{\alpha_2+\beta_1} f^{\sigma'} (x-z)| dz d\xi\\
     &\le C \sup_{\substack{|\alpha_1| + |\alpha_2| \le N\\
    |\beta_j| \le m_j}}
    \int_{\mathbb{R}^4} 
    {\mathfrak w}
    \sum_{q \in \latticeOut} |q|^{|\alpha_1|} (\delta/\delta_0)^{\aver{(\xi - (I-\delta \rotq_\delta)q)/\gamma}/\aver{K/\gamma} - (1-\tbgparam)|\beta_2|} |\partial^{\alpha_2+\beta_1} f^{\sigma'} (x-z)| dz d\xi,
\end{align*}
where we used \eqref{eq:assumption_hperp} to obtain the second inequality. As before, we separate the contributions from $\xi \in \Omega_\delta$ and $\xi \in \Omega_\delta^c$, with $\Omega_\delta$ defined by \eqref{eq:Omega_delta}. Following the arguments below \eqref{eq:Omega_delta}, we obtain
\begin{align*}
    \iin^\sharp (x) &:= \sup_{\substack{|\alpha_1| + |\alpha_2| \le N\\
    |\beta_j| \le m_j}}
    \int_{\mathbb{R}^2 \times \Omega_\delta} {\mathfrak w}
    \sum_{q \in \latticeOut}|q|^{|\alpha_1|} (\delta/\delta_0)^{\aver{(\xi - (I-\delta \rotq_\delta)q)/\gamma}/\aver{K/\gamma} - (1-\tbgparam)|\beta_2|}
    \  |\partial^{\alpha_2+\beta_1} f^{\sigma'} (x-z)| dz d\xi\\
    &\le \sup_{\substack{|\alpha_1| + |\alpha_2| \le N\\
    |\beta_j| \le m_j}}
    \int_{\mathbb{R}^2 \times \Omega_\delta} {\mathfrak w}
    \sum_{q \in \latticeOut}|q|^{|\alpha_1|} (\delta/\delta_0)^{\aver{q/4\gamma}/\aver{K/\gamma} - (1-\tbgparam)|\beta_2|} \ |\partial^{\alpha_2+\beta_1} f^{\sigma'} (x-z)| dz d\xi\\
    &\le C\sup_{\substack{|\alpha_1| + |\alpha_2| \le N\\
    |\beta_j| \le m_j}}
    \int_{\mathbb{R}^2 \times \Omega_\delta} {\mathfrak w}
    (\delta/\delta_0)^{\eta_1 \delta^{-\tpn}}
    (\delta/\delta_0)^{- (1-\tbgparam)|\beta_2|}|\partial^{\alpha_2+\beta_1} f^{\sigma'} (x-z)| dz d\xi.
\end{align*}
Setting $m_1 = m_2 = 3$, it follows that
\begin{align*}
    \iin^\sharp (x) &\le C\sup_{|\alpha'| \le N+3} \int_{\mathbb{R}^2 \times \Omega_\delta}\aver{z}^{-3} \aver{\xi}^{-3}(\delta/\delta_0)^{\eta_1 \delta^{-\tpn}}(\delta/\delta_0)^{- 3(1-\tbgparam)}|\partial^{\alpha'} f^{\sigma'} (x-z)| dz d\xi, \\
    \norm{\iin^\sharp}_{L^2} &\le C (\delta/\delta_0)^{\eta_1 \delta^{-\tpn} - 3(1-\tbgparam)} \norm{f}_{H^{N+3}}.
\end{align*}
The contribution from $\xi \in \Omega_\delta^c$ satisfies
\begin{align*}
    \iout^\sharp (x) &:= \sup_{\substack{|\alpha_1| + |\alpha_2| \le N\\
    |\beta_j| \le m_j}}
    \int_{\mathbb{R}^2 \times \Omega_\delta^c} {\mathfrak w}
    \sum_{q \in \latticeOut}|q|^{|\alpha_1|} (\frac \delta{\delta_0})^{\aver{(\xi - (I-\delta \rotq_\delta)q)/\gamma}/\aver{K/\gamma} - (1-\tbgparam)|\beta_2|}|\partial^{\alpha_2+\beta_1} f^{\sigma'} (x-z)| dz d\xi
    \\
    &\leq C \sup_{\substack{|\alpha_1| + |\alpha_2| \le N\\
    |\beta_j| \le m_j}}
    \int_{\mathbb{R}^2 \times \Omega_\delta^c} \aver{z}^{-m_2} \aver{\xi}^{-m_1} \delta^{|\alpha_1|} (\delta/\delta_0)^{- (1-\tbgparam)|\beta_2|} |\partial^{\alpha_2+\beta_1} f^{\sigma'} (x-z)| dz d\xi
\end{align*}
Using \eqref{eq:sum_q_unif_bd}.
Setting $m_2 =3$, this gives
\begin{align*}
    \iout^\sharp (x) &\le
    C \sup_{|\alpha'| \le N+m_1}
    \int_{\mathbb{R}^2 \times \Omega_\delta^c} \aver{z}^{-3} \aver{\xi}^{-m_1} (\delta/\delta_0)^{-3 (1-\tbgparam)} |\partial^{\alpha'} f^{\sigma'} (x-z)| dz d\xi,
    \\
     \norm{\iout^\sharp}_{L^2} &\le C \norm{f}_{H^{N+m_1}}(\delta/\delta_0)^{-3 (1-\tbgparam)} \int_{\Omega_\delta^c} \aver{\xi}^{-m_1} d\xi \le
    C \norm{f}_{H^{N+m_1}}\delta^{\tpn(m_1-2)-3 (1-\tbgparam)}
\end{align*}
for any $m_1 \ge 3$.
Using that $|\partial^\alpha_x g_{1,1} (x)| \le C (\iin^\sharp (x) + \iout^\sharp (x))$,
we have now shown that for any $m_1 \ge 3$,
\begin{align}\label{eq:bd_g_11}
    \norm{g_{1,1}}_{H^N} \le C\norm{f}_{H^{N+m_1}}\delta^{\tpn(m_1-2)-3 (1-\tbgparam)}.
\end{align}
Combining \eqref{eq:bd_g00}, \eqref{eq:bd_g10}, \eqref{eq:bd_g_01}, \eqref{eq:bd_g_11} and recalling \eqref{eq:f_decomp}, we have shown that for any $m \ge 3$, 
\begin{align*}
    \norm{\ow (\adres) [f]^\sigma}_{H^N} \le C \norm{f}_{H^{N+m}}\delta^{\tpn(m-2)-3 (1-\tbgparam)}.
\end{align*}
The result then follows from choosing $m$ sufficiently large so that $\tpn(m-2)-3 (1-\tbgparam) \ge q$.
\end{proof}

\section{Classical results on the Schr\"odinger equation}
In this section, we present three classical lemmas that will allow us to control Sobolev norms of solutions to a Schr\"odinger equation. The proofs are included for completeness.
\begin{lemma}\label{lemma:Sobolev_0_appendix}
    Let $\phi_0 \in \cS (\mathbb{R}^d; \mathbb{C}^n)$ and $N \ge 0$. 
    Let $\bH$ be a symmetric operator
    on $L^2 (\mathbb{R}^d; \mathbb{C}^n)$ and $\Lambda$ a bounded operator from $H^N (\mathbb{R}^d; \mathbb{C}^n)$ to $L^2 (\mathbb{R}^d; \mathbb{C}^n)$ satisfying $[\bH, \Lambda] = 0$ and 
    \begin{align}\label{eq:equivalent_norms_appendix}
        C_1 \norm{u}_{H^N} \le \norm{\Lambda u}_{L^2} \le C_2 \norm{u}_{H^N}, \qquad u \in H^N, \quad 0 < \delta \le \delta_0
    \end{align}
    for some $0 < C_1 < C_2$. Then the solution $\phi = \phi (T,X)$ to 
    \begin{align*}
        (D_T + \bH) \phi (T,X) = r(T,X), \qquad \phi (0,X) = \phi_0 (X)
    \end{align*}
    satisfies
    \begin{align}\label{eq:phi_S_bd}
        \norm{\phi (T, \cdot)}_{H^N} \le C_1^{-1} C_2 \left( \norm{\phi_0}_{H^N} + \int_0^T \norm{r (S, \cdot)}_{H^N} dS\right), \qquad T \ge 0.
    \end{align}
\end{lemma}
\begin{proof}
    Multiplying both sides of $(D_T + \bH) \phi = r$ by $\bar\phi$ and integrating with respect to $X$, we obtain
    \begin{align*}
        -i \partial_T \left( \norm{\phi}^2 \right) = -i (\phi, D_T \phi) - (D_T \phi, \phi) = 2i \Im (\phi, D_T \phi) =2i \Im (\phi, r) - 2i \Im (\phi, \bH \phi) = 2i \Im (\phi, r),
    \end{align*}
    where all norms and inner products are understood to be in $L^2$, and the $T$-dependence was suppressed for brevity. Note that the last equality follows from the fact that $\bH$ is symmetric. We conclude that
    $
        \partial_T \norm{\phi (T, \cdot)} \le \norm{r (T, \cdot)},
    $
    which after integrating in $T$ yields
    \begin{align}\label{eq:L2_bd}
        \norm{\phi (T, \cdot)}_{L^2} \le \norm{\phi_0}_{L^2} + \int_0^T \norm{r (S, \cdot)}_{L^2} dS, \qquad T \ge 0.
    \end{align}
    To extend this estimate to the $H^N$ norm, we observe that since $[\bH, \Lambda] = 0$, the function $\phi^\Lambda := \Lambda \phi$ satisfies
    \begin{align*}
        (D_T + \bH) \phi^\Lambda (T,X) = r^\Lambda (T,X), \qquad \phi^\Lambda (0,x) = \phi^\Lambda_0 (x),
    \end{align*}
    where $r^\Lambda := \Lambda r$ and $\phi^\Lambda_0 := \Lambda \phi_0$. Applying \eqref{eq:L2_bd}, we thus obtain
    \begin{align*}
        \norm{\phi^\Lambda (T, \cdot)}_{L^2} \le \norm{\phi^\Lambda_0}_{L^2} + \int_0^T \norm{r^\Lambda (S, \cdot)}_{L^2} dS, \qquad T \ge 0.
    \end{align*}
    The result then follows from \eqref{eq:equivalent_norms_appendix}.
\end{proof}
\begin{lemma}\label{lemma:u_delta}
    Fix $\rate > 0$ and $m \in \mathbb{N}_0$. Suppose $\{ H_\delta : 0 < \delta \le \delta_0\}$ is a family of 
    symmetric operators on $L^2 (\mathbb{R}^d; \mathbb{C}^n)$. When $m\geq1$, assume that for any $0 < \delta \le \delta_0$ and $N \in \{0, 1, \dots, m-1\}$, the operator
    $$[\nabla_x, H_\delta] : H^{N+1}(\mathbb{R}^d; \mathbb{C}^n) \to H^N(\mathbb{R}^d; \mathbb{C}^{d \times n})$$
    is bounded, with
    \begin{align}\label{eq:commutator_bd}
        \norm{[\nabla_x, H_\delta]}_{H^{N+1} \to H^N} \le C \delta, \qquad 0 < \delta \le \delta_0.
    \end{align}
    Let $u_\delta = u_\delta (t,x)$ with $t \ge 0$ and $x \in \mathbb{R}^d$ be a family of functions satisfying $u_\delta (0, \cdot) \equiv 0$ and
    \begin{align*}
        \norm{(D_t + H_\delta) u_\delta (t, \cdot)}_{H^{m}} \le C \delta^{1 + \rate} (1 + (\delta t)^M), \qquad t \ge 0, \quad 0 < \delta \le \delta_0,
    \end{align*}
    for some $M \ge 0$.
    Then for some $\gamma, C > 0$,
    \begin{align}\label{eq:u_delta_bd}
        \norm{u_\delta (t, \cdot)}_{H^m} \le C \delta^{\rate} (e^{\gamma \delta t} - 1), \qquad t \ge 0, \quad 0 < \delta \le \delta_0.
    \end{align}
\end{lemma}
The above result establishes that for any $T \ge 0$, we have the bound
\begin{align}\label{eq:finite_T}
    \norm{u_\delta (t, \cdot)}_{H^m} \le C_T \delta^{1+\rate} t, \qquad 0 \le t \le T/\delta, \quad 0 < \delta \le \delta_0.
\end{align}
As will be evident from the proof below, the estimate \eqref{eq:finite_T} can be extended to all $t \ge 0$ when $m=M=0$. 
\begin{proof}
    Below, all inner products and norms are understood to be in $H^m$.
    We have $(D_t + H_\delta) u_\delta = \delta^{1+\rate} r_\delta$,
    where $\norm{r_\delta (t, \cdot)} \le C(1 + (\delta t)^M)$ uniformly in $t$ and $\delta$. It follows that $ (u_\delta, D_t u_\delta) = \delta^{1+\rate} (u_\delta, r_\delta) - (u_\delta, H_\delta u_\delta)$,
    hence
    \begin{align}\label{eq:dt}
        -i \partial_t \left( \norm{u_\delta (t, \cdot)}^2 \right) = 2 i\Im (u_\delta, D_t u_\delta) = 2 i\Im \left( \delta^{1+\rate} (u_\delta, r_\delta) - (u_\delta, H_\delta u_\delta)\right). 
    \end{align}
    Since $H_\delta$ is symmetric, we know that
    \begin{align}\label{eq:uHu}
        (u_\delta, H_\delta u_\delta)
        &= \sum_{|\alpha| \le m} \int_{\mathbb{R}^d} \overline{\partial^\alpha_x u_\delta (t,x)} H_\delta \partial^\alpha_x u_\delta (t,x) {\rm d} x + 
        \sum_{|\alpha| \le m} \int_{\mathbb{R}^d} \overline{\partial^\alpha_x u_\delta (t,x)} [\partial^\alpha_x, H_\delta] u_\delta (t,x) {\rm d} x,
    \end{align}
    with the first term on the above right-hand side real. By expanding the commutator $[\partial^\alpha_x, H_\delta]$ as a sum of terms involving commutators of $H_\delta$ with only first-order derivatives, our assumption \eqref{eq:commutator_bd} 
    implies that the second term on the right-hand side of \eqref{eq:uHu} is bounded by $C \delta \norm{u_\delta(t, \cdot)}^2$, and thus
    \begin{align*}
        \partial_t \left( \norm{u_\delta (t, \cdot)}^2 \right) \le 2 \delta^{1+\rate} \norm{u_\delta (t, \cdot)} \norm{r_\delta(t, \cdot)} + C \delta \norm{u_\delta(t, \cdot)}^2
    \end{align*}
    by \eqref{eq:dt}.
    Evaluating the derivative on the left-hand side and setting $\gamma := C/2$, we obtain
    \begin{align*}
    e^{\gamma \delta t} \partial_t (e^{-\gamma \delta t}\norm{u_\delta (t, \cdot)}) = 
        (\partial_t - \gamma \delta) \norm{u_\delta (t, \cdot)} \le \delta^{1+\rate}\norm{r_\delta(t, \cdot)}.
    \end{align*}
    Multiplying both sides by $e^{-\gamma \delta t}$ and integrating from $0$ to $t$, the 
    bound on $\norm{r_\delta(t, \cdot)}$ and our assumption that $u_\delta (0, \cdot) \equiv 0$ then imply that
    \begin{align*}
        e^{-\gamma \delta t} \norm{u_\delta (t, \cdot)} \le C \delta^{1+\rate} (1+(\delta t)^M)\int_0^t  e^{-\gamma \delta s}ds = C \gamma^{-1} \delta^{\rate} (1 - e^{-\gamma \delta t})(1+(\delta t)^M).
    \end{align*}
    The result follows from multiplying both sides by $e^{\gamma \delta t}$, then taking any larger choice of $\gamma$.
\end{proof}
By strengthening the assumption \eqref{eq:commutator_bd} on the commutator $[\nabla_x, H_\delta]$, the exponential growth in $\delta t$ of the right-hand side of \eqref{eq:u_delta_bd} can be replaced by polynomial growth.
\begin{lemma}\label{lemma:u_delta_reg}
    Fix $\rate > 0$ and $m \in \mathbb{N}_0$. Suppose $\{ H_\delta : 0 < \delta \le \delta_0\}$ is a family of 
    symmetric operators on $L^2 (\mathbb{R}^d; \mathbb{C}^n)$ such that for any $0 < \delta \le \delta_0$ and $N \in \{0, 1, \dots, m-1\}$, the operator
    $$[\nabla_x, H_\delta] : H^{N}(\mathbb{R}^d; \mathbb{C}^n) \to H^N(\mathbb{R}^d; \mathbb{C}^{d \times n})$$
    is bounded, with
    \begin{align}\label{eq:commutator_bd_reg}
        \norm{[\nabla_x, H_\delta]}_{H^{N} \to H^N} \le C \delta, \qquad 0 < \delta \le \delta_0.
    \end{align}
    Let $u_\delta = u_\delta (t,x)$ with $t \ge 0$ and $x \in \mathbb{R}^d$ be a family of functions satisfying $u_\delta (0, \cdot) \equiv 0$ and
    \begin{align*}
        \norm{(D_t + H_\delta) u_\delta (t, \cdot)}_{H^{m}} \le C \delta^{1 + \rate} (1 + (\delta t)^M), \qquad t \ge 0, \quad 0 < \delta \le \delta_0,
    \end{align*}
    for some $M \ge 0$. Then for some $C > 0$,
    \begin{align*}
        \norm{u_\delta (t, \cdot)}_{H^m} \le C \delta^{1+\rate} t (1 + (\delta t)^{m+M}), \qquad t \ge 0, \quad 0 < \delta \le \delta_0.
    \end{align*}
\end{lemma}
\begin{proof}
    Following the proof of Lemma \ref{lemma:u_delta} but invoking the additional regularity \eqref{eq:commutator_bd_reg} of $[\nabla_x, H_\delta]$ in \eqref{eq:uHu}, we obtain
    \begin{align*}
        \partial_t \left( \norm{u_\delta (t, \cdot)}_{H^N}^2 \right) \le 2 \delta^{1+\rate} \norm{u_\delta (t, \cdot)}_{H^N} \norm{r_\delta(t, \cdot)}_{H^N} + C \delta \norm{u_\delta(t, \cdot)}_{H^N}\norm{u_\delta(t, \cdot)}_{H^{N-1}}
    \end{align*}
    for all $N \in \{0, 1, \dots, m\}$, where $\delta^{1 + \rate} r_\delta = (D_t + H_\delta) u_\delta$. 
    Therefore,
    \begin{align*}
        \partial_t \norm{u_\delta (t, \cdot)}_{H^N} \le \delta^{1+\rate} \norm{r_\delta(t, \cdot)}_{H^N} + C \delta \norm{u_\delta(t, \cdot)}_{H^{N-1}},
    \end{align*}
    with the second term on the right-hand side understood to vanish if $N=0$. After integrating both sides with respect to $t$, the result follows from induction in $N$.
\end{proof}

\begin{remark}\label{remark:twisted_grad}
    Let $\Theta := (\theta_1, \dots, \theta_n) \in \mathbb{T}^n$. By unitarity, it is clear that the conclusions of Lemmas \ref{lemma:u_delta} and \ref{lemma:u_delta_reg} would still hold if 
    $\nabla_x$ in \eqref{eq:commutator_bd} and \eqref{eq:commutator_bd_reg} were replaced by the operator $\nabla^\Theta_x$ defined by
    \begin{align*}
        \nabla_x^\Theta [\psi]^\sigma (x) := \rot_{\theta_\sigma}^\top \nabla_x \psi^\sigma (x), \qquad x \in \mathbb{R}^d, \quad \sigma \in \{1, \dots, n\}, 
    \end{align*}
    with $\rot^\top_\theta$ clockwise rotation by the angle $\theta$.
\end{remark}
\section{Results on approximations to graphene TB models}
\subsection{Proof of Lemma \ref{lem:ellipb0p}.}\label{sec:Ab0p}
We will prove that for any $p \ge 1$ and with $a$ as in \eqref{eq:Hdelta}, the symbol
\begin{align*}
    \delta b_{0p} (X,\zeta;\delta) = \sum_{j=1}^p \frac{\delta^j}{j!} \sum_{i_1, \cdots, i_j = 1}^{3} \bar{\zeta}_{i_1} \dots \bar{\zeta}_{i_j} \partial_{i_1, \dots, i_j} a(X,\bar{K})
\end{align*}
satisfies \eqref{eq:ellipticty_higher_order}.
Using the notation $z := \xi_1 + i\xi_2$ and $\bar z := \xi_1 - i\xi_2$, we find that
\begin{align*}
    \partial^m_z \partial^n_{\bar z} (e^{i \xi \cdot \la_1}) &= \left(\frac{i v}{2 \sqrt{3}} \right)^{m+n} e^{i \xi \cdot \la_1}, \qquad 
    \partial^m_z \partial^n_{\bar z} (e^{i \xi \cdot \la_2}) = \left(\frac{i v}{2 \sqrt{3}} \right)^{m+n} e^{i2\pi (n-m)/3} e^{i \xi \cdot \la_2}, \\
    \partial^m_z \partial^n_{\bar z} (e^{i \xi \cdot \la_3}) &= \left(\frac{i v}{2 \sqrt{3}} \right)^{m+n} e^{-i2\pi (n-m)/3} e^{i \xi \cdot \la_3}.
\end{align*}
Therefore, defining $g (\xi) := \sum_{j=1}^3 e^{i\xi \cdot \la_j}$, it follows that
\begin{align*}
    \partial^m_z \partial^n_{\bar z} g (K) = \left(\frac{i v}{2 \sqrt{3}} \right)^{m+n} \left( 1 + e^{i2\pi (n-m-1)/3} + e^{-i2\pi (n-m-1)/3}\right) = \begin{cases}
        3\left(\frac{i v}{2 \sqrt{3}} \right)^{m+n}, & n-m-1 \in 3 \mathbb{Z},\\
        0, & \text{else}.
    \end{cases},
\end{align*}
For the Taylor expansion of $a^1$ in \eqref{eq:a0}, we see that
\begin{align*}
    \partial^m_z \partial^n_{\bar z} (e^{i \xi \cdot \lb_1}) &= \left(\frac{i v}{2} \right)^{m+n} e^{i \pi (n-m)/6} e^{i \xi \cdot \lb_1}, \qquad 
    \partial^m_z \partial^n_{\bar z} (e^{i \xi \cdot \lb_2}) = \left(\frac{-v}{2} \right)^{m+n} (-1)^n e^{i \xi \cdot \lb_2}, \\
    \partial^m_z \partial^n_{\bar z} (e^{i \xi \cdot \lb_3}) &= \left(\frac{i v}{2} \right)^{m+n} e^{i5\pi (n-m)/6} e^{i \xi \cdot \lb_3},
\end{align*}
from which one can verify that with $h (\xi) := \sum_{j=1}^3 \sin (\lb_j \cdot \xi)$,
\begin{align*}
    \partial^m_z \partial^n_{\bar z} h(K) = \left(\frac{i v}{2} \right)^{m+n} \begin{cases}
        \frac{3 \sqrt{3}}{2} (-1)^{j+1}, & n-m = 6j, \quad j \in \mathbb{Z},\\
        \frac{3}{2} (-1)^j, & n-m = 6j-3, \quad j \in \mathbb{Z},\\
        0, & \text{else}.
    \end{cases}
\end{align*}
The entries of the $2\times2$ symbol $b_{0p}$
are thus given by
\begin{align*}
    \delta b_{0p}^{jj} (X, \zeta; \delta) &= (-1)^{j-1} \delta \left(M (X) + t_2 (X) \sum_{N=0}^{p-1} \frac{\delta^N}{N!} \sum_{m=0}^N \binom{N}{m}(\zeta_1 + i \zeta_2)^m (\zeta_1 - i \zeta_2)^{N-m} \partial^m_z \partial^{N-m}_{\bar z} h(K)\right)\\
    &=(-1)^{j-1} \delta \Bigg\{M (X) 
    \\
    + \frac{3\sqrt{3}}{2}t_2 (X) &\sum_{N'=0}^{\lfloor (p-1)/2 \rfloor} \frac{\delta^{2N'}}{(2N')!} \left( \frac{iv}{2}\right)^{2N'} \sum_{j=-\lfloor N'/3 \rfloor}^{\lfloor N'/3 \rfloor} \binom{2N'}{N' - 3j}(\zeta_1 + i \zeta_2)^{N'-3j} (\zeta_1 - i \zeta_2)^{N'+3j} (-1)^{j+1} \\
    + \frac{3}{2} t_2 (X) &\sum_{N''=0}^{\lfloor p/2 \rfloor - 1}
    \frac{\delta^{2N''+1}}{(2N''+1)!} \left( \frac{iv}{2}\right)^{2N''+1} \\
    &\qquad \sum_{j=-\lfloor (N''+2)/3 \rfloor+1}^{\lfloor (N''+2)/3 \rfloor} \binom{2N''+1}{N''+2 - 3j}(\zeta_1 + i \zeta_2)^{N''+2-3j} (\zeta_1 - i \zeta_2)^{N''-1+3j} (-1)^{j} \Bigg\}
\end{align*}
and with $I_N := \{m \in \{0,1, \dots, N\} : N+m-1 \in 3 \mathbb{Z}\}$,
\begin{align*}
    \delta b_{0p}^{12} (X, \zeta; \delta) &= t_1 (X) \sum_{N=1}^p \frac{\delta^N}{N!} \sum_{m = 0}^N \binom{N}{m} (\zeta_1 + i \zeta_2)^m (\zeta_1 - i \zeta_2)^{N-m} \partial^m_z \partial^{N-m}_{\bar z} g(K)\\
    &= 3 t_1 (X) \sum_{N=1}^p \frac{\delta^N}{N!} \left( \frac{iv}{2\sqrt{3}}\right)^N \sum_{m \in I_N} \binom{N}{m} (\zeta_1 + i \zeta_2)^m (\zeta_1 - i \zeta_2)^{N-m}.
\end{align*}
To prove that the symbol $b_{0p}$ is elliptic in the sense of \eqref{eq:ellipticty_higher_order}, it suffices to show that the leading-order term
\begin{align*}
    f (z) := \frac{1}{2^p}\sum_{m \in I_p} \binom{p}{m} z^m \bar z^{p-m}
\end{align*}
has no zeros $z$ on the unit circle (recall that $t_1$ is bounded away from zero). 
Using the identity
\begin{align*}
    f (z) = \frac{1}{3} \frac{1}{2^p} \left( (z + \bar z)^p + e^{i2\pi (p-1)/3} (e^{i2\pi/3} z + \bar z)^p + e^{-i 2\pi (p-1)/3} (e^{-i2\pi/3} z + \bar z)^p \right),
\end{align*}
it follows that
\begin{align*}
    f (e^{i \theta}) &= \frac{1}{3} \Big( \cos^p (\theta) + \frac{1}{2}(-1)^{p-1} \left(\cos^p \left(\theta + \frac{\pi}{3}\right) + \cos^p \left(\theta - \frac{\pi}{3}\right) \right)\\
    &\hspace{4cm} + i \frac{\sqrt{3}}{2} (-1)^{p-1} \left(\cos^p \left(\theta + \frac{\pi}{3}\right) - \cos^p \left(\theta - \frac{\pi}{3}\right) \right)\Big).
\end{align*}
If $p$ is odd, then the imaginary part vanishes if and only if $\theta \in \pi \mathbb{Z}$, in which case
\begin{align*}
    3|\Re f (e^{i\theta})| \ge |\cos^p (\theta)| - \frac{1}{2} \left| \cos^p \left(\theta + \frac{\pi}{3}\right) + \cos^p \left(\theta - \frac{\pi}{3}\right)\right| = 1 - 2^{-p} \ge \frac{1}{2},
\end{align*}
and so $f$ has no zeros on the unit circle. If $p$ is even, then $\Im f$ vanishes if and only if $\theta \in \pi \mathbb{Z}/2$. For $\theta \in \pi \mathbb{Z}$, we have $|\Re f| \ge 1/6$ as before. If $\theta \in \pi \mathbb{Z} + \frac{\pi}{2}$ and $p$ is even, then
\begin{align*}
    \Re f (e^{i\theta}) = -\frac{1}{6} \left(\cos^p \left(\theta + \frac{\pi}{3}\right) + \cos^p \left(\theta - \frac{\pi}{3}\right) \right) = -\frac{1}{3} \left( \frac{\sqrt{3}}{2}\right)^p.
\end{align*}
Thus we have verified that $f$ has no zeros on the unit circle, meaning that the ellipticity condition \eqref{eq:ellipticty_higher_order} is satisfied at any order $p \ge 1$.

\subsection{Proof of Lemma \ref{lem:varK}.}\label{sec:varK}
We will verify the self-adjointness requirement and commutator estimate from Assumption \ref{assumption:commutator2}. Since $a_\delta$ is Hermitian-valued, $H_\delta = \ow a_\delta$ is self-adjoint if $H_\delta$ is bounded. Thus, recalling that each entry of $a_\delta (x, \xi)$ is a linear combination of $t_1 (\delta x) e^{i \xi \cdot \la_j (\delta x)}$ and $\frac{1}{3}M(\delta x) + t_2 (\delta x) \sin (\lb_j (\delta x) \cdot \xi),$ it suffices to show that for each $\sym \in \{M, t_1, t_2\}$ and $\la \in \{(0,0)\} \cup \{\la_j\}_{j=1}^3 \cup \{\lb_j\}_{j=1}^3$, the operator 
\begin{align}\label{eq:P_delta_def}
    P_\delta := \ow \tilde a_\delta, \qquad \tilde a_\delta (x, \xi) := e^{i \xi \cdot \la (\delta x)} \sym (\delta x)
\end{align}
is bounded on $L^2 (\mathbb{R}^2; \mathbb{C}^2)$, with
\begin{align}\label{eq:norm_commutator_P}
    \norm{[\nabla_x, P_\delta]}_{H^{N+1} \to H^N} \le C_N \delta, \qquad 0 < \delta \le \delta_0, \quad N \in \mathbb{N}_0,
\end{align}
for some $\delta_0 > 0$.
The definition of $P_\delta$ states that for any $f \in L^2 (\mathbb{R}^2; \mathbb{C}^2)$,
\begin{align*}
    P_\delta f (x) = \frac{1}{(2\pi)^2}\int_{\mathbb{R}^4} e^{i \xi \cdot (x-y + \la (\frac{\delta}{2} (x+y)))} \sym (\frac{\delta}{2} (x+y)) f(y) dy d\xi.
\end{align*}
Since $\la \in C^\infty_b (\mathbb{R}^2)$, the contraction mapping theorem implies that (provided $\delta > 0$ is sufficiently small) for every $x \in \mathbb{R}^2$, there exists a unique $y \in \mathbb{R}^2$ such that $x-y + \la (\frac{\delta}{2} (x+y)) = 0$. Denoting this unique point by $y(x)$ in a slight abuse of notation, it follows that
\begin{align}\label{eq:Pf}
    P_\delta f(x) = \sym (\frac{\delta}{2}(x+y(x))) f(y(x)).
\end{align}
By the regularity of $\la$, we know that $y \in C^\infty (\mathbb{R}^2; \mathbb{R}^2)$ is a bijection. Indeed, if $y(x_1) = y(x_2)$ for some $x_1, x_2 \in \mathbb{R}^2$, then $x_1 + \la (\frac{\delta}{2} (x_1 + y(x_1))) = x_2 + \la (\frac{\delta}{2} (x_2 + y(x_1)))$, which implies
\begin{align*}
    |x_1- x_2| \le \left|\la (\frac{\delta}{2} (x_2 + y(x_1))) - \la (\frac{\delta}{2} (x_1 + y(x_1)))\right| \le \frac{1}{2} \norm{\nabla \la}_{L^\infty} \delta |x_1 - x_2|, \qquad \delta > 0,
\end{align*}
and thus $x_1 = x_2$ (provided $\delta < 2 \norm{\nabla \la}_{L^\infty}^{-1}$). The surjectivity of $y$ follows from the same argument that established that $y$ is well-defined. Finally, we observe that the Jacobian $\nabla y \in \mathbb{R}^{2\times 2}$ satisfies
\begin{align*}
    \nabla y = I + \frac{\delta}{2} (I + \nabla y) \nabla \la (\frac{\delta}{2} (x + y)).
\end{align*}
It follows that for any $0 < \delta < 2 \norm{\nabla \la}_{L^\infty}^{-1}$,
\begin{align}\label{eq:J}
    \nabla y = \left( I + \frac{\delta}{2} \nabla \la (\frac{\delta}{2} (x + y)) \right) \left(I - \frac{\delta}{2} \nabla \la (\frac{\delta}{2} (x + y))\right)^{-1},
\end{align}
and thus $|\det \nabla y (x)| \ge c >0$ is bounded away from zero. We conclude by \eqref{eq:Pf} that
\begin{align}\label{eq:P_bdd}
    \norm{P_\delta f}_{L^2}^2 = \int_{\mathbb{R}^2} \left | \sym (\frac{\delta}{2} (y^{-1} (x) + x)) f(x) \right| ^2\frac{dx}{|\det \nabla y(x)|} \le C \norm{f}_{L^2}^2,\qquad 0 < \delta \le \norm{\nabla \la}_{L^\infty}^{-1},
\end{align}
where we recall the uniform boundedness of $\sym$ to establish the inequality.

Next, we use \eqref{eq:Pf} to write
\begin{align}\label{eq:commutator_P}
    [\partial_{x_j}, P_\delta] f(x) = \partial_{x_j} \left ( \sym  (\frac{\delta}{2} (x + y(x)))\right) f(y(x)) + \sym (\frac{\delta}{2} (x + y(x))) \left( \sum_{i=1}^2 \frac{\partial y_i}{\partial_{x_j}} (x)\frac{\partial f}{\partial_{x_i}}(y(x)) - \frac{\partial f}{\partial x_j} (y(x))\right).
\end{align}
Recall \eqref{eq:J}, which implies that $\norm{\nabla y}_{L^\infty} \le C$ uniformly in $0 < \delta \le \norm{\nabla \la}_{L^\infty}^{-1}$.
Expanding the first term as $$\partial_{x_j} \left ( \sym  (\frac{\delta}{2} (x + y(x)))\right) = \frac{\delta}{2} \left(\partial_{x_j} \sym (\frac{\delta}{2} (x + y(x))) + \sum_{i=1}^2 \partial_{x_j} y_i (x) \partial_{x_i} \sym (\frac{\delta}{2} (x + y(x)))\right),$$
the regularity of $y$ and $\sym$, together with the argument that established the boundedness of 
$P_\delta$ on $L^2 (\mathbb{R}^2; \mathbb{C}^2)$, implies that
\begin{align}\label{eq:first_term_bd}
    \int_{\mathbb{R}^2}\left| \partial_{x_j} \left ( \sym  (\frac{\delta}{2} (x + y(x)))\right) f(y(x)) \right|^2 dx \le C \delta^2 \norm{f}_{L^2}^2, \qquad 0 < \delta \le \norm{\nabla \la}_{L^\infty}^{-1}.
\end{align}
We write the second term on the right-hand side of \eqref{eq:commutator_P} as
\begin{align*}
    \Delta(x):= \sym (\frac{\delta}{2} (x + y(x))) \sum_{i=1}^2 \left(\partial_{x_j} y_i (x) - I_{ij} \right) \partial_{x_i} f (y(x)).
\end{align*}
By \eqref{eq:J} and the boundedness of $\nabla \la$, we know that $\norm{\nabla y - I}_{L^\infty} \le C \delta$ uniformly in $0 < \delta \le \norm{\nabla \la}_{L^\infty}^{-1}$, hence
$|\Delta (x)| \le C \delta |(\nabla f)(y(x))|$ for $0 < \delta \le \norm{\nabla \la}_{L^\infty}^{-1}$
and so 
\begin{align}\label{eq:Delta_bd}
    \norm{\Delta}_{L^2}^2 \le C \delta^2 \int_{\mathbb{R}^2} |\nabla f (x)|^2 \frac{dx}{|\det \nabla y(x)|} \le C \delta^2 \norm{\nabla f}_{L^2}^2, \qquad 0 < \delta \le \norm{\nabla \la}_{L^\infty}^{-1}.
\end{align}
Combining \eqref{eq:commutator_P}-\eqref{eq:first_term_bd}-\eqref{eq:Delta_bd}, we have shown that
\begin{align}\label{eq:K_X_bd}
    \norm{[\nabla_x, P_\delta] f}_{L^2} \le C \delta \norm{f}_{H^1}, \qquad 0 < \delta \le \norm{\nabla \la}_{L^\infty}^{-1}.
\end{align}
By the regularity of $\sym, \la, y$, this estimate easily extends to all $N$ in \eqref{eq:norm_commutator_P}.
Recalling \eqref{eq:P_bdd}, we have thus verified that $H_\delta = \ow a_\delta$ satisfies Assumption \ref{assumption:commutator2} for any $m \in \mathbb{N}_0$, as desired.

\end{document}